\def\OPTIONArxiv{1}
\def\OPTIONAppendix{1}

\def\OPTIONTalk{0}
\def\OPTIONConf{1}
\def\OPTIONArxiv{1}
\def\OPTIONAppendix{1}
\ifnum\OPTIONConf=1%
  \ifnum\OPTIONArxiv=0%
     \documentclass[nonatbib]{sigplanconf}
  \else
     \documentclass[preprint,nonatbib]{sigplanconf}
  \fi
\else
\documentclass[10pt,leqno]{article}
\fi
\usepackage[overload]{textcase}
\ifnum\OPTIONConf=1%
\else%
   \usepackage[headings]{fullpage}
\fi
\usepackage{etex}
\usepackage{color}
\usepackage[breaklinks]{hyperref} %
\usepackage{breakurl}
\usepackage{pstricks,pst-node}
\usepackage[OT2,T1]{fontenc}
\usepackage{balance}

\usepackage{framed}
\ifnum\OPTIONConf=0%
   \usepackage{lmodern}
   \usepackage{titlesec}
        \titleformat{\section}
          {\normalfont\bfseries\sffamily\fontsize{14pt}{15pt}\selectfont}
          {\SectionHiliteStrong{\thesection}}{0.8em}{}

        \titleformat{\subsection}
          {\normalfont\bfseries\sffamily\fontsize{12pt}{13pt}\selectfont}
          {\SectionHilite{\thesubsection}}{0.6em}{}

        \titleformat{\subsubsection}
          {\normalfont\bfseries\sffamily\fontsize{10pt}{11pt}\selectfont}
          {\SectionHilite{\thesubsubsection}}{0.25em}{}

   \usepackage{abstract}
         \renewcommand{\abstractnamefont}%
           {\normalfont\bfseries\sffamily\fontsize{10pt}{11pt}\selectfont}
\fi

\usepackage{listings}
\usepackage[authoryear]{natbib}
\bibpunct{(}{)}{;}{a}{}{,}

\usepackage{pgf,tikz}
\usetikzlibrary{shapes,arrows,matrix,fit,backgrounds,calc,decorations,decorations.pathmorphing}

\usepackage{xspace}
\ifnum\OPTIONConf=1%
\else
    \usepackage{charter}
\fi
\usepackage{euler}
\renewcommand{\ttdefault}{cmtt}   %

\usepackage{amsmath,amssymb,amsthm}

\usepackage{multirow}
\usepackage{stmaryrd}
\usepackage{pifont}
\usepackage{enumerate}
\usepackage{comment}
\usepackage{mathpartir} \usepackage{proof}

\def\MathparLineskip{\lineskiplimit=0.7em\lineskip=0.7em plus 0.2em}  %

\usepackage{framed}
\usepackage{rotating}
\usepackage{relsize}

\let\MathRightArrow\Rightarrow %
\usepackage{marvosym} %
\def\Rightarrow{\MathRightArrow}

\usepackage{thmtools,thm-restate}

\declaretheoremstyle[
  bodyfont=\rm
]{mytheoremstyle}
\declaretheorem[style=mytheoremstyle]{theorem}

\declaretheorem[style=mytheoremstyle, sibling=theorem]{lemma}

\declaretheorem[style=mytheoremstyle, sibling=lemma]{corollary}

\declaretheorem{example}


\declaretheorem[style=mytheoremstyle]{definition}

\makeatletter
\ifnum\OPTIONConf=1
    \DeclareRobustCommand{\mksubsection}[1]{#1}
\else
    \DeclareRobustCommand{\mksubsection}[1]{\@element#1\@nil}
    \def\@element#1#2\@nil{%
      #1%
      \if\relax#2\relax\else\MakeLowercase{#2}\fi}
\fi
\makeatother

\newcommand{\addmytocentry}[2]{%
   \addtocontents{toc}{\protect\contentsline {subsubsection}{\protect\numberline {#1}#2}{\thepage}{xx}}%
}
\newcommand{\MyThmRestateHook}[3]{%
  \relax
}

\newcommand{\lxbel}[1]{\label{#1}}

\makeatletter
\newcommand{\XLabel}[1]{%
  \ifcsname IDEMPOTFLAG#1\endcsname%
      \lxbel{PROOF#1}%
      \LoudLabel{#1}%
      \global\edef\CurrentProofName{#1}%
      \addmytocentry{\ref{PROOF#1}}{\hyperref[PROOF#1]{Proof of \thmt@thmname~(\thmt@optarg)}}%
  \else%
       \@ifundefined{r@PROOF#1}{%
          \textcolor{red}{\bf NO PROOF}
       }{%
          \hyperref[PROOF#1]{\textcolor{blue}{Go to proof}}%
       }%
       \Label{#1}%
       \addmytocentry{\ref{#1}}{\hyperref[#1]{\thmt@thmname~(\thmt@optarg)}%
      }%
      \expandafter\gdef\csname IDEMPOTFLAG#1\endcsname{d}%
   \fi}
\makeatother

\usepackage{srcltx}

\ifnum\OPTIONConf=1%
\def\OPTIONLoudLabels{0}
\else
\def\OPTIONLoudLabels{1}
\fi

\newcommand{\newrulecommand}[2]{%
    \expandafter\newcommand\csname NoLink#1\endcsname{#2}%
    \expandafter\newcommand\csname #1\endcsname{%
        \ifcsname ZZRULE#1\endcsname%
            \hyperlink{rule:#1}{\expandafter\csname NoLink#1\endcsname}\xspace%
        \else%
             \expandafter\gdef\csname ZZRULE#1\endcsname{}%
             \hypertarget{rule:#1}{\expandafter\csname NoLink#1\endcsname}\xspace%
        \fi%
    }%
}

\newcommand{\newrulecommandONE}[2]{%
    \expandafter\newcommand\csname NoLink#1\endcsname[1]{#2}%
    \expandafter\newcommand\csname #1\endcsname[1]{%
        \ifcsname ZZRULE#1\endcsname%
            \hyperlink{rule:#1}{\expandafter\csname NoLink#1\endcsname{##1}}%
        \else%
             \hypertarget{rule:#1}{\expandafter\csname NoLink#1\endcsname{##1}}%
             \expandafter\gdef\csname ZZRULE#1\endcsname{}%
        \fi%
    }
}

\newcommand{\N}{\ensuremath{\mathsf N}}
\newcommand{\V}{\ensuremath{\mathsf V}}

\newcommand{\ditto}{\ensuremath{''}}

\newcommand{\textvtt}[1]{{\normalfont\fontfamily{cmvtt}\selectfont {#1}}}
\newcommand{\textfn}[1]{\textit{#1}}
\newcommand{\datacon}[1]{\ensuremath{\mathsf{#1}}}
\newcommand{\tyname}[1]{\ensuremath{\mathsf{#1}}}

\newcommand{\keyword}[1]{\text{\textvtt{#1}}}
\newcommand{\textkw}[1]{\keyword{#1}}

\newcommand{\xlam}{\lambda}

\newcommand{\Lam}[1]{\xlam{#1}.\,}

\newcommand{\xThunk}{\keyword{thunk}}
\newcommand{\Thunk}{\xThunk\,}
\newcommand{\thunkty}{\U}
\newcommand{\xForce}{\keyword{force}}
\renewcommand{\Force}{\xForce\,}

\newcommand{\Anno}[2]{\textt{(}#1\textt{:}#2\textt{)}}
\newcommand{\Pair}[2]{\textt{(}#1\textt{,}\;#2\textt{)}}
\newcommand{\xFix}{\ensuremath{\keyword{fix}}}
\newcommand{\Fix}[1]{\xFix~#1.\:}
\newcommand{\arr}{\rightarrow}
\newcommand{\entailssym}{\vdash}         %
\newcommand{\entails}{\,\entailssym\,}         %
\newcommand{\entailssub}[2]{\;\mathcolor{#1}{\entailssym_{#2}}\;}         %

\newcommand{\subtype}{\mathrel{\leq}}       %

\newcommand{\sectty}{\mathrel{\mathcolor{dBlue}{\land}}}          %
\newcommand{\unty}{\mathrel{\mathcolor{dDkRed}{\lor}}}          %

\newcommand{\ExisSym}{\exists}

\newcommand{\step}{\mapsto}
\newcommand{\steps}{\step^*}
\newcommand{\stepR}{\step_{\normalfont\textsf{R}}}

\newcommand{\dstep}{\rightsquigarrow}

\newcommand{\srcstep}{\rightsquigarrow}
\newcommand{\srcsteps}{\dstep^*}

\newcommand{\srcredV}{\srcstep_{\normalfont\textsf{R\V}}}
\newcommand{\srcredN}{\srcstep_{\normalfont\textsf{R\N}}}

\newcommand{\xCase}{\ensuremath{\keyword{case}}}
\newcommand{\Case}[1]{\xCase~#1~\keyword{of}~}

\newcommand{\Casesumy}[2]{\matchor \Match{\Inr{#1}}{#2}}
\newcommand{\LongCasesum}[5]{\Case{#1}{\Match{\Inl{#2}}{#3} \Casesumy{#4}{#5}}}

\newcommand{\Casesumstart}[3]{\keyword{case}(%
               {#1}%
               \textt,\,%
               {#2}%
               .%
               {#3}%
               \textt,\,%
}
\newcommand{\Casesumend}[2]{%
               {#1}%
               .%
               {#2}
               )%
}
\newcommand{\Casesum}[5]{%
               \Casesumstart{#1}{#2}{#3}\Casesumend{#4}{#5}}

\newcommand{\matchor}{\ensuremath{\normalfont\,\textt{|}\hspace{-5.35pt}\textt{|}\,}}
\newcommand{\Match}[1]{{#1}\Rightarrow}
\newcommand{\xLet}[2]{{\xxLet\;#1\;{\normalfont\textt{=}}\;#2\;\keyword{in}\;}}

\newcommand{\Int}{\tyname{int}}

\newcommand{\Nil}{\datacon{Nil}}
\newcommand{\Cons}{\datacon{Cons}}

\newcommand{\Deref}[1]{\ensuremath{{\textt{!}}{#1}}}
\newcommand{\Gets}{\ensuremath{\mathrel{\textt{:=}}}}

\newcommand{\arrayenvc}[1]{\renewcommand{\arraystretch}{1}\ensuremath{\begin{array}[c]{@{}c@{}}#1\end{array}}}
\newcommand{\arrayenvcl}[1]{\renewcommand{\arraystretch}{1}\ensuremath{\begin{array}[c]{@{}l@{}}#1\end{array}}}

\newcommand{\arrayenvl}[1]{\renewcommand{\arraystretch}{1}\ensuremath{\begin{array}[t]{@{}l@{}}#1\end{array}}}

\newcommand{\arrayenvbl}[1]{\renewcommand{\arraystretch}{1}\ensuremath{\begin{array}[b]{@{}l@{}}#1\end{array}}}

\newcommand{\tabularenvc}[1]{\renewcommand{\arraystretch}{1} \begin{tabular}[c]{@{}c@{}}#1\end{tabular}}
\newcommand{\tabularenvcl}[1]{\renewcommand{\arraystretch}{1} \begin{tabular}[c]{@{}l@{}}#1\end{tabular}}

\newcommand{\tabularenvl}[1]{\renewcommand{\arraystretch}{1} \begin{tabular}[t]{@{}l@{}}#1\end{tabular}}

\newcommand{\unitty}{\textbf{1}}

\newcommand{\unit}{\textt{()}}

\newcommand{\synabovechk}{\text{\raisebox{0.7ex}{$\chk$}\backupwidth{\chk}\hspace{0.2ex}\raisebox{-0.7ex}{$\syn$}}}
\newcommand{\chkcolor}{dBlue}
\newcommand{\syncolor}{dDkGreen}
\newcommand{\appcolor}{dDkRed}
\newcommand{\chk}{\mathrel{\mathcolor{\chkcolor}{\Leftarrow}}}
\newcommand{\uncoloredsyn}{{\Rightarrow}}
\newcommand{\syn}{\mathrel{\mathcolor{\syncolor}{\uncoloredsyn}}}
\newcommand{\xProj}[1]{\keyword{proj$_{#1}$}}
\newcommand{\Proj}[1]{\xProj{#1}\,}

\newcommand{\Inj}[2]{\keyword{inj$_{#1}$}\,#2}

\newcommand{\Inl}[1]{\Inj{1}{#1}}
\newcommand{\Inr}[1]{\Inj{2}{#1}}

\newcommand{\troll}{\keyword{roll}\;}
\newcommand{\tunroll}{\keyword{unroll}\;}

\newcommand{\D}{\mathcal{D}}
\newcommand{\Dee}{\D}

\newcounter{codeLineCntr}

\newcommand{\out}[1]{}

\newcommand{\floor}[1]{\lfloor{#1}\rfloor}

\renewcommand{\phi}{\varphi}

\newcommand{\Appendixref}[1]{Appendix \ref{#1}}
\newcommand{\Figureref}[1]{Figure \ref{#1}}

\newcommand{\Sectionref}[1]{Section \ref{#1}}

\newcommand{\Theoremref}[1]{Theorem \ref{#1}}

\newcommand{\Lemmaref}[1]{Lemma \ref{#1}}

\newcommand{\Definitionref}[1]{Definition \ref{#1}}

\newcommand{\val}{\textsl{\,value}}

\newcommand{\AND}{\textrm{~and~}}

\newcommand{\hole}{\ensuremath{[\,]}}

\newcommand{\Refexp}[1]{\keyword{ref}~{#1}}
\newcommand{\Refty}[1]{\textsf{ref}~{#1}}

\usepackage{etal}

\definecolor{lred}{rgb}{1.0, 0.3, 0.3}
\newrgbcolor{lred}{1.0 0.3 0.3}

\newlength\zzskipwidthlen

\newcommand{\backupwidth}[1]{\settowidth{\zzskipwidthlen}{\ensuremath{#1}}\hspace{-\zzskipwidthlen}}

\newenvironment{bnfarray}{\!\!\begin{tabular}[t]{c}\begin{array}[t]{@{}l@{~~}l@{~~}c@{~~}ll@{}}}{\end{array}\end{tabular}}
\newcommand{\bnfas}{\mathrel{::=}}
\newcommand{\bnfalt}{\mathrel{\mid}}
\newcommand{\xbnfaltBRK}{\!\!\!\bnfalt}
\newcommand{\bnfaltBRK}{ \\ & & & \xbnfaltBRK}
\gdef\bnfaltBR[#1]{ \\[#1] & & & \xbnfaltBRK}

\newcommand{\textt}[1]{\text{\normalfont\tt #1}}

\newdimen{\zzzpbox}

\newdimen\zzfontsz
\newcommand{\fontsz}[2]{\zzfontsz=#1%
{\fontsize{\zzfontsz}{1.2\zzfontsz}\selectfont{#2}}}

\newcommand{\mathsz}[2]{\text{\fontsz{#1}{$#2$}}}

\newlength{\zzsplatboldwidth}
\newcommand{\xsplatbold}[2]{\settowidth{\zzsplatboldwidth}{{#2}}{#2}\addtolength{\zzsplatboldwidth}{-#1}\hspace{-\zzsplatboldwidth}\raisebox{#1}{{#2}}}
\newcommand{\splatbold}[1]{\xsplatbold{-0.04mm}{#1}}

\makeatletter
\def\url@MGstyle{%
\def\UrlFont{\tiny\ttfamily}%
\def\do@url@hyp{\do\-}%
\Url@do
}
\makeatother
\newcommand{\marginPseudoURL}{\begingroup \urlstyle{MG}\Url}

\ifnum\OPTIONTalk=1
  
\else
  \newcommand{\marginnote}[1]{\marginpar{\raggedright\scriptsize{#1}}}
\fi

\newcommand{\LoudLabel}[1]{\label{#1}%
\marginnote{\marginPseudoURL{#1}}}

\ifnum\OPTIONLoudLabels=1%
\newcommand{\Label}[1]{\LoudLabel{#1}}%
\newcommand{\FLabel}[1]{\label{#1}%
{\tt\scriptsize{#1}}}%
\else%
\newcommand{\Label}[1]{\label{#1}}%
\newcommand{\FLabel}[1]{\label{#1}}%
\fi

\newcommand{\derives}{\ensuremath{\mathrel{::}}}

\newcommand{\ProofCaseRule}[1]{\item \textbf{Case }\textrm{{#1}}: ~ }

\newcommand{\ProofCasesRules}[1]{\item \textbf{Cases }\textrm{{#1}}: ~ }

\gdef\xxDerivationProofCaseColor{N}

\newcommand{\DerivationProofCase}[3]{%
     \smallskip
     \item %
       \parbox[t]{100ex}{%
       \textbf{Case } \\[-0.5em]
       $~$\hspace{5ex}
       \if\xxDerivationProofCaseColor N%
           \ensuremath{%
              \Infer{#1}{#2}{#3}%
            }
       \else%
           \colorbox{\xxDerivationProofCaseColor}{%
              \ensuremath{%
                \Infer{#1}{#2}{#3}%
              }%
           }%
        \fi%
     }%
     \nopagebreak \\[-0.4ex]
  }

\newcommand{\DoubleDerivationProofCase}[6]{%
     \smallskip
     \item %
       \parbox[t]{100ex}{%
       \textbf{Case } \\[-0.5em]
       $~$\hspace{5ex}
       \if\xxDerivationProofCaseColor N%
           \ensuremath{%
              \Infer{#1}{#2}{#3}%
              ~~~~~
              \Infer{#4}{#5}{#6}%
            }
       \else%
           \colorbox{\xxDerivationProofCaseColor}{%
              \ensuremath{%
                \Infer{#1}{#2}{#3}%
                ~~~~~
                \Infer{#4}{#5}{#6}%
              }%
           }%
        \fi%
     }%
     \nopagebreak \\[-0.4ex]
  }

\newcommand{\Infer}[3]{\ensuremath{\inferrule*[right={\text{\strut#1}}]{{}#2\mathstrut}{{}#3\mathstrut}}}

\newcommand{\BeginProof}{\renewcommand{\arraystretch}{1.1} \smallskip \begin{tabular}[b]{r@{}r @{} l  l}}
\newcommand{\EndProof}{\end{tabular} \smallskip \renewcommand{\arraystretch}{\mydefaultarraystretch}}

\newcommand{\Pf}[4] {&$#1$ $#2$\, & $#3$ & #4 \\}
\newcommand{\stepPf}[3] {\Pf{#1}{\,\step\,}{#2}{#3}}

\newcommand{\stepRPf}[3] {\Pf{#1}{\,\stepR\,}{#2}{#3}}
\newcommand{\srcstepPf}[3] {\Pf{#1}{\,\srcstep\,}{#2}{#3}}
\newcommand{\srcstepsPf}[3] {\Pf{#1}{\,\srcsteps\,}{#2}{#3}}

\newcommand{\srcredNPf}[3] {\Pf{#1}{\,\srcredN\,}{#2}{#3}}

\newcommand{\FreePf}[2]{& \multicolumn{2}{l}{#1} & {#2} \\}
\newcommand{\FreePfC}[2]{&  & {#1} & {#2} \\}
\newcommand{\LetPf}[3] {\Pf{\text{Let}\,~{#1}}{=\,}{#2\text{.}}{#3}}

\newcommand{\inPf}[3] {\Pf{#1}{\in}{#2}{#3}}

\newcommand{\ePf}[3] {\Pf{#1}{\entails\,}{#2}{#3}}
\newcommand{\eqPf}[3] {\Pf{#1}{\,=\;}{#2}{#3}}
\newcommand{\valleqPf}[3] {\Pf{#1}{\valleq}{#2}{#3}}

\newcommand{\proofsep}{\,\\[-0.5em]}

\newenvironment{llproof}{\BeginProof}{\EndProof}
\newcommand{\decolumnizePf}{\end{llproof} ~\\ \begin{llproof}}

\newcommand{\proofheading}[1]{}  %

\newcommand{\er}[1]{\ensuremath{\textsl{er}(#1)}}

\newcommand{\xxLet}{\keyword{let}}

\newgray{grayboxgray}{0.85}

\ifnum\OPTIONConf=1
  \newcommand{\judgboxfontsize}[1]{\mathsz{10pt}{#1}}
\else
  \newcommand{\judgboxfontsize}[1]{\mathsz{11pt}{#1}}
\fi

\newcommand{\judgbox}[2]{%
      {\raggedright \textgraybox{\judgboxfontsize{#1}}\!\begin{tabular}[c]{l} #2 \end{tabular} %
      \ifnum\OPTIONConf=1 \\[3pt] \else \\[5pt] \fi%
}}
\newcommand{\judgboxtwelf}[3]{%
      {\raggedright \textgraybox{\judgboxfontsize{\ensuremath{#1}}}\!\begin{tabular}[c]{l} #2 \end{tabular} %
      \!\!\textgraybox{\fontsz{8pt}{#3}}%
      \ifnum\OPTIONConf=1 \\[3pt] \else \\[6pt] \fi%
}}

\newcommand{\Hand}{\text{\Pointinghand~~~~}}

\newcommand{\mathcolor}[2]{\textcolor{#1}{\ensuremath{#2}}}

\newcommand{\E}{\mathcal C}
\newcommand{\EV}{\E_\V}
\newcommand{\EN}{\E_\N}

\newcommand{\mydefaultarraystretch}{1.2}

\newcommand{\wildcard}{\text{\_\!\_}}

\newcommand{\LAMSYM}{\Lambda}
\newcommand{\LAM}[1]{\LAMSYM {#1}.~}

\newcommand{\xtyapp}[1]{\text{\normalfont\tt{[}}{#1}\text{\normalfont\tt{]}}}
\newcommand{\tyapp}[2]{{#1}\xtyapp{#2}}

\newcommand{\tylam}[1]{\LAM{#1}}

\newcommand{\ttylam}{\tylam{\wildcard}}
\newcommand{\ttyapp}[1]{\tyapp{#1}{\wildcard}}

\newenvironment{mathdispl}{\ifnum\OPTIONConf=1\vspace{-15pt}\else\vspace{-16pt}\fi \begin{center}\begin{mathpar} ~\!\!}{\end{mathpar}\end{center}}

\def\quofile{15}
\newcommand{\startMquotes}{\immediate\openout \quofile=\jobname.quo}
\newcommand{\MquoteM}[4]{{\small \label{#3} \begin{quotation} #1 \vspace{-0.5em} \flushright{---#2} \end{quotation}} { \immediate\write\quofile{\expandafter\csname Mquoteentry\endcsname{}{#2}{#3}{\expandafter\csname #4\endcsname}}}}
\newcommand{\Mquote}[4]{{\small \label{#3} \begin{quotation} #1 \vspace{-0.5em} \flushright{---#2} \end{quotation}} { \immediate\write\quofile{\expandafter\csname Mquoteentry\endcsname{}{#2}{#3}{#4}}}}

\newcommand{\AllSym}{\forall}
\newcommand{\xAll}[1]{\AllSym#1}
\newcommand{\All}[1]{\xAll{#1}.\:}

\newcounter{zzInOinkComment}
\newcommand{\oinkkw}[1]{\ifnum\value{zzInOinkComment}=0{\usefont{T1}{cmtt}{m}{it}{\splatbold{#1}}}\else{\normalfont\textsl{#1}}\fi}

\newcommand{\lladdconj}{\mathrel{\binampersand}}

\newcommand{\RZZaddconjR}[1]{\ensuremath{{\lladdconj}\text{R}}\xspace}

\newdimen\zzlistingsize
\newdimen\zzlistingsizedefault
\zzlistingsize=\zzlistingsizedefault
\global\def\CommentCopter{0}
\newcommand{\Lstbasicstyle}{\fontsize{\zzlistingsize}{1.05\zzlistingsize}\ttfamily}
\newcommand{\keywordcopter}{\fontsize{1.0\zzlistingsize}{1.0\zzlistingsize}\bf}
\newcommand{\fontcopter}{\if0\CommentCopter\keywordcopter\fi}
\newcommand{\commentcopter}{\def\CommentCopter{1}\fontsize{0.95\zzlistingsize}{1.0\zzlistingsize}\rmfamily\slshape}

\newlength{\zzlstwidth}
\newcommand{\setlistingsize}[1]{\zzlistingsize=#1%
\settowidth{\zzlstwidth}{{\Lstbasicstyle~}}}
\setlistingsize{\zzlistingsizedefault}

\definecolor{dHilite}{rgb}{0.9, 0.9, 0.6}
\definecolor{dLighterRed}{rgb}{0.9, 0.2, 0.2}

\definecolor{erratumcolor}{rgb}{0.9, 0.2, 0.2}

\definecolor{dRed}{rgb}{0.65, 0.0, 0.0}
\definecolor{dDkRed}{rgb}{0.35, 0.0, 0.0}
\definecolor{DRED}{rgb}{0.65, 0.0, 0.0}
\definecolor{dGreen}{rgb}{0.0, 0.65, 0.0}
\definecolor{dDkGreen}{rgb}{0.0, 0.35, 0.0}
\definecolor{dBlue}{rgb}{0.0, 0.0, 0.65}
\definecolor{dDkBlue}{rgb}{0.0, 0.0, 0.45}
\definecolor{dLightPurple}{rgb}{0.9, 0.5, 0.9}
\definecolor{dLighterPurple}{rgb}{1.0, 0.7, 1.0}
\definecolor{dPurple}{rgb}{0.65, 0.0, 0.65}
\definecolor{dDigPurple}{rgb}{0.5, 0.0, 0.5}
\definecolor{DDIGPURPLE}{rgb}{0.5, 0.0, 0.5}  %
\definecolor{dFaint}{rgb}{0.7, 0.7, 0.7}
\definecolor{dGray}{rgb}{0.5, 0.5, 0.5}
\definecolor{dDark}{rgb}{0.2, 0.2, 0.2}
\definecolor{dAlmostBlack}{rgb}{0.1, 0.1, 0.1}

\newcommand{\tytrans}[1]{|{#1}|}
\newcommand{\ctxtrans}[1]{\tytrans{#1}}

\newcommand{\econtrans}[1]{\floor{#1}}
\newcommand{\expecontrans}[1]{\econtrans{#1}}
\newcommand{\ctxecontrans}[1]{\econtrans{#1}}

\newcommand{\doublecomma}{{,\hspace{-0.2em},\hspace{0.17em}}}
\newcommand{\Merge}[2]{{#1} \doublecomma {#2}}

\newcommand{\suspsymbol}{\blacktriangleright}
\newcommand{\suspsym}[1]{{#1}{\suspsymbol}}
\newcommand{\susp}[1]{\suspsym{#1}}

\newcommand{\recsymbol}{\mu}
\newcommand{\recsym}[1]{{\recsymbol}^{#1}}
\newcommand{\rec}[2]{\recsym{#1}{#2}.\,}

\newcommand{\Rec}[1]{\recsymbol{#1}.\,}

\definecolor{impartialcolor}{rgb}{0.1, 0.1, 0.9}  %
\definecolor{econcolor}{rgb}{0.4, 0.0, 0.0}  %
\definecolor{targetcolor}{rgb}{0.2, 0.2, 0.2}  %

\newcommand{\Intro}{\textcolor{\chkcolor}{Intro}}
\newcommand{\Elim}{\textcolor{\syncolor}{Elim}}
\newcommand{\WeakElim}{\textcolor{\appcolor}{Elim}}

\newcommand{\rulename}[1]{\text{\normalfont\textsf{#1}}}

\newcommand{\Irulename}[1]{\rulename{\textcolor{impartialcolor}{$\Iglyph$#1}}\xspace}
\newrulecommand{Iallintro}{\Irulename{$\AllSym$\Intro}}
\newrulecommand{Iallelim}{\Irulename{$\AllSym$\Elim}}
\newrulecommand{Ialleointro}{\Irulename{$\alleosym$\Intro}}
\newrulecommand{Ialleoelim}{\Irulename{$\alleosym$\Elim}}
\newrulecommand{Isuspintro}{\Irulename{$\suspsymbol$\Intro}}
\newrulecommand{Isuspelim}{\Irulename{$\suspsymbol$\Elim}}
\newrulecommand{Iunitintro}{\Irulename{$\unitty$\Intro}}
\newrulecommand{Isectintro}{\Irulename{$\sectty$\Intro}}
\newrulecommandONE{Isectelim}{\Irulename{$\sectty$\Elim$_{#1}$}}
\newrulecommandONE{Isumintro}{\Irulename{${+}$\Intro$_{#1}$}}
\newrulecommand{Isumelim}{\Irulename{${+}$\Elim}}
\newrulecommand{Iprodintro}{\Irulename{${*}$\Intro}}
\newrulecommandONE{Iprodelim}{\Irulename{${*}$\Elim$_{#1}$}}
\newrulecommand{Irecintro}{\Irulename{${\recsymbol}$\Intro}}
\newrulecommand{Irecelim}{\Irulename{${\recsymbol}$\Elim}}
\newrulecommand{Ivar}{\Irulename{var}}
\newrulecommand{Ifix}{\Irulename{fix}}
\newrulecommand{Ifixvar}{\Irulename{fixvar}}

\newrulecommand{Isub}{\Irulename{sub}}
\newrulecommand{Ianno}{\Irulename{anno}}

\newrulecommand{Iarrintro}{\Irulename{$\arr$\Intro}}
\newrulecommand{Iarrelim}{\Irulename{$\arr$\Elim}}

\newrulecommand{Ivarrintro}{\Irulename{$\varr$\Intro}}
\newrulecommand{Ivarrelim}{\Irulename{$\varr$\Elim}}
\newrulecommand{Inarrintro}{\Irulename{$\narr$\Intro}}
\newrulecommand{Inarrelim}{\Irulename{$\narr$\Elim}}

\newcommand{\Iglyph}{\bf I}
\newcommand{\Eglyph}{\bf E}
\newcommand{\Tglyph}{\bf T}

\newcommand{\Rrulename}[1]{\text{\rulename{\textcolor{econcolor}{$\Eglyph$}#1}}}
\newrulecommand{Rallintro}{\Rrulename{$\AllSym$\Intro}}
\newrulecommand{Rallelim}{\Rrulename{$\AllSym$\WeakElim}}
\newrulecommand{Ralleointro}{\Rrulename{$\alleosym$\Intro}}
\newrulecommand{Ralleoelim}{\Rrulename{$\alleosym$\Elim}}
\newrulecommand{Rsuspintro}{\Rrulename{$\suspsymbol$\Intro}}
\newrulecommandONE{Rsuspelim}{\Rrulename{$\suspsymbol$\Elim$_{#1}$}}
\newrulecommand{Runitintro}{\Rrulename{$\unitty$\Intro}}
\newrulecommand{Rsectintro}{\Rrulename{$\sectty$\Intro}}
\newrulecommandONE{Rsectelim}{\Rrulename{$\sectty$\Elim$_{#1}$}}
\newrulecommandONE{Rsumintro}{\Rrulename{${+}$\Intro$_{#1}$}}
\newrulecommand{Rsumelim}{\Rrulename{${+}$\Elim}}
\newrulecommand{Rprodintro}{\Rrulename{${*}$\Intro}}
\newrulecommandONE{Rprodelim}{\Rrulename{${*}$\Elim$_{#1}$}}
\newrulecommand{Rrecintro}{\Rrulename{${\recsymbol}$\Intro}}
\newrulecommand{Rrecelim}{\Rrulename{${\recsymbol}$\Elim}}
\newrulecommand{Rvar}{\Rrulename{var}}
\newrulecommand{Rfix}{\Rrulename{fix}}
\newrulecommand{Rfixvar}{\Rrulename{fixvar}}

\newrulecommand{Rsub}{\Rrulename{sub}}
\newrulecommand{Rforget}{\Rrulename{forget}}
\newrulecommand{Rweak}{\Rrulename{strongweak}}
\newrulecommand{Ranno}{\Rrulename{anno}}

\newrulecommand{Rarrintro}{\Rrulename{$\arr$\Intro}}
\newrulecommand{Rarrelim}{\Rrulename{$\arr$\Elim}}

\newcommand{\Subname}[1]{\rulename{\textcolor{econcolor}{$\subtype$#1}}\xspace}
\newrulecommand{SubRefl}{\Subname{refl}}
\newrulecommand{SubArr}{\Subname{$\arr$}}
\newrulecommand{SubProd}{\Subname{$*$}}
\newrulecommand{SubSum}{\Subname{$+$}}
\newrulecommand{SubRecL}{\Subname{${\recsymbol}$L}}
\newrulecommand{SubRecR}{\Subname{${\recsymbol}$R}}
\newrulecommand{SubAllL}{\Subname{$\AllSym$L}}
\newrulecommand{SubAllR}{\Subname{$\AllSym$R}}
\newrulecommand{SubAlleoL}{\Subname{$\alleosym$L}}
\newrulecommand{SubAlleoR}{\Subname{$\alleosym$R}}
\newrulecommand{SubGateL}{\Subname{$\suspsymbol$L}}
\newrulecommand{SubGateR}{\Subname{$\suspsymbol$R}}

\newcommand{\elabrulename}[1]{\rulename{\textcolor{econcolor}{$\textsl{elab}$#1}}}
\newrulecommand{Eallintro}{\elabrulename{$\AllSym$\Intro}}
\newrulecommand{Eallelim}{\elabrulename{$\AllSym$\Elim}}
\newrulecommand{Ealleointro}{\elabrulename{$\alleosym$\Intro}}
\newrulecommand{Ealleoelim}{\elabrulename{$\alleosym$\Elim}}
\newrulecommand{Esuspintro}{\elabrulename{$\suspsymbol$\Intro}}
\newrulecommandONE{Esuspelim}{\elabrulename{$\suspsymbol$\Elim$_{#1}$}}
\newrulecommand{Eunitintro}{\elabrulename{$\unitty$\Intro}}
\newrulecommand{Eforget}{\elabrulename{forget XXX}}
\newrulecommand{Esectintro}{\elabrulename{$\sectty$\Intro}}
\newrulecommandONE{Esectelim}{\elabrulename{$\sectty$\Elim$_{#1}$}}
\newrulecommandONE{Esumintro}{\elabrulename{${+}$\Intro$_{#1}$}}
\newrulecommand{Esumelim}{\elabrulename{${+}$\Elim}}
\newrulecommand{Eprodintro}{\elabrulename{${*}$\Intro}}
\newrulecommandONE{Eprodelim}{\elabrulename{${*}$\Elim$_{#1}$}}
\newrulecommand{Erecintro}{\elabrulename{${\recsymbol}$\Intro}}
\newrulecommand{Erecelim}{\elabrulename{${\recsymbol}$\Elim}}
\newrulecommand{Evar}{\elabrulename{var}}
\newrulecommand{Efix}{\elabrulename{fix}}
\newrulecommand{Efixvar}{\elabrulename{fixvar}}

\newrulecommand{Eanno}{\elabrulename{anno}}

\newrulecommand{Earrintro}{\elabrulename{$\arr$\Intro}}
\newrulecommand{Earrelim}{\elabrulename{$\arr$\Elim}}

\newcommand{\Mrulename}[1]{\rulename{\textcolor{targetcolor}{$\Tglyph$}#1}\xspace}

\newrulecommand{Mthunk}{\Mrulename{$\xU$\Intro}}
\newrulecommand{Mforce}{\Mrulename{$\xU$\Elim}}

\newrulecommand{Mret}{\Mrulename{$\xF$\Intro}}
\newrulecommand{Mdo}{\Mrulename{$\xF$\Elim}}

\newrulecommand{Mvar}{\Mrulename{var}}
\newrulecommand{Mfix}{\Mrulename{fix}}
\newrulecommand{Mfixvar}{\Mrulename{fixvar}}
\newrulecommand{Munitintro}{\Mrulename{$\unitty$\Intro}}

\newrulecommand{Marrintro}{\Mrulename{$\arr$\Intro}}
\newrulecommand{Marrelim}{\Mrulename{$\arr$\Elim}}

\newrulecommand{Mthunkintro}{\Mrulename{$\xU$\Intro}}
\newrulecommand{Mthunkelim}{\Mrulename{$\xU$\Elim}}

\newrulecommand{Mvarrintro}{\Mrulename{$\varr$\Intro}}
\newrulecommand{Mvarrelim}{\Mrulename{$\varr$\Elim}}
\newrulecommand{Mnarrintro}{\Mrulename{$\narr$\Intro}}
\newrulecommand{Mnarrelim}{\Mrulename{$\narr$\Elim}}

   \newrulecommand{Msplit}{\Mrulename{$*$\Elim-split}}
\newrulecommand{Mprodintro}{\Mrulename{$*$\Intro}}
\newrulecommandONE{Mprodelim}{\Mrulename{$*$\Elim$_{#1}$}}
\newrulecommandONE{Msumintro}{\Mrulename{$+$\Intro$_{#1}$}}
\newrulecommand{Msumelim}{\Mrulename{$+$\Elim}}
\newrulecommand{Mlet}{\Mrulename{let}}

\newrulecommand{Mrecintro}{\Mrulename{$\recsymbol$\Intro}}
\newrulecommand{Mrecelim}{\Mrulename{$\recsymbol$\Elim}}
\newrulecommand{Mallintro}{\Mrulename{$\AllSym$\Intro}}
\newrulecommand{Mallelim}{\Mrulename{$\AllSym$\Elim}}

\newcommand{\Srcredrulename}[1]{\rulename{{{#1}$\mathsf{reduce}$}}}
\newrulecommand{SrcRedBetaV}{\Srcredrulename{$\beta\V$}}
\newrulecommand{SrcRedBetaN}{\Srcredrulename{$\beta\N$}}
\newrulecommand{SrcRedProjV}{\Srcredrulename{\keyword{proj}\V}}
\newrulecommand{SrcRedProjN}{\Srcredrulename{\keyword{proj}\N}}
\newrulecommand{SrcRedCaseV}{\Srcredrulename{\keyword{case}\V}}
\newrulecommand{SrcRedCaseN}{\Srcredrulename{\keyword{case}\N}}
\newrulecommand{SrcRedFixV}{\Srcredrulename{\keyword{fix}\V}}
\newrulecommand{SrcRedFixN}{\Srcredrulename{\keyword{fix}\N}}

\newrulecommand{SrcRedAnnoV}{\Srcredrulename{\keyword{Anno}\V}}
\newrulecommand{SrcRedAnnoN}{\Srcredrulename{\keyword{Anno}\N}}

\newcommand{\Srcsteprulename}[1]{\rulename{SrcStep#1}}
\newrulecommand{SrcStepContextV}{\Srcsteprulename{Ctx\V}}
\newrulecommand{SrcStepContextN}{\Srcsteprulename{Ctx\N}}

\newcommand{\Redrulename}[1]{\rulename{\textcolor{targetcolor}{{#1}$\mathsf{Reduce}$}}}
\newrulecommand{RedBeta}{\Redrulename{$\beta$}}
\newrulecommand{RedProj}{\Redrulename{\keyword{proj}}}
\newrulecommand{RedTyapp}{\Redrulename{\keyword{tyapp}}}
\newrulecommand{RedForce}{\Redrulename{\keyword{force}}}
\newrulecommand{RedCase}{\Redrulename{\keyword{case}}}
\newrulecommand{RedUnroll}{\Redrulename{\keyword{unroll}}}
\newrulecommand{RedFix}{\Redrulename{\keyword{fix}}}

\newcommand{\Steprulename}[1]{\rulename{\textcolor{targetcolor}{Step#1}}}
\newrulecommand{StepContext}{\Steprulename{Context}}

\newcommand{\textcyss}[1]{{\fontencoding{OT2}\fontfamily{wncyss}\selectfont #1}}

\newcommand{\mathcyss}[1]{\text{\textcyss{#1}}}

\newcommand{\alleosym}{\mathcyss{D}}
\newcommand{\alleo}[1]{\alleosym {#1}.\,}
\newcommand{\eovar}{\mathcyss{a}}
\newcommand{\eo}{\textsl{\,evalorder}}

\newcommand{\garr}[1]{\stackrel{#1}{\rightarrow}}
\newcommand{\narr}{\garr{\N}}
\newcommand{\varr}{\garr{\V}}

\newcommand{\cbpvLetter}[1]{\mboldsf{#1}}
\newcommand{\xU}{\cbpvLetter{U}}
\newcommand{\xF}{\cbpvLetter{F}}
\renewcommand{\U}{\xU\;}

\newcommand{\mboldsf}[1]{\text{\normalfont\bfseries\sffamily\selectfont{#1}}}

\newcommand{\elabsymbol}{\hookrightarrow}

\newcommand{\elab}[4]{\xeta{#1}{#2}{#3} \elabsymbol {#4}}
\newcommand{\elabPf}[6]{\ePf{#1}{\elab{#2}{#3}{#4}{#5}}{#6}}

\newcommand{\explam}[1]{\lambda{#1}.\,}
\newcommand{\expapp}[2]{{#1} \mathrel{\appsymbol} {#2}}

\newcommand{\appsymbol}{\text{\sf @}}
\newcommand{\appsym}[1]{{\appsymbol}^{#1}}
\newcommand{\app}[1]{\,\appsym{#1}\,}

\newcommand{\vappsymbol}{\app{\V}}
\newcommand{\nappsymbol}{\app{\N}}
\newcommand{\vapp}{\mathrel{\appsymbol}}

\newcommand{\tunit}{\unit}
\newcommand{\tfix}[1]{\Fix{#1}}

\newcommand{\tpair}[2]{\textt({#1}{\hspace{-0.2ex}\textt{,}}\ {#2}\textt)}
\newcommand{\tinj}[1]{\textkw{inj}_{#1}\,}
\newcommand{\tproj}[1]{\Proj{#1}}
\newcommand{\tlam}[1]{\Lam{#1}}
\newcommand{\xxtcase}{\textkw{case}}

\newcommand{\tcase}[5]{\xxtcase\textt(%
               {#1}%
               \textt,\,%
               {#2}%
               .%
               {#3}%
               \textt,\,%
               {#4}%
               .%
               {#5}
               \textt)%
               }

\newcommand{\valof}[1]{\text{\normalfont\sf valueness}(#1)}

\newcommand{\etypeoft}{\entailssub{targetcolor}{\Tglyph}}
\newcommand{\xtypeoft}[2]{{#1} : {#2}}
\newcommand{\typeoft}[3]{{#1} \etypeoft \xtypeoft{#2}{#3}}

\newcommand{\typeoftPf}[4]{\Pf{#1}{\etypeoft}{\xtypeoft{#2}{#3}}{#4}}

\newcommand{\etypeof}{\entailssub{impartialcolor}{\Iglyph}}
\newcommand{\xtypeof}[2]{{#1} : {#2}}

\newcommand{\chksym}[1]{{{}_{#1}{\chk}}}
\newcommand{\synsym}[1]{{{}_{#1}{\syn}}}
\newcommand{\xchk}[2]{{#1} \mathrel{\chksym{#2}}}
\newcommand{\xsyn}[2]{{#1} \mathrel{\synsym{#2}}}

\newcommand{\chktype}[4]{{#1} \etypeof \xchk{#2}{#3}{#4}}
\newcommand{\syntype}[4]{{#1} \etypeof \xsyn{#2}{#3}{#4}}
\newcommand{\chktypePf}[5]{\Pf{#1}{\etypeof}{\xchk{#2}{#3}{#4}}{#5}}
\newcommand{\syntypePf}[5]{\Pf{#1}{\etypeof}{\xsyn{#2}{#3}{#4}}{#5}}

\newcommand{\etypeofe}{\entailssub{econcolor}{\Eglyph}}
\newcommand{\xtypeofe}[2]{{#1} : {#2}}

\newcommand{\valuenesscolor}{black}
\newcommand{\VAL}{\textcolor{\valuenesscolor}{\normalfont\sf val}}
\newcommand{\NONVAL}{\mathcolor{\valuenesscolor}{\top}}
\newcommand{\VVAR}{\mathcolor{\valuenesscolor}{\varphi}}
\newcommand{\valleq}{\sqsubseteq}

\newcommand{\xechk}[3]{{#1} \mathrel{\chksym{#2}} {#3}}
\newcommand{\xesyn}[3]{{#1} \mathrel{\synsym{#2}} {#3}}

\newcommand{\xetasym}[1]{\mathrel{{}_{#1}{:}}}
\newcommand{\xeta}[3]{{#1} \xetasym{#2} {#3}}

\newcommand{\chktypee}[4]{{#1}~\etypeofe \xechk{#2}{#3}{#4}}
\newcommand{\chktypeePf}[5]{\Pf{#1}{\etypeofe}{\xechk{#2}{#3}{#4}}{#5}}
\newcommand{\syntypee}[4]{{#1}~\etypeofe \xesyn{#2}{#3}{#4}}
\newcommand{\syntypeePf}[5]{\Pf{#1}{\etypeofe}{\xesyn{#2}{#3}{#4}}{#5}}

\newcommand{\varsymbol}{:}
\newcommand{\var}[2]{{#1} \varsymbol {#2}}

\newcommand{\RuleHead}[1]{\text{\raisebox{0.4em}[0pt]{\ensuremath{\mathsz{\ifnum\OPTIONConf=1 16pt\else 24pt \fi}{#1}}}}~~~~~}

\newcommand{\type}{\textsl{\,type}}

\newcommand{\Vable}{\tilde{V}}

\newcommand{\textgraybox}[1]{\boxed{#1}}

\newcommand{\tightcolorbox}[2]{\setlength{\fboxsep}{1pt}\colorbox{#1}{#2}}

\newgray{grayboxgray}{0.85}
\newcommand{\mathcolorbox}[2]{\text{\tightcolorbox{#1}{$\displaystyle {#2}$}}}

\newcommand{\hilite}[2]{\mathcolorbox{#1}{{#2}\mathstrut}}

\newcommand{\fighi}[1]{\hilite{yellow!45}{#1}}

\newcommand{\joinsym}{\sqcup}
\newcommand{\join}{\mathrel{\joinsym}}

\newcommand{\byinv}[1]{By inversion on rule #1}

\newlength{\erratumrule}
\makeatletter
\long\def\OpenFBox#1#2#3{\fboxsep\FrameSep
   \CustomFBox{}{}{#1}{#2}\FrameRule\FrameRule{\hspace{5pt}#3\hspace{5pt}}}

\long\def\CustomFBox#1#2#3#4#5#6#7{%
  \leavevmode\begingroup
  \setbox\@tempboxa\hbox{%
    \color@begingroup
      \kern\fboxsep{#7}\kern\fboxsep
    \color@endgroup}%
  \hbox{%
    \begingroup
      \@tempdima#4\relax
      \advance\@tempdima\fboxsep
      \advance\@tempdima\dp\@tempboxa
    \expandafter\endgroup\expandafter
    \lower\the\@tempdima\hbox{%
      \vbox{%
        {\color{erratumcolor}\hrule\@height#3}%
        #1%
        \hbox{%
          {\color{erratumcolor}\vrule\@width#5}%
          \vbox{%
            \vskip\fboxsep %
            \copy\@tempboxa
            \vskip\fboxsep}%
          {\color{erratumcolor}\vrule\@width#6\relax}}%
        #2%
        {\color{erratumcolor}\hrule\@height#4}\relax}%
    }%
  }%
  \endgroup
}
\makeatother

\newcommand{\erratumnotifycore}{
        \textcolor{erratumcolor}{\!\fontsz{16pt}{!}} \\
        {\fontsz{6pt}{cf.\ Erratum}} \\
        \fontsz{7pt}{page \pageref{erratum}}
}

\newcommand{\erratumnotify}[1]{%
  \marginpar{~\\[#1]
        \erratumnotifycore}
}

\newcommand{\erratumnotifyfloatleft}[1]{%
  \text{%
     \rput[tr]{0}(-6pt,0){%
         \parbox[s]{30pt}{%
             \vspace{#1}%
             \erratumnotifycore}}}%
  \hfill%
}

\makeatletter
\newenvironment{erratumbox}{%
  \vspace*{4pt}%
  \def\FrameRule{\erratumrule}%
  \def\FrameSep{0pt}%
  \def\OuterFrameSep{0pt}%
  \def\FrameCommand{\OpenFBox\FrameRule\FrameRule}%
  \def\FirstFrameCommand{\OpenFBox\FrameRule\z@}%
  \def\MidFrameCommand{\OpenFBox\z@\z@}%
  \def\LastFrameCommand{\fboxsep30pt
         \OpenFBox\z@\FrameRule}%
  \erratumnotify{3pt}
  \MakeFramed {\advance\hsize0pt \FrameRestore}%
  \vspace*{3pt}
  }{\vspace*{3pt}\endMakeFramed\vspace*{3pt}}
\makeatother

\newcommand{\errorhi}[1]{%
        \psframebox[framesep=1pt,linecolor=erratumcolor,linewidth=\erratumrule]{#1}}
\newcommand{\matherrorhi}[1]{\errorhi{\ensuremath{#1}}}

\newcommand{\ErratumSection}[1]{\section*{\protect\errorhi{Erratum:} #1}}

\usepackage{etoolbox}

\newcommand{\SectionHiliteStrong}[1]{\colorbox{dLightPurple}{#1}}
\newcommand{\SectionHilite}[1]{\colorbox{dLighterPurple}{#1}}

\newcommand{\spamfuscatedemail}{jd~\!\hspace{-1.7pt}169@queensu.ca}

\ifnum\OPTIONConf=1%
  
\else
  
\fi

\title{Elaborating Evaluation-Order Polymorphism}

\ifnum\OPTIONArxiv=0%
        \conferenceinfo{ICFP'15}{August 31 -- September 2, 2015, Vancouver, BC, Canada}
        \CopyrightYear{2015}
        \crdata{/15/08}
        \exclusivelicense
        \doi{nnnnnnn.nnnnnnn}
  \else
   \fi
    \authorinfo{
          Jana Dunfield
    }
    {University of British Columbia \\
     Vancouver, Canada
    }
    {\spamfuscatedemail
\vspace{-2.0ex}}

\begin{document}
\maketitle

\begin{abstract}
  We classify programming languages according to
  evaluation order:
  each language fixes one evaluation order
  as the default, making it transparent to program in that evaluation order,
  and troublesome to program in the other.
  This paper develops a type system that is impartial with respect to evaluation order.
  Evaluation order is implicit in terms, and explicit in types, with by-value and by-name
  versions of type connectives.
  A form of intersection type quantifies over evaluation orders,
  describing code that is agnostic over (that is, polymorphic in) evaluation
  order.  By allowing such generic code, programs can express the by-value and
  by-name versions of a computation without code duplication.

  We also formulate a type system that only has by-value connectives, plus a
  type that generalizes the difference between by-value and by-name connectives:
  it is either a suspension (by name) or a ``no-op'' (by value).  We show a straightforward
  encoding of the impartial type system into the more economical one.
  Then we define an elaboration from the economical language to a
  call-by-value semantics, and prove that elaborating
  a well-typed source program, where evaluation order is implicit, produces a well-typed
  target program where evaluation order is explicit.
  We also prove a simulation between evaluation of the target program and
  reductions (either by-value or by-name) in the source program.

  Finally, we prove that typing, elaboration, and evaluation are faithful
  to the type annotations given in the source program: if the programmer
  only writes by-value types, no by-name reductions can occur at run time.
\end{abstract}

\ifnum\OPTIONConf=1
    {\fontsize{8pt}{9pt}\selectfont%
    {\raggedright
     }
     ~\\ ~\\[-23pt]
    }
\else
\fi
\setcounter{footnote}{0}

\section{Introduction}

It is customary to distinguish languages according to how they pass function arguments.
We tend to treat this as a basic taxonomic distinction: for example, OCaml is a call-by-value
language, while Haskell is call-by-need.  
Yet this taxonomy has been dubious from the start:
Algol-60, in which arguments were call-by-name by default, also
supported call-by-value.
For the $\lambda$-calculus, \citet{Plotkin75} showed how to use
\emph{administrative reductions}
to translate a cbv program into one that behaves equivalently under cbn evaluation,
and vice versa.  Thus, one can write a call-by-name program in a call-by-value language,
and a call-by-value program
in a call-by-name language, but at the price of administrative burdens:
creating and forcing thunks (to simulate call-by-name), or using special
strict forms of function application, binding, etc.\ (to simulate call-by-value).

But programmers rarely want to encode an entire program into a different
evaluation order.  Rather, the issue is how to use
the other evaluation order in \emph{part} of a program.  For example, game search
can be expressed elegantly using a lazy tree, but in an ordinary call-by-value language
one must explicitly create and force thunks.  Conversely, a big advantage
of call-by-value semantics is the relative ease of reasoning about cost (time and space);
to recover some of this ease of reasoning, languages that are not call-by-value often
have strict versions of function application and strictness annotations on types.

\paragraph{An impartial type system.}
For any given language, the language designers' favourite evaluation order is the
linguistically \emph{unmarked} case.  Programmers are not forced to use that order,
but must do extra work to use another,
even in languages with mechanisms specifically designed to mitigate these burdens, such as a
\textit{lazy} keyword~\citep{Wadler98}.

The first step we'll take in this paper is to stop playing favourites:
our source language allows each evaluation order to be used as easily
as the other.
Our \emph{impartial type system} includes by-value and by-name versions of function types
($\varr$, $\narr$), product types ($*^\V$, $*^\N$), sum types
($+^\V$, $+^\N$) and recursive types ($\recsym{\V}$, $\recsym{\N}$).
Using bidirectional typing, which distinguishes checking and
inference, we can use information found in the types of functions
to determine whether an unmarked $\xlam$ or application should be
interpreted as call-by-name or call-by-value.

What if we want to define the same operation over both 
evaluation orders, say \textfn{compose}, or
\textfn{append} (that is, for strict and lazy lists)?
Must we write two identical versions, with
nearly-identical type annotations?  No:  We can use polymorphism based on
intersection types.  The abstruse reputation of intersection types
is belied by a straightforward formulation as implicit products~\citep{Dunfield14},
a notion also used by \citet{Chen14} to express polymorphism
over a finite set of levels (though without using the word ``intersection'').
In these papers' type systems, elaboration takes a polymorphic source program
and produces a target program explicitly specifying 
necessary, but tedious, constructs.  For
\citet{Dunfield14}, the extra constructs introduce and eliminate the products that
were implicit in the source language; for \citet{Chen14}, the extra constructs
support a dynamic dependency graph for efficient incremental computation.

In this paper, we express the intersection type $\sectty$ as a universal
quantifier over evaluation orders.  For example,
the type $\alleo{\eovar} \Int \garr{\eovar} \Int$
corresponds to $(\Int \varr \Int) \sectty (\Int \narr \Int)$.
Thus, we can type code that is generic over evaluation
orders. %
Datatype definitions, expressed as recursive/sum types, can also be polymorphic
in evaluation order; for example, operations on binary search trees can be written
just once.
Much of the theory in this paper follows smoothly from existing work on intersection
types, particularly \citet{Dunfield14}.  However, since we only consider intersections
equivalent to the 
quantified type $\alleo{\eovar} A$, our intersected types have parametric structure:
they differ only in the evaluation orders decorating the connectives.
This limitation, a cousin of the \emph{refinement restriction} in
datasort refinement systems~\citep{Freeman91,DaviesThesis},
avoids the need for a merge construct~\citep{Reynolds96:Forsythe,Dunfield14}
and the issues that arise from it.

\begin{figure}[t]
$\hspace{-10pt}$\begin{tikzpicture}
  [auto, node distance=2cm, >=stealth, %
   descr/.style= {fill=white, inner sep=2.5pt, anchor=center}
  ]
  \node (lsrc) {%
  };
  \node [right of=lsrc,
          node distance=1.3cm,
          label={above:{
           $\begin{array}[t]{l@{~}l@{~}l@{~}l@{~}l}
             \varr & *^\V & +^\V & \recsym{\V}
             \\
             \narr & *^\N & +^\N & \recsym{\N}
             \\
             &~\AllSym
             &\!\!\fighi{\alleosym}
           \end{array}$              
          }}
          ] (upperleft) {$e \synabovechk \tau\mathstrut$};
  \node [above of=upperleft, node distance=2.1cm] (upperleftabove)
          {\tabularenvl{Impartial \\ ~type system}};
  \node [above of=upperleftabove, right=0.5cm, node distance=0.6cm] (upperleftabove) {\bf Source language ($e$)};
  \node [right of=upperleft, node distance=2.0cm,
           label={above:{
                  \arrayenvc{
                     {\arr}
                     ~{*}
                     ~{+}
                     ~{\recsymbol}
                    \\
                    \fighi{\susp{\V}}~\fighi{\susp{\N}}
                    \\
                    \AllSym
                    ~~~
                      \fighi{\alleosym}%
                   }
           }}
        ] (uppercentre) {$e \synabovechk S\mathstrut$};
     \node [above of=uppercentre, node distance=1.9cm]
              (uppercentreabove)
              {\tabularenvl{Economical \\ ~type system}};
     \draw [->] (upperleft)  -- node[below] {\tabularenvc{encode \\ types}} (uppercentre);

   \node[right of=uppercentre, node distance=1.5cm,
           label={above:{\textsl{er}ase types}}
   ] (uppererase) {$\er{e} : S$};

  \node [below of=uppererase, node distance=2.6cm] (lowererase) {$e' : S$};

  \draw [->,
         decorate,decoration=
               {zigzag,segment length=4pt,amplitude=0.8pt,post=lineto,post length=3pt}
      ]
      ($(uppererase.south)+(-7pt,1pt)$)
            -- node[left] {%
                   \!\!\begin{tabular}{l}
                     cbv + cbn \\ evaluation \\
                     ($\geq 0$ steps)
                   \end{tabular}\!\!\!\!}
      ($(lowererase.north)+(-7pt,1pt)$);   
   
   \node[right of=uppererase, node distance=3.0cm,
           label={above:{
                  \arrayenvc{
                     {\arr}
                     ~{*}
                     ~{+}
                     ~{\recsymbol}
                    \\
                    \U\text{(thunk)}
                    \\
                    \AllSym
                   }
           }}
      ] (upperright) {$M : A$};
   \node [below of=upperright, node distance=2.6cm] (lowerright) {$W : A$};

   \node [above of=upperright, node distance=1.6cm, left=-0.6cm] (upperrightabove) {Cbv type system};

   \node [above of=upperrightabove, node distance=0.4cm, left=-1.2cm] (upperrightaboveabove) {\bf Target language ($M$)};

 \draw [right hook->] (uppererase)  -- node[below] {elaborate} (upperright);

 \draw [right hook->] (lowererase)  -- node[above] {elaborate} (lowerright);

  \draw [|->]
      ($(upperright.south)+(-7pt,1pt)$)
            -- node[left] {%
                   \!\!\begin{tabular}{l}
                     standard cbv
                     \\
                     evaluation
                     \\
                     ($\geq 0$ steps)
                   \end{tabular}\!\!\!\!\!}
      ($(lowerright.north)+(-7pt,1pt)$);   
\end{tikzpicture}
\caption{Encoding and elaboration}
\FLabel{fig:diagram}
\end{figure}

\paragraph{A simple, fine-grained type system.}
The source language just described meets our goal of impartiality, but
the large number of connectives yields a slightly unwieldy type system.
Fortunately, we can refine this system by 
abstracting out the differences between the by-name and by-value versions of each connective.
That is, each by-name connective corresponds to a by-value connective with suspensions
(thunks) added:
the by-name function type $S_1 \narr S_2$ corresponds to $(\thunkty S_1) \arr S_2$
where $\arr$ is by-value, whereas $S_1 \varr S_2$ is simply $S_1 \arr S_2$.  Here,
$\thunkty S_1$ is a \emph{thunk type}---essentially $\unitty \arr S_1$.
We realize this difference through a connective $\susp{\epsilon} S$,
read ``$\epsilon$ suspend $S$'', where $\susp{\N} S$
corresponds to $\thunkty S$ and $\susp{\V} S$ is equivalent to $S$.
This gives an economical type system with call-by-value versions of the usual connectives
($\arr$, $*$, $+$, $\mu$), plus $\susp{\epsilon} S$.
This type system is biased towards call-by-value (with call-by-name being
``marked''), but we can easily encode the impartial connectives:
$S_1 \garr{\epsilon} S_2$ becomes $(\susp{\epsilon} S_1) \arr S_2$,
the sum type $S_1 +^\epsilon S_2$ becomes $\susp{\epsilon} (S_1 + S_2)$, etc.

Another advantage of this type system is that, in combination with
polymorphism, it is simple to define variants of data structures that
mix different evaluation orders.  For example, a single list definition can encompass
lists with strict ``next pointers'' (so that ``walking'' the list is guaranteed
linear time) and lazy elements (so that examining the element may not be constant
time), as well as lists with lazy ``next pointers'' and strict contents (so that
``walking'' the list is not guaranteed linear---but once a cons cell has been
produced, its element can be accessed in constant time).

Having arrived at this economical type system for source programs, in which evaluation
order is implicit in terms, we develop an elaboration that produces a target program in which
evaluation order is explicit: thunks are explicitly created and forced, and multiple versions
of functions---by-value and by-name---are generated and selected explicitly.

\paragraph{Contributions.}
This paper makes the following contributions:

\begin{itemize}
\item[(\S\ref{sec:source})]
  We define an \emph{impartial} source language and type system that are equally suited
  to call-by-value and call-by-name.
  Using a type $\alleo{\eovar} \tau$ that quantifies over evaluation orders $\eovar$,
  programmers can define data structures and functions that are generic
  over evaluation order.  The type system is bidirectional, alternating
  between checking an expression against a known type (derived from
  a type annotation) and synthesizing a type from an expression.

\item[(\S\ref{sec:econ})]
  Shifting to a call-by-value perspective, we abstract out the 
  suspensions implicit in the by-name connectives,
  yielding a smaller \emph{economical type system},
  also suitable for a (non-impartial) source language.
  We show that programs well-typed in the impartial type system
  remain well-typed in the economical type system.
  Evaluation order remains implicit in terms, and is specified only in type annotations,
  using the \emph{suspension point} $\susp{\epsilon} S$.

\item[(\S\ref{sec:elab})]
  We give \emph{elaboration typing} rules from the economical type system
  into target programs with fully explicit evaluation order.  We prove that,
  given a well-typed source program, the result of the translation is well-typed in
  a call-by-value target language (\Sectionref{sec:target}).

\item[(\S\ref{sec:consistency})]
  We prove that the target program behaves like the source program:
  when the target takes a step from $M$ to $M'$, the source program that elaborated to $M$
  takes some number of steps, yielding an expression that elaborates to $M'$.
  We also prove that if a program is typed (in the economical type system) without
  by-name suspensions, the source program can take only ``by-value steps''
  possible in a cbv semantics.  This result exploits a kind of subformula property of
  the bidirectional type system.  Finally, we prove that if a program is impartially typed
  without using by-value connectives, it can be economically typed without by-name
  suspensions.

\end{itemize}

{\noindent Figure} \ref{fig:diagram} shows the structure of our
approach.

\paragraph{Extended version with appendices.}
Proofs omitted from the main paper for space reasons can
be found in \citet{Dunfield15arxiv}.

\section{Source Language and Impartial Type System}
\Label{sec:source}

\begin{figure}[htbp]
  \centering

  \begin{bnfarray}
    \text{Program variables}
    & x
    &&&
    \\
    \text{Source expressions}
    & e
    &\bnfas&
        \unit
        \bnfalt
        x
        \bnfalt
        u
        \bnfalt
        \explam{x} e
        \bnfalt
        \expapp{e_1}{e_2}
        \bnfalt
        \Fix{u} e
    \bnfaltBRK
        \tylam{\alpha} e
        \bnfalt
        \tyapp{e}{\tau}
    \bnfalt
        \Anno{e}{\tau}
    \bnfaltBRK
        \Pair{e_1}{e_2}        \bnfalt        \Proj{k}{e}
    \bnfaltBRK
        \Inj{k} e
        \bnfalt
        \Casesum{e}{x_1}{e_1}{x_2}{e_2}
  \end{bnfarray}

  \caption{Impartial source language syntax}
  \FLabel{fig:source}
\end{figure}

\begin{figure}[thbp]
  
  \centering
  \begin{bnfarray}
    \text{Evaluation order vars.}
    & \eovar
    &&&
    \\
    \text{Evaluation orders}
    & \epsilon
    &\bnfas&
       \V
       \bnfalt
       \N
       \bnfalt
       \eovar
    \\
    \text{Type variables}
    & \alpha
    &&&
    \\
    \text{Valuenesses}
    & \VVAR
    &\bnfas&
       \VAL
       \bnfalt
       \NONVAL
    \\
    \text{Source types}
    & \tau
    &\bnfas&
       \unitty
       \bnfalt
       \alpha
       \bnfalt
       \All{\alpha} \tau
       \bnfalt
       \alleo{\eovar} \tau
       \bnfalt%
       \tau_1 \garr{\epsilon} \tau_2
       \bnfaltBRK
       \tau_1 *^\epsilon \tau_2
       \bnfalt
       \tau_1 +^\epsilon \tau_2
       \bnfalt
       \rec{\epsilon}{\alpha} \tau
    \\[0.3ex]
    \text{\small Source typing contexts}
    & \gamma
    &\bnfas&
       \cdot
       \bnfalt
       \gamma, \xsyn x \VVAR \tau
       \bnfalt
       \gamma, \xsyn u \NONVAL \tau
       \bnfaltBRK
       \gamma, \eovar \eo
       \bnfalt
       \gamma, \alpha \type
  \end{bnfarray}
  \caption{Impartial types for the source language}
  \FLabel{fig:impartial-types}
\end{figure}

In our source language (\Figureref{fig:source}),
expressions $e$ are the unit value $\unit$, variables $x$, abstraction
$\explam{x} e$, application $\expapp{e_1}{e_2}$, fixed points $\Fix{u} e$
with fixed point variables $u$,
pairs and projections,
and sums $\Inj{k} e$ with conditionals
$\Casesum{e}{x_1}{e_1}{x_2}{e_2}$ (shorthand for
$\LongCasesum{e}{x_1}{e_1}{x_2}{e_2}$).
Both of our type systems for this source language---the impartial type system in this
section, and the economical type system of \Sectionref{sec:econ}---have features not
evident from the source syntax: polymorphism over evaluation orders, and recursive types.

\subsection{\mksubsection{Values}}
\Label{sec:values}
If we wanted a standard call-by-value language, we would give a grammar for
values, and use values to define the operational semantics (and to impose
a value restriction on polymorphism introduction).
But we want an impartial
language, which means that a function argument $x$ is a value \emph{only} if
the function is being typed under call-by-value.  That is, when checking
$(\Lam{x} e)$ against type $(\tau \varr \tau)$, the variable $x$ should be considered
a value (it will be replaced with a value at run time), but when checking against
$(\tau \narr \tau)$, it should not be considered a value (it could be replaced with
a non-value at run time).  Since ``valueness'' depends on typing,
our typing judgments will have to carry information about whether an expression
should be considered a value.

We will also use valueness to impose a value restriction on polymorphism
over evaluation orders, as well as polymorphism over types; see \Sectionref{sec:value-restriction}.
In contrast, our operational semantics for the source language (\Sectionref{sec:source-opsem}),
which permits two flavours (by-value and by-name) of reductions, will use
a standard syntactic definition of values in the by-value reductions.

\subsection{\mksubsection{An Impartial Type System}}

In terms of evaluation order, the expressions in \Figureref{fig:source} are
a blank slate.  You can imagine them as having whichever evaluation order
you prefer.   You can write down the typing rules for functions,
pairs and sums, and you will get the same rules regardless of which
evaluation order you chose.
This is the conceptual foundation for many
functional languages: start with the simply-typed $\lambda$-calculus, choose
an evaluation order, and build up the language from there.\footnote{%
  The choice need not be easy.  The first call-by-name language, Algol 60,
  also supported call-by-value.  It seems that call-by-value was
  the language committee's preferred default, but
  Peter Naur, the editor of the Algol 60 report, independently
  reversed that decision---which he said was merely one of a ``few matters
  of detail''~\citep[p.\ 112]{Wexelblat81}.
  A committee member, F.L.\ Bauer, said this showed that Naur ``had 
  absorbed the Holy Ghost after the Paris meeting\dots there was nothing one could
  do\dots it was to be swallowed for the sake of loyalty.''~\citep[p.\ 130]{Wexelblat81}.
}
Our goal here is to allow different evaluation orders to be mixed.
As a first approximation, we can try to put evaluation orders in the type system
simply by decorating all the connectives.  For example, in place of the
standard $\arr$-introduction rule
\[
   \Infer{}
        {\gamma, x : \tau_1 \entails e : \tau_2}
        {\gamma \entails (\Lam{x}e) : (\tau_1 \arr \tau_2)}
\]
we can decorate $\arr$ with an evaluation order $\epsilon$ (either $\V$ or $\N$):
\[
   \Infer{}
        {\gamma, x : \tau_1 \entails e : \tau_2}
        {\gamma \entails (\Lam{x}e) : (\tau_1 \garr{\epsilon} \tau_2)}
\]
Products $*$, sums $+$, and recursive types $\mu$ follow similarly.

We add a universal quantifier
$\alleo{\eovar} \tau$
over evaluation orders\footnote{%
  The Cyrillic letter $\alleosym$, transliterated
  into English as \emph{D}, bears some
  resemblance to an $\textsf{A}$ (and thus to $\AllSym$);
  more interestingly, it is the first letter of the Russian word
  \mathcyss{da} (\emph{da}).  Many non-Russian speakers know that this word
  means ``yes'', but another meaning is ``and'', connecting it to intersection types.
}.  %
Its rules follow the usual type-assignment rules for $\AllSym$:
the introduction rule is parametric over 
an arbitrary evaluation order $\eovar$, and the elimination rule
replaces $\eovar$ with a particular evaluation order $\epsilon$:
\[
\Infer{}%
      {{\gamma, \eovar \eo}\entails{e}:{\tau}
      }
      {{\gamma}\entails{e}:{\alleo{\eovar} \tau}}
~~~~~~
\Infer{}%
     {{\gamma}\entails{e}:{\alleo{\eovar} \tau}
       \\
       \gamma \entails \epsilon \eo}
     {{\gamma}\entails{e}:{[\epsilon / \eovar]\tau}}
\]
These straightforward rules have a couple of issues:

\begin{itemize}
\item Whether a program diverges can depend on whether it is run under
  call-by-value, or call-by-name.  The simply-typed $\lambda$-calculus has the
  same typing rules for call-by-value and call-by-name, because those rules cannot
  distinguish programs that return something from programs that diverge.
  Since we want to elaborate to call-by-value or call-by-name depending on
  which type appeared, evaluation depends on the particular typing derivation.
  Suppose that evaluation of $e_2$ diverges, and that $f$ is bound to $(\Lam{x} e_1)$.
  Then whether
  $\expapp{f}{e_2}$ diverges depends on whether the type
  of $f$ has $\varr$ or $\narr$.
  The above rules allow a compiler to make either choice.  Polymorphism in the form
  of $\alleosym$
  aggravates the problem: it is tempting to
  infer for $f$ the principal type $\alleo{\eovar} \cdots \garr{\eovar} \cdots$;
  the compiler can then choose how to instantiate $\eovar$ at each of $f$'s call sites.
  Allowing such code is one of this paper's goals, but only when the programmer
  knows that either evaluation order is sensible and has written an appropriate
  type annotation or module signature.
  
  We resolve this through bidirectional typing, which ensures that quantifiers
  are introduced only via type annotation (a kind of subformula property).
  Internal details of the typing derivation still affect elaboration, and thus evaluation,
  but the internal details will be consistent with programmers' expressed intent.

\item If we extend the language with effects, we may need a value restriction
  in certain rules.  For example, mutable references will break type safety
  unless we add a value restriction to the introduction rules for $\AllSym$
  and $\alleosym$.
  
  A traditional value restriction~\citep{Wright95} would simply require
  changing $e$ to $v$ in the introduction rules, where $v$ is a class of
  syntactic values.  In our setting,
  whether a variable $x$ is a value depends on typing,
  so a value restriction is less straightforward.
  We resolve this by extending the typing judgment with information
  about whether the expression is a value.
\end{itemize}

\paragraph{Bidirectional typing.}
We can refine the traditional typing judgment into \emph{checking}
and \emph{synthesis} judgments.  In the checking judgment $e \chk \tau$,
we already know that $e$ should have type $\tau$, and are checking that $e$ is
consistent with this knowledge.  In the synthesis judgment $e \syn \tau$,
we extract $\tau$ from $e$ itself (perhaps directly from a type annotation),
or from assumptions available in a typing context.

The use of bidirectional typing~\citep{Pierce00,Dunfield13} is
often motivated by the need to typecheck programs that use features
Damas-Milner inference cannot handle,
such as indexed and refinement types~\citep{XiThesis,Davies00icfpIntersectionEffects,Dunfield04:Tridirectional}
and higher-rank polymorphism.
But decidability is not our motivation for using bidirectional typing.
Rather, we want typing to remain
predictable even though evaluation order is implicit.
By following the approach of \citet{Dunfield04:Tridirectional}, in which
``introduction forms check, elimination forms synthesize'', we ensure that
the evaluation orders in typing match what programmers intended: a type
connective with a $\V$ or $\N$ evaluation order can be introduced \emph{only}
by a checking judgment.  Since the types in checking judgments are derived
from type annotations, they match the programmer's expressed intent.

Programmers must write annotations on expressions
that are redexes: in $\expapp{(\Lam{x} e)}{e_2}$,
the $\lambda$ needs an annotation, because $\Lam{x} e$ is an introduction form
in an elimination position: $\expapp{\hole}{e_2}$.
In contrast, $\expapp{f}{(\Lam{x} e_2)}$ needs no annotation,
though the type of $f$ must be derived (if indirectly) from an annotation.
Recursive functions $\Fix{u} \Lam{x} e$ ``reduce'' to their unfolding,
so they also need annotations.

\paragraph{Valueness.}
Whether an expression is a value may depend on typing, so we put a \emph{valueness}
in the typing judgments: $\xsyn e \VAL S$ (or $\xchk e \VAL S$) means that
$e$ at type $S$ is definitely a value,
while $\xsyn e \NONVAL S$
(or $\xchk e \NONVAL S$)
means that $e$ at type $S$ is not known to be a value.
In the style of abstract interpretation,
we have a partial order $\valleq$ such that $\VAL \valleq \NONVAL$.
Then the \emph{join} $\VVAR_1 \join \VVAR_2$ is $\VAL$
when $\VVAR_1 = \VVAR_2 = \VAL$, and $\NONVAL$ otherwise.
g
Since valueness is just a projection of $\epsilon$, we could formulate the
system without it, using $\epsilon$ to mark judgments as
denoting values ($\V$) or possible nonvalues ($\N$).  But that seems
prone to confusion: is $\chksym{\N}$ saying the expression is
``by name'' in some sense?

\paragraph{Types and typing contexts.}
In \Figureref{fig:impartial-types} we show the grammar for
evaluation orders $\epsilon$, which are either by-value ($\V$), by-name ($\N$),
or an evaluation order variable $\eovar$.  We have the unit type $\unitty$,
type variables $\alpha$, ordinary parametric polymorphism $\All{\alpha} \tau$,
evaluation order polymorphism $\alleo{\eovar} \tau$,
functions $\tau_1 \garr{\epsilon} \tau_2$,
products $\tau_1 *^{\epsilon} \tau_2$,
sums $\tau_1 +^{\epsilon} \tau_2$,
and recursive types $\rec{\epsilon}{\alpha} \tau$.

A source typing context $\gamma$ consists of variable declarations $\xsyn x \VVAR \tau$
denoting that $x$ has type $\tau$ with valueness $\VVAR$, fixed-point variable declarations
$\xsyn u \NONVAL \tau$ (fixed-point variables are never values),
evaluation-order variable declarations $\eovar \eo$, and type variable declarations
$\alpha \type$.

\begin{figure*}[htbp]
  \centering
  \judgbox{\valof{\epsilon} = \VVAR}
               {Evaluation order $\epsilon$ maps to valueness $\VVAR$}
  \vspace{-4ex}
  $~$\hspace{33ex}\begin{tabular}[t]{r@{~~~}c@{~~~}lll}
    $\valof{\V}$ &$=$& \VAL
  \\
    $\valof{\N}$ &$=$& \NONVAL
  \\
    $\valof{\eovar}$ &$=$& \NONVAL
  \end{tabular}

  \vspace{-1.0ex}

      \judgbox{\arrayenvcl{\chktype{\gamma}{e}{\VVAR}{\tau}
                      \\
                      \syntype{\gamma}{e}{\VVAR}{\tau}}}
                   {Source expression $e$ checks against impartial type $\tau$
                    \\[0.7ex]
                    Source expression $e$ synthesizes impartial type $\tau$}
  \vspace{-1.0ex}
  \begin{mathpar}
      \Infer{\Ivar}
            {(\xsyn x \VVAR \tau) \in \gamma}
            {\syntype \gamma x \VVAR \tau}
      \and
      \Infer{\Ifixvar}
            {(\xsyn u \NONVAL \tau) \in \gamma}
            {\syntype \gamma u \NONVAL \tau}
      ~~~~
      \Infer{\Ifix}
            {\chktype{\gamma, \xsyn u \NONVAL \tau}{e}{\VVAR}{\tau}}
            {\chktype{\gamma}{(\Fix{u} e)}{\NONVAL}{\tau}}
    ~~~~~~~~
      \Infer{\Isub}
            {\syntype{\gamma}{e}{\VVAR}{\tau}}
            {\chktype{\gamma}{e}{\VVAR}{\tau}}
    ~~~~~~
      \Infer{\Ianno}
            {\chktype{\gamma}{e}{\VVAR}{\tau}}
            {\syntype{\gamma}{\Anno{e}{\tau}}{\VVAR}{\tau}}
\\
\RuleHead{\AllSym}
\Infer{\Iallintro}
      {\chktype{\gamma, \alpha \type}{e}{\VAL}{\tau}
      }
      {\chktype{\gamma}{\tylam{\alpha} e}{\VAL}{\All{\alpha} \tau}}
\and
\Infer{\Iallelim}
     {\syntype{\gamma}{e}{\VVAR}{\All{\alpha} \tau}
      \\
      \gamma \entails \tau' \type}
     {\syntype{\gamma}{\tyapp{e}{\tau'}}{\VVAR}{[\tau' / \alpha]\tau}}
\and
\RuleHead{\unitty}
   \Infer{\Iunitintro}
             {}
             {\chktype{\gamma}{\unit}{\VAL}{\unitty}}
   \\
\RuleHead{\alleosym}
\Infer{\Ialleointro}
      {\chktype{\gamma, \eovar \eo}{e}{\VAL}{\tau}
      }
      {\chktype{\gamma}{e}{\VAL}{\alleo{\eovar} \tau}}
\and
\Infer{\Ialleoelim}
     {\syntype{\gamma}{e}{\VVAR}{\alleo{\eovar} \tau}
       \\
       \gamma \entails \epsilon \eo}
     {\syntype{\gamma}{e}{\VVAR}{[\epsilon / \eovar]\tau}}
   \\
    \RuleHead{\garr{\epsilon}}
        \Infer{\Iarrintro}
              {\chktype{\gamma, (\xsyn x {\valof\epsilon} \tau_1)}{e}{\VVAR}{\tau_2}
              }
              {\chktype{\gamma}{(\explam{x} e)}{\VAL}{(\tau_1 \garr{\epsilon} \tau_2)}}
       \and
       \Infer{\Iarrelim}
              {\syntype{\gamma}{e_1}{\VVAR_1}{(\tau_1 \garr{\epsilon} \tau_2)}
               \\
               \chktype{\gamma}{e_2}{\VVAR_2}{\tau_1}
              }
              {\syntype{\gamma}{(\expapp{e_1}{e_2})}{\NONVAL}{\tau_2}
              }
    \\
    \RuleHead{*^\epsilon}
          \Infer{\Iprodintro}
                 {\chktype{\gamma}{e_1}{\VVAR_1}{\tau_1}
                   \\
                   \chktype{\gamma}{e_2}{\VVAR_2}{\tau_2}
                  }
                  {\chktype{\gamma}{\Pair{e_1}{e_2}}{\VVAR_1 \join \VVAR_2}{(\tau_1 *^\epsilon \tau_2)}}
           ~~~~~
           \Infer{\Iprodelim{k}}
                 {
                   \syntype{\gamma}
                           {e}
                           {\VVAR}
                           {(\tau_1 *^\epsilon \tau_2)}
                 }
                 {
                   \syntype{\gamma}
                        {(\Proj{k} e)}
                        {\NONVAL}
                        {\tau_k}
                 }
          \hspace{42ex}
    \vspace{-2.0ex}
    \\
    \RuleHead{+^\epsilon}
       \Infer{\Isumintro{k}}
             {\chktype{\gamma}{e}{\VVAR}{\tau_k}
             }
             {\chktype{\gamma}{(\Inj{k} e)}{\VVAR}{(\tau_1 +^\epsilon \tau_2)}}
       \and
       \Infer{\Isumelim}
              {   \syntype{\gamma}{e}{\VVAR_0}{(\tau_1 +^\epsilon \tau_2)}
                  \\
                  \arrayenvbl{
                      \chktype{\gamma, (\xsyn{x_1}{\VAL}{\tau_1})}{e_1}{\VVAR_1}{\tau}
                      \\
                      \chktype{\gamma, (\xsyn{x_2}{\VAL}{\tau_2})}{e_2}{\VVAR_2}{\tau}
                  }
              }
              {\chktype{\gamma}{\Casesum{e}{x_1}{e_1}{x_2}{e_2}}{\NONVAL}{\tau}
              }
    \\
    \RuleHead{\recsym{\epsilon}}
       \Infer{\Irecintro}
             {\chktype{\gamma}{e}{\VVAR}{\big[(\rec{\epsilon}{\alpha}\tau)\big/\alpha\big]\tau}
             }
             {\chktype{\gamma}{e}{\VVAR}{\rec{\epsilon}{\alpha} \tau}}
       \and
       \Infer{\Irecelim}
              {   \syntype{\gamma}{e}{\VVAR}{\rec{\epsilon}{\alpha} \tau}
              }
              {
                \syntype{\gamma}{e}{\NONVAL}{\big[(\rec{\epsilon}{\alpha}\tau)\big/\alpha\big]\tau}
              }
    \end{mathpar}

  \caption{Impartial bidirectional typing for the source language}
  \FLabel{fig:impartial-typing}
\end{figure*}

\paragraph{Impartial typing judgments.}
\Figureref{fig:impartial-typing} shows the bidirectional rules for impartial typing.
The judgment forms are
$\chktype{\gamma}{e}{\VVAR}{\tau}$, meaning that $e$ checks against $\tau$
(with valueness $\VVAR$), and $\syntype{\gamma}{e}{\VVAR}{\tau}$, meaning
that $e$ synthesizes type $\tau$.  The ``I'' on the turnstile stands for ``impartial''.

\paragraph{Connective-independent rules.}
Rules \Ivar and \Ifixvar simply use assumptions stored in $\gamma$.
Rule \Ifix checks a fixed point $\Fix{u} e$ against type $\tau$ by introducing the assumption
$\xsyn u \NONVAL \tau$ and checking $e$ against $\tau$; its premise has valueness $\VVAR$
because even if $e$ is a value, $\Fix{u} e$ is not ($\NONVAL$ in the conclusion).

Rule \Isub says that if $e$ synthesizes $\tau$ then $e$ checks against $\tau$.  For example,
in the (ill-advised) fixed point expression $\Fix{u} u$, the premise of \Ifix tries to check
$u$ against $\tau$, but \Ifixvar derives a synthesis judgment, not a checking judgment;
\Isub bridges the gap.

Rule \Ianno also mediates between synthesis and checking, in the opposite direction:
if we can check an expression $e$ against an annotated type $\tau$,
then $\Anno{e}{\tau}$ synthesizes $\tau$.

\paragraph{Introductions and eliminations.}
The rest of the rules are linked to type connectives.  For easy reference, the figure
shows each connective to the left of its introduction and elimination rules.
We follow the recipe of \citet{Dunfield04:Tridirectional}:
introduction rules check, and elimination rules synthesize.  This recipe
yields the smallest sensible set of rules, omitting some rules that are not
absolutely necessary but can be useful in practice.  For example, our rules never
synthesize a type for an unannotated pair, because the pair is an introduction form.

Rule \Isumelim follows the recipe, despite having a checking judgment in its
conclusion: the connective being eliminated, $+^\epsilon$, \emph{is} synthesized
(in the first premise).

\paragraph{Functions.}
Rule \Iarrintro introduces the type $\tau_1 \garr{\epsilon} \tau_2$.
Its premise adds an assumption $\xsyn x {\valof\epsilon} {\tau_1}$, where
$\valof\epsilon$ is $\VAL$ if $\epsilon = \V$, and $\NONVAL$ if $\epsilon$ is $\N$
or is an evaluation-order variable $\eovar$.  This rule thereby encompasses both
variables that will be substituted with values ($\valof\epsilon = \VAL$) and variables that
might be substituted with non-values ($\valof\epsilon = \NONVAL$).
Applying a function of type $\tau_1 \garr{\epsilon} \tau_2$ yields
something of type $\tau_2$ regardless of $\epsilon$, so \Iarrelim ignores $\epsilon$.

Consistent with the usual definition of syntactic values, \Iarrintro's conclusion has $\VAL$,
while \Iarrelim's conclusion has $\NONVAL$.

In rule \Iarrelim, the first premise has the connective to eliminate, so the first premise
synthesizes $(\tau_1 +^\epsilon \tau_2)$.  This provides the type $\tau_1$, so
the second premise is a checking judgment; it also provides $\tau_2$, so the conclusion
synthesizes.

\paragraph{Products.}
Rule \Iprodintro types a value if and only if both $e_1$ and $e_2$ are typed as values,
so its conclusion has valueness $\VVAR_1 \join \VVAR_2$.

\paragraph{Sums.}
Rule \Isumintro{k} is straightforward.
In rule \Isumelim, the assumptions added to $\gamma$ in the branches
say that $x_1$ and $x_2$ are values ($\VAL$), because our by-name sum type
is ``by-name'' on the \emph{outside}.  This point should become more clear
when we see the translation of types into the economical system.

\paragraph{Recursive types.}
Rules \Irecintro and \Irecelim have the same $e$ in the premise and conclusion,
without explicit ``roll'' and ``unroll'' constructs.
In a non-bidirectional type inference system, this would be awkward
since the expression doesn't give direct clues about when to apply these rules.
In this bidirectional system, the type tells us to apply \Irecintro (since its conclusion
is a checking judgment).  Knowing when to apply \Irecelim is more subtle: we should
try to apply it whenever we need to synthesize some \emph{other} type connective.  For instance,
the first premise of \Isumelim needs a $+$, so if we synthesize a $\mu$-type
we should apply \Irecelim in the hope of exposing a $+$.

The lack of explicit [un]rolls suggests that these are not iso-recursive but equi-recursive
types~\citep[chapter 20]{Pierce02:TAPL}.  However, we don't semantically equate a recursive type with its unfolding,
so perhaps they should be called \emph{implicitly} iso-recursive.

Note that an implementation would need to check that the type under the
$\mu$ is guarded by a type connective that does have explicit constructs,
to rule out types like $\rec{\epsilon}{\alpha} \alpha$, which is its own unfolding
and could make the typechecker run in circles.

\paragraph{Explicit type polymorphism.}
In contrast to recursive types, we explicitly introduce and eliminate
type polymorphism via the expressions $\tylam{\alpha} e$ and $\tyapp{M}{\tau}$.
This guarantees that a $\forall$ can be instantiated with a type
containing a particular evaluation order if and only if such a type
appears in the source program.

\paragraph{Principality.}
Suppose $\syntype \gamma {e_1} \phi {\alleo{\eovar} \tau_1 \arr \tau_2}$.
Then, for any $\epsilon$, we can derive
$\syntype \gamma {\expapp{e_1}{e_2}} \NONVAL {[\epsilon/\eovar]\tau_2}$.
But we can't use \Ialleointro to derive the type $\alleo{\eovar'} [\eovar'/\eovar]\tau_2$,
because $\expapp{e_1}{e_2}$.  The only sense in which this expression
has a principal type is if we have an evaluation-order variable in $\gamma$
that we can substitute for $\eovar$.

\subsection{\mksubsection{Programming with Polymorphic Evaluation Order}}
\Label{sec:impartial-examples}

\paragraph{Lists and streams.}
The impartial type system can express lists and
(potentially terminating) streams in a single declaration:
\[
   \textkw{type}~\tyname{List}\;\eovar\;\alpha
   ~=~
         \rec{\eovar}{\beta}
         \big(
             \unitty
          +^\eovar
             (
                \alpha *^\eovar \beta
             )
         \big)
\]
Choosing $\eovar = \V$ yields 
$
         \rec{\V}{\beta}
         \big(
             \unitty
          +^\V
             (
                \alpha *^\V \beta
             )
         \big)
$, which is the type of lists of elements $\alpha$.
Choosing $\eovar = \N$ yields
$
         \rec{\N}{\beta}
         \big(
             \unitty
          +^\N
             (
                \alpha *^\N \beta
             )
         \big)
$, which is the type of streams that may end---essentially, lazy lists.
Since evaluation order is implicit in source expressions, we can write
operations on $\tyname{List}\;\eovar\;\alpha$ that work for lists
\emph{and} streams:

\begin{mathdispl}
  \arrayenvbl{
      \textfn{map}
      :
      \alleo{\eovar}
         \All{\alpha}
             (\alpha \garr{\V} \beta)
             \garr{\V}
             (\tyname{List}\;\eovar\;\alpha)
             \garr{\V}
             (\tyname{List}\;\eovar\;\beta)
      \\
      =~
          \tylam{\alpha}
              \Fix{map}
                  \Lam{f}
                      \Lam{xs}
      \\
      ~~~~~~~~
      \Casesumstart{xs}{x_1}{\Inj{1} \unit}
      \\~~~~~~~~~~~~~~~~~~~~~~~~
       \Casesumend{x_2}{%
            \Inj{2} \Pair{\expapp{f}{(\Proj{1} x_2)}}
                                   {\expapp{\expapp{map}{f}}{(\Proj{2} x_2)}}
      }
  }
\end{mathdispl}
This sugar-free syntax bristles; in an implementation
with conveniences like pattern-matching on tuples and named constructors,
we could write
\begin{mathdispl}
  \arrayenvbl{
      \textfn{map}\;f\;xs
      :
      \alleo{\eovar}
         \All{\alpha}
             (\alpha \garr{\V} \beta)
             \garr{\V}
             (\tyname{List}\;\eovar\;\alpha)
             \garr{\V}
             (\tyname{List}\;\eovar\;\beta)
      \\
      =~
      \Case{xs} \Match{\Nil} \Nil
      \\~~~~~~~~~~~~~~~~~~~
       \matchor
            \Match{\Cons\Pair{hd}{tl}}
                   \Cons\Pair{f\;hd}{\textfn{map}\;f\;tl}
  }
\end{mathdispl}
Note that, except for the type, this is standard code for \textfn{map}.

Even this small example raises interesting questions:

\begin{itemize}
\item Must \emph{all} the connectives in \tyname{List}
   have $\eovar$?  No.  Putting $\eovar$ on either the $\mu$ or the
   $+$ and writing $\V$ on the other connectives is enough to get
   stream behaviour when
   $\eovar$ is instantiated with $\N$:
   the only
   reason to eliminate (unroll) the $\mu$ is to eliminate (case on) the $+$;
   marking either connective will suspend the underlying computation.
   Marking both $\mu$ and $+$ induces a suspension of a suspension,
   where forcing the outer suspension immediately forces the inner one;
   one of the suspensions is superfluous.
   
   Note that marking only $*$ with
   $\eovar$, that is,
        $\rec{\V}{\beta}
         \big(
             \unitty
          +^\V
             (
                \alpha *^\eovar \beta
             )
         \big)$, yields an ``odd'' data structure~\citep{Wadler98},
         one that is not entirely lazy:
         we know immediately---without forcing a thunk---which
         injection we have (\ie whether we have $\Nil$ or $\Cons$).
 
\item What evaluation orders should we use in the type of \textfn{map}?
   We used by-value ($\garr{\V}$), but we could
   use the same evaluation order as the list:
   $
      \alleo{\eovar}
         \All{\alpha}
             (\alpha \garr{\eovar} \beta)
             \garr{\eovar}
             (\tyname{List}\;\eovar\;\alpha)
             \garr{\eovar}
             (\tyname{List}\;\eovar\;\beta)
   $.
   This essentially gives ``ML-ish'' behaviour when $\eovar = \V$,
   and ``Haskell-ish'' behaviour when $\eovar = \N$.  The
   type system, however, permits other variants---even the
   outlandishly generic
   \[
      \alleo{\eovar_1,
              \eovar_2,
              \eovar_3,
              \eovar_4,
              \eovar_5}
         \All{\alpha}
             (\alpha \garr{\eovar_1} \beta)
             \garr{\eovar_2}
             (\tyname{List}\;\eovar_3\;\alpha)
             \garr{\eovar_4}
             (\tyname{List}\;\eovar_5\;\beta)
   \]
\end{itemize}
We leave deeper investigation of these questions to future work:
our purpose, in this paper, is to develop the type systems that make
such questions matter.

\paragraph{Variations in being odd and even.}
The Standard ML type of ``streams in odd style''~\citep[Fig.\ 1]{Wadler98},
given by
\[
\mathsz{8pt}{\texttt{\textkw{datatype} $\alpha$ stream = Nil | Cons of $\alpha$ * $\alpha$ stream susp}}
\]
where \texttt{$\alpha$ stream susp} is the type of a thunk that yields an \texttt{$\alpha$ stream},
can be represented as the impartial type
        $\rec{\V}{\beta}
         \big(
             \unitty
          +^\V
             (
                \alpha *^\V (\rec{\N}{\gamma} \beta)
             )
         \big)$.
Note the slightly awkward $(\rec{\N}{\gamma} \beta)$, in which $\gamma$
doesn't occur; we can't simply write $\rec{N}{\beta}$ on the outside, because
that would suspend the entire sum.  (In the economical type system in
\Sectionref{sec:econ}, it's easy to put the suspension in either position.)
This type differs subtly from another ``odd'' stream type,
        $\rec{\V}{\beta}
         \big(
             \unitty
          +^\V
             (
                \alpha *^\eovar \beta
             )
         \big)$,
which corresponds to the SML type
\[
\mathsz{8pt}{\texttt{\textkw{datatype} $\alpha$ stream = Nil | Cons of \fighi{\texttt(}$\alpha$ * $\alpha$ stream\fighi{\texttt)} susp}}
\]
Here, the contents $\alpha$ are under the suspension; given a value
of this type, we immediately know whether we have $\Nil$ or
$\Cons$, but we must force a thunk to see what the value is, which
will also reveal whether the tail is $\Nil$ or $\Cons$.

We can also encode ``streams in even style''~\citep[Fig.\ 2]{Wadler98}:
The SML declarations
\[
\arrayenvl{
\mathsz{8pt}{\texttt{\textkw{datatype} $\alpha$ stream\_ = Nil\_ | Cons\_ of $\alpha$ * $\alpha$ stream}}
\\
\mathsz{8pt}{\texttt{\textkw{withtype} $\alpha$ stream~~= $\alpha$ stream\_ susp}}
}
\]
correspond to
        $\rec{\N}{\beta}
         \big(
             \unitty
          +^\V
             (
                \alpha *^\V \beta)
             )
         \big)$, with the $\N$ on $\recsymbol$ playing the role of the
\texttt{withtype} declaration.

\citet{Wadler98} note that ``streams in odd style'' can be encoded with ease
in SML, while ``streams in even style'' can be encoded with difficulty
(see their Figure 2).  In the impartial type system, both encodings are straightforward,
and we would only need to write one (polymorphic) version of each of their functions
over streams.

\paragraph{Binary trees.}
As with lists, we can define evaluation-order-polymorphic trees:
\[
   \textkw{type}~\tyname{Tree}\;\eovar\;\alpha
   ~=~
         \rec{\eovar}{\beta}
         \big(
             \unitty
          +^\V
             (
                \alpha *^\V \beta *^\V \beta
             )
         \big)
\]
Here, only $\mu$ is polymorphic in $\eovar$, to suppress
redundant thunks.

\subsection{\mksubsection{Operational Semantics for the Source Language}}
\Label{sec:source-opsem}

\begin{figure}[t]
  \centering

  $~$\!\!\!\!\!\!\begin{bnfarray}
    \text{Source values}
    & v
    &\bnfas&
        \unit
        \bnfalt
        \explam{x} e
        \bnfalt        \Pair{v_1}{v_2}
        \bnfalt
        \Inj{k} v
\\[0.6ex]
     \text{\small By-value eval.\ contexts}
            & \EV
            &\bnfas&
              \hole
              \bnfaltBRK
              \EV \vapp e_2
              \bnfalt
              v_1 \vapp \EV
              \bnfaltBRK
              \Pair{\EV}{e_2}
              \bnfalt
              \Pair{v_1}{\EV}
              \bnfalt
              \Proj{k}{\EV}
              \bnfaltBRK
              \tinj{k}{\EV}
              \bnfalt
              \Casesum{\EV}{x_1}{e_1}{x_2}{e_2}
     \\[0.6ex]
     \text{\small By-name eval.\ contexts}
            & \EN
            &\bnfas&
              \hole
              \bnfaltBRK
              \EN \vapp e_2
              \bnfalt
              \matherrorhi{e_1 \vapp \EN}
              \bnfaltBRK
              \matherrorhi{\Pair{\EN}{e_2}}
              \bnfalt
              \matherrorhi{\Pair{{e_1}}{\EN}}
              \bnfalt
              \Proj{k}{\EN}
              \bnfaltBRK
              \tinj{k}{\EN}
              \bnfalt
              \Casesum{\EN}{x_1}{e_1}{x_2}{e_2}
  \end{bnfarray}

  {\erratumnotifyfloatleft{-40pt}\hfill}

  \judgbox{e \srcstep e'}
        {Source expression $e$ steps to $e'$}
  \vspace{-3ex}
  \begin{mathpar}
     \Infer{\SrcStepContextV}
           {e \srcredV e'}
           {\EV[e] \srcstep \EV[e']}
~~~~~
     \Infer{\SrcStepContextN}
           {e \srcredN e'}
           {\EN[e] \srcstep \EN[e']}
  \end{mathpar}

  \judgbox{\arrayenvcl{e \srcredV e' \\ e \srcredN e'}}{%
     \tabularenvcl{$e$ reduces to $e'$ by value
       \\[0.3ex]
       $e$ reduces to $e'$ by name}}
  \begin{array}[t]{r@{~~}c@{~~}ll}
    \expapp{(\Lam{x} e_1)}{v_2}
           &\srcredV&{}
          [v_2/x]e_1
    & \SrcRedBetaV
\\[0.2ex]
          \expapp{(\Lam{x} e_1)}{e_2}
          &\srcredN&{}
          [e_2/x]e_1
    & \SrcRedBetaN
\\[0.5ex]
          (\Fix{u} e)
          &\srcredV&
            \big[(\Fix{u} e) \big/ u\big]e          
    & \SrcRedFixV
\\[0.2ex]
          (\Fix{u} e)
    &\srcredN&
            \big[(\Fix{u} e) \big/ u\big]e
    & \SrcRedFixN
\\[0.5ex]
        \Proj{k}{\Pair{v_1}{v_2}}
         &\srcredV&
         v_k
    & \SrcRedProjV
\\[0.2ex]
        \Proj{k}{\Pair{e_1}{e_2}}
         &\srcredN&
         e_k
    & \SrcRedProjN
\\[0.5ex]
          \multicolumn{3}{l}{\Casesum{\tinj{k} v}{x_1}{e_1}{x_2}{e_2}
           ~~\srcredV~~
           [v/x_k]e_k}
    & \SrcRedCaseV
\\[0.2ex]
                    \multicolumn{3}{l}{\Casesum{\tinj{k} e}{x_1}{e_1}{x_2}{e_2}
           ~~\srcredN~~
           [e/x_k]e_k}
    & \SrcRedCaseN
  \end{array}
  
  \caption{Source reduction}
  \FLabel{fig:src-step}
\end{figure}

\begin{figure}[t]
  \centering

  \judgbox{\er{e} = e'}
               {Source expression $e$ erases to $e'$}
  \begin{tabular}[t]{r@{~~}c@{~~}lll}
    $\er{\fighi{\tylam{\alpha} e}}$ &$=$& $\er{e}$
  \\[0.2ex]
    $\er{\fighi{\tyapp{e}{S}}}$ &$=$& $\er{e}$
  \\[0.2ex]
    $\er{\fighi{\Anno{e}{S}}}$ &$=$& $\er{e}$
  \end{tabular}
  ~
  \begin{tabular}[t]{r@{~~}c@{~~}lll}
    $\er{\unit}$ &$=$& $\unit$
  \\
    $\er{x}$ &$=$& $x$
  \\
    $\er{\expapp{e_1}{e_2}}$ &$=$& $\expapp{\er{e_1}}{\er{e_2}}$
  \\
  & etc.
  \end{tabular}

  \caption{Erasing types from source expressions}
  \FLabel{fig:er}
\end{figure}

A source expression takes a step
if a subterm in evaluation position can be reduced.
We want to model by-value computation \emph{and}
by-name computation, so we define the source stepping relation
$\srcstep$
usings two notions of evaluation position
and two notions of reduction.  A \emph{by-value evaluation context}
$\EV$ is an expression with a hole $\hole$, where $\EV[e]$
is the expression with $e$ in place of the $\hole$.
If $e$ reduces by value to $e'$, written $e \srcredV e'$,
then $\EV[e] \srcstep \EV[e']$.
For example, if $e_2 \srcredV e_2'$ then
$\expapp{v_1}{e_2} \srcstep \expapp{v_1}{e_2'}$,
because $\expapp{v_1}{\hole}$ is a by-value evaluation context.

\definecolor{TFFrameColor}{rgb}{0.9, 0.2, 0.2}
\definecolor{TFTitleColor}{rgb}{1.0, 1.0, 0.0}

\begin{erratumbox}
Dually, $\EN[e] \srcstep \EN[e']$
if $e \srcredN e'$.  Every by-value context is a by-name context,
and every pair related by $\srcredV$ is also related by $\srcredN$,
but the converses do not hold. For instance,
$\expapp{e_1}{\hole}$ is a $\EN$ but not a $\EV$,
and $\Proj{2} \Pair{e_1}{e_2} \srcredN e_2$,
but $\Proj{2} \Pair{e_1}{e_2}$ reduces by value only
when $e_1$ and $e_2$ are values.

Values, by-value evaluation contexts $\EV$, by-name evaluation contexts $\EN$,
and the relations $\srcstep$, $\srcredV$ and $\srcredN$ are
defined in \Figureref{fig:src-step}.  The definitions of $v$, $\EV$ and $\srcredV$,
taken together, are standard for call-by-value;
the definitions of $\EN$ and $\srcredN$ are standard for call-by-name.
The peculiarity is that $\srcstep$ can behave either by value
(rule \SrcStepContextV) or by name (rule \SrcStepContextN).
\end{erratumbox}

We assume that the expressions being reduced have been erased (\Figureref{fig:er}),
so we omit a rule for reducing annotations.  Alternatives are discussed in
\Sectionref{sec:source-side-consistency}.

\subsection{\mksubsection{Value Restriction}}
\Label{sec:value-restriction}

Our calculus excludes effects such as mutable references; however, to allow it to serve
as a basis for larger languages, we impose a value restriction on
certain introduction rules.  Without this restriction, the system would be unsound
in the presence of mutable references.
Following \citet{Wright95},
the rule \Iallintro requires that its subject be a value, as in Standard ML~\citep{RevisedDefinitionOfStandardML}.
A similar value restriction is needed for intersection types~\citep{Davies00icfpIntersectionEffects}.  The following example shows the need for the restriction on $\alleosym$:
\vspace{-2pt}
\[
    \arrayenvbl{
        \xLet{r : \Refty{(\alleo{\eovar} \tau \garr{\eovar} \tau)}}%
             {\Refexp{f}}
        \\
        ~~~r \Gets g;
        ~~~h (\Deref{r})
    }
\vspace{-4pt}
\]
Assume we have $f : \alleo{\eovar} \tau{\garr{\eovar}}\tau$ and $g : \tau \narr \tau$
and $h : (\tau{\varr}\tau) \varr \tau$.
By a version of \Ialleointro that doesn't require its subject to be a value, we have
$r : \alleo{\eovar} \Refty{(\tau \garr{\eovar} \tau)}$.
By \Ialleoelim with $\N$ for $\eovar$, we have $r : \Refty{(\tau \narr \tau)}$, making the assignment $r \Gets g$ well-typed.
However, by \Ialleoelim with $\V$ for $\eovar$, we have
$r : \Refty{(\tau \varr \tau)}$.  It follows that the dereference $\Deref{r}$ has type $\tau \varr \tau$,
so $\Deref{r}$ can be passed to $h$.  But $\Deref{r} = g$ is actually call-by-name.
If $h = \Lam{x} x(e_2)$, we should be able to assume that $e_2$ will be evaluated
exactly once, but $x = g$ is call-by-name, violating this assumption.

If we think of $\alleosym$ as an intersection type, so that $r$ has type
$(\tau \varr \tau)
\sectty
(\tau \narr \tau)
$, the example and argument closely follow \citet{Davies00icfpIntersectionEffects}
and, in turn, \citet{Wright95}.  (For union types, a similar problem arises, which can be
solved by a dual solution---restricting the union-elimination rule to evaluation contexts~\citep{Dunfield03:IntersectionsUnionsCBV}.)

\subsection{\mksubsection{Subtyping and $\eta$-Expansion}}

Systems with intersection types often include subtyping.
The strength of subtyping in intersection type systems varies,
from syntactic approaches
that emphasize simplicity (\eg \citet{Dunfield03:IntersectionsUnionsCBV})
to semantic approaches
that emphasize completeness (\eg \citet{Frisch02}).
Generally, subtyping---at minimum---allows intersections to be transparently eliminated
even at higher rank (that is, to the left of an arrow), so that the following
function application is well-typed:
\[
   f : \big((\tau_1 \sectty \tau_1') \arr \tau_2\big) \arr \tau_3,
   \;
   g : (\tau_1 \arr \tau_2)
   \entails
   f\;g : \tau_3
\]
Through a subsumption rule, $g : (\tau_1 \arr \tau_2)$ checks against type
$(\tau_1 \sectty \tau_1') \arr \tau_2$, because a function that accepts all values of
type $\tau_1$ should also accept all values that have type $\tau_1$ \emph{and} type
$\tau_1'$.

Using the analogy between intersection and $\alleosym$, in our impartial type system,
we might expect to derive
\vspace{-1pt}
\[
   f : \big((\alleo{\eovar} \tau_1{\garr{\eovar}}\tau_1){\varr}\tau_2\big) \varr \tau_3,
   \;
   g : (\tau_1{\garr{\N}}\tau_1) \varr \tau_2
   \entails
   f\;g : \tau_3
\vspace{-2pt}
\]
Here, $f$ asks for a function of type
$\big(\alleo{\eovar} \tau_1{\garr{\eovar}}\tau_1) \varr \tau_2\big)$,
which works on all evaluation orders; but $g$'s type $(\tau_1 \garr{\N} \tau_1) \varr \tau_2$
says that $g$ calls its argument only by name.

For simplicity, this paper excludes subtyping:
our type system does not permit this derivation.  But it would be possible
to define a subtyping system, and incorporate subtyping into the subsumption
rule \Isub---either by treating $\alleosym$ similarly to $\AllSym$ \citep{Dunfield13},
or by treating $\alleosym$ as an intersection type \citep{Dunfield03:IntersectionsUnionsCBV}.
A simple subtyping system could be derived from the typing rules that are \emph{stationary}---where the premises
type the same expression as the conclusion~\citep{Leivant86}.
For example, \Ialleoelim corresponds to
\[
   \Infer{${\subtype}{\alleosym}$-Left}
        {\Gamma \entails \epsilon \eo}
        {\Gamma \entails (\alleo{\eovar} \tau) \subtype [\epsilon/\eovar]\tau}
\]
Alternatively, $\eta$-expansion can substitute for subtyping: even without subtyping
and a subsumption rule, we can derive
\[
\arrayenvl{
   f : \big((\alleo{\eovar} \tau_1{\garr{\eovar}}\tau_1) \arr \tau_2\big) \arr \tau_3,
   \\
   ~~~g : (\tau_1 \garr{\N} \tau_1) \arr \tau_2
   ~\entails~
   f\;(\Lam{x} g\;x) : \tau_3
}
\]
This idea, developed by \citet{Barendregt83}, can be
automated; see, for example, \citet{Dunfield14}.

\section{Economical Type System}
\Label{sec:econ}

\begin{figure}[thbp]
  \centering

  \judgbox{\econtrans{\tau} = S}
        {Impartial type $\tau$ translates to economical type $S$}
  \hspace*{-6.5ex}
  \begin{tabular}[t]{r@{~~}c@{~~}l}
      $\econtrans{\unitty}$  &$=$&   $\unitty$
      \\
      $\econtrans{\tau_1 \garr{\epsilon} \tau_2}$
         &$=$&
         $\left(\fighi{\susp{\epsilon}} \econtrans{\tau_1}\right) \arr \econtrans{\tau_2}$
         \hspace{-2ex}
      \\
      $\econtrans{\tau_1 +^\epsilon \tau_2}$
         &$=$&
         $\fighi{\susp{\epsilon}} \left(\econtrans{\tau_1} + \econtrans{\tau_2}\right)$
         \hspace{-2ex}
      \\
      $\econtrans{\tau_1 *^\epsilon \tau_2}$
         &$=$&
         $\left(\fighi{\susp{\epsilon}} \econtrans{\tau_1}\right) * \left(\fighi{\susp{\epsilon}} \econtrans{\tau_2}\right)$
         \hspace{-4ex}
  \end{tabular}
  ~\hspace{0ex}
  \begin{tabular}[t]{r@{~~}c@{~~}l}
      $\econtrans{\alleo{\eovar} \tau}$
          &$=$&
          $\alleo{\eovar} \econtrans{\tau}$
      \\
      $\econtrans{\rec{\epsilon}{\alpha} \tau}$  &$=$&   $\Rec{\alpha} \fighi{\susp{\epsilon}} \econtrans{\tau}$
      \\
      $\econtrans{\All{\alpha} \tau}$  &$=$&   $\All{\alpha} \econtrans{\tau}$
      \\
      $\econtrans{\alpha}$  &$=$&   $\alpha$
  \end{tabular}\hspace{-6ex}

  \medskip

  \judgbox{\ctxecontrans{\gamma} = \Gamma}
          {Impartial context $\gamma$ translates to economical context $\Gamma$}
  \vspace{-0.5ex}
  \hspace*{-7ex}
  \begin{tabular}[t]{r@{~}c@{~}lll}
      $\ctxecontrans{\cdot}$  &$=$&   $\cdot$
      \\
      $\ctxecontrans{\gamma, \alpha \type}$  &$=$&   $\ctxecontrans{\gamma}, \alpha \type$
      \\
      $\ctxecontrans{\gamma, \xsyn u \NONVAL \tau}$
          &$=$&
          $\ctxecontrans{\gamma}, \xtypeoft{u}{\fighi{\econtrans \tau}}$
  \end{tabular}\hspace{-2ex}
  \begin{tabular}[t]{r@{~}c@{~}lll}
      $\ctxecontrans{\gamma, \eovar \eo}$  &$=$& $\ctxecontrans{\gamma}, \eovar \eo$
     \\%
      $\ctxecontrans{\gamma, \xsyn x \VAL \tau}$
          &$=$&
          $\ctxecontrans{\gamma}, \xtypeoft{x}{\fighi{\susp{\V} \econtrans{\tau}}}$
      \\
      $\ctxecontrans{\gamma, \xsyn x \NONVAL \tau}$
          &$=$&
          $\ctxecontrans{\gamma}, \xtypeoft{x}{\fighi{\susp{\N} \econtrans{\tau}}}$
  \end{tabular}\hspace{-6ex}

  \medskip

  \judgbox{\expecontrans{e} = e'}
          {Expression $e$ with $\tau$-annotations
            \\
            ~~translates to expression $e'$ with $S$-annotations}
  \begin{tabular}[t]{r@{~~}c@{~~}lll}
      $\expecontrans{\Anno{e}{\tau}}$  &$=$&   $\Anno{\expecontrans{e}}{\fighi{\econtrans{\tau}}}$
      \\
      $\expecontrans{\tyapp{e}{\tau}}$  &$=$&   $\tyapp{\expecontrans{e}}{\fighi{\econtrans{\tau}}}$
      \\
      $\expecontrans{\expapp{e_1}{e_2}}$  &$=$&   $\expapp{\expecontrans{e_1}}{\expecontrans{e_2}}$ 
      \\
      etc.
  \end{tabular}

  \caption{Type translation into the economical language}
  \FLabel{fig:econtrans}
\end{figure}

\begin{figure*}[htbp]
  $~$\!\!\!\!\!\begin{bnfarray}
    \text{\small Economical types}
    & S
    &\bnfas&
       \unitty
       \bnfalt
       \alpha
       \bnfalt
       \All{\alpha} S
       \bnfalt
       \alleo{\eovar} S
       \bnfalt
       \fighi{\susp{\epsilon} S}
       \bnfaltBRK
       S_1 \arr S_2
       \bnfalt
       S_1 * S_2
       \bnfalt
       S_1 + S_2
       \bnfalt
       \Rec{\alpha} S
  \end{bnfarray}
  \!\!\!\!\!
  \begin{bnfarray}
    \text{\small Econ.\ typing contexts}
    & \Gamma
    &\bnfas&
       \cdot
       \bnfalt
       \Gamma, \xtypeof{x}{\fighi{S}}
       \bnfalt
       \Gamma, \xtypeof{u}{\fighi{S}}
       \bnfalt
       \Gamma, \eovar \eo
       \bnfalt
       \Gamma, \alpha \type
    \\
    \text{\small Econ.\ source expressions}
    & e
    &\bnfas&
        \dots
    \bnfalt
        \tylam{\alpha} e
        \bnfalt
        \tyapp{e}{\fighi{S}}
    \bnfalt
        \Anno{e}{\fighi{S}}
  \end{bnfarray}

  \smallskip

  \centering
      \judgbox{\arrayenvcl{\chktypee{\Gamma}{e}{\VVAR}{S}
                      \\
                      \syntypee{\Gamma}{e}{\VVAR}{S}}}
                   {Source expression $e$ checks against economical type $S$
                    \\[0.5ex]
                    Source expression $e$ synthesizes economical type $S$}
      \vspace{-1.5ex}
      \begin{mathpar}
      \Infer{\Rvar}
            {(\xtypeofe{x}{S}) \in \Gamma}
            {\syntypee{\Gamma}{x}{\VAL}{S}}
      ~~~~
      \Infer{\Rfixvar}
            {(\xtypeofe{u}{S}) \in \Gamma}
            {\syntypee{\Gamma}{u}{\NONVAL}{S}}
      \and
      \Infer{\Rfix}
            {\chktypee{\Gamma, \xtypeofe{u}{S}}{e}{\VVAR}{S}}
            {\chktypee{\Gamma}{(\Fix{u} e)}{\NONVAL}{S}}
   \and
      \Infer{\Rsub}
            {\syntypee{\Gamma}{e}{\VVAR}{S}}
            {\chktypee{\Gamma}{e}{\VVAR}{S}}
      ~~~~
      \Infer{\Ranno}
            {\chktypee{\Gamma}{e}{\VVAR}{S}}
            {\syntypee{\Gamma}{\Anno{e}{S}}{\VVAR}{S}}
    \\
\RuleHead{\AllSym}
\Infer{\Rallintro}
      { 
        \chktypee{\Gamma, \alpha \type}{e}{\VAL}{S}
      }
      {\chktypee{\Gamma}{\tylam{\alpha} e}{\VAL}{\All{\alpha} S}}
\and
\Infer{\Rallelim}
     {\syntypee{\Gamma}{e}{\VVAR}{\All{\alpha} S}
      \\
      \Gamma \entails S' \type}
     {\syntypee{\Gamma}{\tyapp{e}{S'}}{\VVAR}{[S' / \alpha]S}}
\and
\RuleHead{\unitty}\hspace*{-1.5ex}
   \Infer{\Runitintro}
             {}
             {\chktypee{\Gamma}{\unit}{\VAL}{\unitty}}
   \\
\RuleHead{\alleosym}
\Infer{\Ralleointro}
      {
        \chktypee{\Gamma, \eovar \eo}{e}{\VAL}{S}
      }
      {\chktypee{\Gamma}{e}{\VAL}{\alleo{\eovar} S}}
\and
\Infer{\Ralleoelim}
     {\syntypee{\Gamma}{e}{\VVAR}{\alleo{\eovar} S}
       \\
       \Gamma \entails \epsilon \eo}
     {\syntypee{\Gamma}{e}{\VVAR}{[\epsilon / \eovar]S}}
\\
\RuleHead{\susp{\epsilon}}
\Infer{\Rsuspintro}
      {\chktypee{\Gamma}{e}{\VVAR}{S}
      }
      {\arrayenvbl{
        \chktypee{\Gamma}{e}{\VVAR}{\susp{\epsilon} S}
        \\
        \chktypee{\Gamma}{e}{\VAL}{\susp{\N} S}
      }}
~~~~~~~~~~~~~
\Infer{\Rsuspelim{\V}}
     {\syntypee{\Gamma}{e}{\VVAR}{\susp{\V} S}}
     {\syntypee{\Gamma}{e}{\VVAR}{S}}
~~~~~~~
\Infer{\Rsuspelim{\epsilon}}
     {\syntypee{\Gamma}{e}{\VVAR}{\susp{\epsilon} S}}
     {\syntypee{\Gamma}{e}{\NONVAL}{S}}
   \\
    \RuleHead{\arr}
       \Infer{\Rarrintro}
             {\chktypee{\Gamma, x:S_1}{e}{\VVAR}{S_2}
             }
             {\chktypee{\Gamma}{(\explam{x} e)}{\VAL}{(S_1 \arr S_2)}}
       \and
       \Infer{\Rarrelim}
              {\syntypee{\Gamma}{e_1}{\VVAR_1}{(S_1 \arr S_2)}
               \\
               \chktypee{\Gamma}{e_2}{\VVAR_2}{S_1}
              }
              {\syntypee{\Gamma}{(\expapp{e_1}{e_2})}{\NONVAL}{S_2}
              }
    \\
    \RuleHead{*}
          \Infer{\Rprodintro}
                 {\chktypee{\Gamma}{e_1}{\VVAR_1}{S_1}
                   \\
                   \chktypee{\Gamma}{e_2}{\VVAR_2}{S_2}
                  }
                  {\chktypee{\Gamma}{\Pair{e_1}{e_2}}{\VVAR_1 \join \VVAR_2}{(S_1 * S_2)}}
           \and
           \Infer{\Rprodelim{k}}
                 {
                   \syntypee{\Gamma}
                           {e}
                           {\VVAR}
                           {(S_1 * S_2)}
                 }
                 {
                   \syntypee{\Gamma}
                        {(\Proj{k} e)}
                        {\NONVAL}
                        {S_k}
                 }
    \\
    \RuleHead{+}
       \Infer{\Rsumintro{k}}
             {\chktypee{\Gamma}{e}{\VVAR}{S_k}
             }
             {\chktypee{\Gamma}{(\Inj{k} e)}{\VVAR}{(S_1 + S_2)}}
       \and
       \Infer{\Rsumelim}
              {   \syntypee{\Gamma}{e}{\VVAR_0}{(S_1 + S_2)}
                  \\
                  \arrayenvbl{
                      \chktypee{\Gamma, x_1 : S_1}{e_1}{\VVAR_1}{S}
                      \\
                      \chktypee{\Gamma, x_2 : S_2}{e_2}{\VVAR_2}{S}
                  }
              }
              {\chktypee{\Gamma}{\Casesum{e}{x_1}{e_1}{x_2}{e_2}}{\NONVAL}{S}
              }
    \\
    \RuleHead{\recsymbol}
       \Infer{\Rrecintro}
             {\chktypee{\Gamma}{e}{\VVAR}{\big[(\Rec{\alpha}S)\big/\alpha\big]S}
             }
             {\chktypee{\Gamma}{e}{\VVAR}{\Rec{\alpha} S}}
       \and
       \Infer{\Rrecelim}
              {   \syntypee{\Gamma}{e}{\VVAR}{\Rec{\alpha} S}
              }
              {
                \syntypee{\Gamma}{e}{\NONVAL}{\big[(\Rec{\alpha}S)\big/\alpha\big]S}
              }
    \end{mathpar}

  \caption{Economical bidirectional typing}
  \FLabel{fig:econ-bi}
\end{figure*}

The impartial type system directly generalizes a call-by-value system and
a call-by-name system, but the profusion of connectives is unwieldy, and
impartiality doesn't fit a standard operational semantics.
Instead of elaborating the impartial system into our target language,
we pause to develop an \emph{economical} type system whose
standard connectives ($\arr$, $*$, $+$, $\mu$) are by-value,
but with a \emph{suspension point}
$\susp{\epsilon} S$ to provide by-name behaviour.
This intermediate system yields a straightforward elaboration.
It also constitutes an alternative source language that, while biased
towards call-by-value, conveniently allows call-by-name and evaluation-order
polymorphism.

In the grammar in \Figureref{fig:econ-bi}, the economical types $S$ are
obtained from the impartial types $\tau$
by dropping all the $\epsilon$ decorations and adding a connective
$\susp{\epsilon} S$ (read ``$\epsilon$ suspend $S$'').  When $\epsilon$ is $\V$,
this connective is a no-op: elaborating $e$ at type $\susp{\V} S$
and at type $S$ yield the same term.
But when $\epsilon$ is $\N$, elaborating $e$ at type $\susp{\N} S$
is like elaborating $e$ at type $\unitty \arr S$.

In economical typing contexts $\Gamma$, variables $x$ denote
values, so we replace the assumption form $\xsyn x \VVAR \tau$
with $x : S$.  Similarly, we replace $\xsyn u \NONVAL \tau$ with $u : S$.

Dropping $\epsilon$ decorations means that---apart from the valueness annotations---most
of the economical rules in \Figureref{fig:econ-bi} look fairly standard.
The only new rules are for suspension points $\susp{\epsilon}$, halfway
down \Figureref{fig:econ-bi}.
It would be nice to have only two rules (an introduction and an elimination),
but we need to track whether $e$ is a value, which depends on the $\epsilon$
in $\susp{\epsilon} S$: if we introduce the type $\susp{\N} S$, then $e$ will be elaborated
to a thunk, which is a value; if we are eliminating $\susp{\N} S$, the elaboration of $e$ will
have the form $\Force \cdots$, which (like function application) is not a value.

\subsection{\mksubsection{Translating to Economical Types}}

To relate economical types to impartial types, we define a type translation
$\econtrans{\tau} = S$ that inserts suspension points (\Figureref{fig:econtrans}).
Given an impartially-typed source program $e$ of type $\tau$, we can show
that $\expecontrans{e}$ has the economical type $\econtrans{\tau}$
(\Theoremref{thm:econ}).

Some parts of the translation are straightforward.
Functions  $\tau_1 \garr{\epsilon} \tau_2$
are translated to
$(\susp{\epsilon} \econtrans{\tau_1}) \arr \econtrans{\tau_2}$
because when $\epsilon = \N$, we get the expected type
$(\susp{\N} \econtrans{\tau_1}) \arr \econtrans{\tau_2}$
of a call-by-name function.

We are less constrained in how to translate other connectives:

\begin{itemize}
\item 
    We could translate $\tau_1 +^\epsilon \tau_2$ to
    $(\susp{\epsilon} \econtrans{\tau_1}) + (\susp{\epsilon} \econtrans{\tau_2})$.
    But then $\unitty +^\N \unitty$---presumably intended as
    a non-strict boolean type---would be translated to
    $(\susp{\N} \unitty) + (\susp{\N} \unitty)$, which exposes
    which injection was used (whether the boolean is true or false)
    without forcing the (spurious) thunk around the unit value.
    Thus, we instead place the thunk around the entire sum, so that
    $\unitty +^\N \unitty$ translates to $\susp{\N} (\unitty + \unitty)$.

\item 
    We could translate $\tau_1 *^\epsilon \tau_2$ to
    $\susp{\epsilon} (\econtrans{\tau_1} * \econtrans{\tau_2})$---which
    corresponds to how we decided to translate sum types.
    Instead, we translate it to
    $ (\susp{\epsilon} \econtrans{\tau_1})
    * (\susp{\epsilon} \econtrans{\tau_2})$,
    so that, when $\epsilon = \N$, we get a pair of thunks;
    accessing one component of the pair
    (by forcing its thunk) won't cause the other component to be forced.

\item
    Finally, in translating $\rec{\epsilon}{\alpha} \tau$,
    we could put a suspension on each occurrence of $\alpha$ in $\tau$,
    rather than a single suspension on the outside of $\tau$.
    Since $\tau$ is often a sum type, writing $+^\epsilon$ already
    puts a thunk on $\tau$; we don't need a thunk around a thunk.
    But by the same token, suspensions around the occurrences of $\alpha$
    can also lead to double thunks:
    translating the type of lazy natural numbers
    $\rec{\N}{\alpha} (\unitty +^\N \alpha)$
    would give $\Rec{\alpha} \big(\susp{\N} (\unitty + \susp{\N}\alpha)\big)$,
    which expands to $\susp{\N} \big(\unitty + \susp{\N} \susp{\N} (\unitty + \dots)\big)$.
\end{itemize}

The rationales for our translation of products and recursive types are
less clear than the rationale for sum types; it's possible
that different encodings would be preferred in practice.

The above translation does allow programmers to use the
alternative encodings, though awkwardly.  For example,
a two-thunk variant of $\tau_1 *^\epsilon \tau_2$ can be obtained
by writing $(\rec{\epsilon}{\beta}\tau_1) *^\V (\rec{\epsilon}{\beta}\tau_2)$,
where $\beta$ doesn't occur; the only purpose of $\mu$ here is to
insert a suspension.  (This suggests a kind of ill-founded argument for
our chosen translation of $\recsymbol$: it enables us to insert suspensions,
albeit awkwardly.)

\subsection{\mksubsection{Programming with Economical Types}}
\Label{sec:econ-examples}

We can translate the list/stream example from \Sectionref{sec:impartial-examples}
to the economical system:
\[
   \textkw{type}~\tyname{List}\;\eovar\;\alpha
   ~=~
         \Rec{\beta}
             \,\susp{\eovar}
                 \big(
                     \unitty
                  +
                     (
                        \alpha * \beta
                     )
                 \big)
\]
The body of \textfn{map} is the same; only the type annotation is different.
\begin{mathdispl}
  \arrayenvbl{
      \textfn{map}
      :
      \alleo{\eovar}
         \All{\alpha}
             (\alpha \arr \beta)
             \arr
             (\tyname{List}\;\eovar\;\alpha)
             \arr
             (\tyname{List}\;\eovar\;\beta)
      \\
      =~
          \tylam{\alpha}
              \Fix{map}
                  \Lam{f}
                      \Lam{xs}
      \\
      ~~~~~~~~
      \Casesumstart{xs}{x_1}{\Inj{1} \unit}
      \\~~~~~~~~~~~~~~~~~~~~~~~~
       \Casesumend{x_2}{%
            \Inj{2} \Pair{\expapp{f}{(\Proj{1} x_2)}}
                                   {\expapp{\expapp{map}{f}}{(\Proj{2} x_2)}}
      }
  }
\end{mathdispl}
The above type for \textfn{map} corresponds to the impartial type with $\garr{\V}$.
At the end of \Sectionref{sec:impartial-examples}, we gave
a very generic type for \textfn{map}, which we can translate
to the economical system:
\begin{mathdispl}
  \arrayenvcl{
      \alleo{\eovar_1,
              \eovar_2,
              \eovar_3,
              \eovar_4,
              \eovar_5}
        \\
        ~~ \All{\alpha}
           \Big(
             \susp{\eovar_2}
             \big((\susp{\eovar_1} \alpha) \arr \beta\big)
           \Big)
           \arr
           \big (\susp{\eovar_4} (\tyname{List}\;\eovar_3\;\alpha) \big)
           \arr
           (\tyname{List}\;\eovar_5\;\beta)
   }
\end{mathdispl}
This type might not look economical, but makes redundant
suspensions more evident: $\tyname{List}\;\eovar_3\;\alpha$ is
$\Rec{\cdots}\,\susp{\eovar_3} \cdots$, so the suspension
controlled by $\eovar_4$ is never useful, showing that $\eovar_4$ is unnecessary.

\subsection{\mksubsection{Economizing}}

The main result of this section is that impartial typing derivations can be transformed
into economical typing derivations.
The proof~\citep[Appendix B.3]{Dunfield15arxiv} relies on a lemma
that converts typing assumptions with $\susp{\V} S'$
to assumptions with $S'$.

\begin{restatable}[Economizing]{theorem}{thmecon}
\Label{thm:econ}
~
\begin{enumerate}[(1)]
  \item If $\syntype{\gamma}{e}{\VVAR}{\tau}$
     then $\syntypee{\ctxecontrans{\gamma}}{\expecontrans{e}}{\VVAR}{\econtrans{\tau}}$.

  \item If $\chktype{\gamma}{e}{\VVAR}{\tau}$
     then $\chktypee{\ctxecontrans{\gamma}}{\expecontrans{e}}{\VVAR}{\econtrans{\tau}}$.
\end{enumerate}
\end{restatable}

\section{Target Language}
\Label{sec:target}

\begin{figure}[htbp]
  \centering

  \smallskip

  \begin{bnfarray}
    \text{Target terms}
    & M%
    &\bnfas&
        \unit
        \bnfalt
        x
        \bnfalt
        \tlam{x} M
        \bnfalt
        M_1\,M_2
        \bnfaltBRK
        u
        \bnfalt
        \tfix{u} M
        \bnfalt%
        \ttylam M   \bnfalt   \ttyapp{M}
        \bnfaltBRK
        \Thunk M   \bnfalt   \Force M
        \bnfaltBRK
        \tpair{M_1}{M_2}
        \bnfalt
        \tproj{k}{M}
        \bnfaltBRK
        \tinj{k} M
        \bnfalt
        \tcase{M}{x_1}{M_1}{x_2}{M_2}
        \bnfaltBRK
        \troll M
        \bnfalt
        \tunroll{M}
  \\[1ex]
    \text{Values}
    & W
    &\bnfas&
        \unit
        \bnfalt
        x
        \bnfalt
        \tlam{x} M
        \bnfalt%
        \ttylam M
        \bnfaltBRK
        \Thunk M
        \bnfalt%
        \tpair{W_1}{W_2}
        \bnfaltBRK
        \tinj{k} W
        \bnfalt%
        \troll W
  \\[1ex]
    \text{Valuables}
    & \Vable
    &\bnfas&
        \unit
        \bnfalt
        x
        \bnfalt
        \tlam{x} M
        \bnfalt%
        \ttylam \Vable  \bnfalt   \ttyapp \Vable
        \bnfaltBRK
        \Thunk M
        \bnfalt%
        \tpair{\Vable_1}{\Vable_2}
        \bnfaltBRK
        \tproj{k}{\Vable}
        \bnfalt%
        \tinj{k} \Vable
        \bnfalt%
        \troll \Vable
        \bnfalt
        \tunroll \Vable
  \\[1ex]
    \text{Eval.\ contexts}
    & \E
    &\bnfas&
      \hole
      \bnfalt%
      \E \vapp M_2
      \bnfalt
      W_1 \vapp \E
      \bnfalt%
      \ttyapp{\E}
      \bnfalt%
      \Force\,{\E}
      \bnfaltBRK
      \tpair{\E}{M_2}
      \bnfalt
      \tpair{W_1}{\E}
      \bnfalt
      \tproj{k}{\E}
      \bnfaltBRK
      \tinj{k}{\E}
      \bnfalt
      \tcase{\E}{x_1}{M_1}{x_2}{M_2} %
      \bnfaltBRK
      \troll{\E}
      \bnfalt
      \tunroll{\E}
  \end{bnfarray}
   
  \smallskip

  \begin{bnfarray}
    \text{Target types}
    & A, B
    &\bnfas&
       \unitty
       \bnfalt
       \alpha
       \bnfalt
       \All{\alpha} A
       \bnfalt
       A_1 \arr A_2
       \bnfalt%
       \fighi{\thunkty A_1}
       \bnfaltBRK
       A_1 * A_2
       \bnfalt
       A_1 + A_2
       \bnfalt
       \Rec{\alpha} A
    \\[0.3ex]
    \text{Typing contexts}
    & G
    &\bnfas&
       \cdot
       \bnfalt
       G, \xtypeoft{x}{A}
       \bnfalt
       G, \alpha \type
  \end{bnfarray}
\caption{Syntax of the target language}
\FLabel{fig:target-syntax}
\end{figure}

Our target language (\Figureref{fig:target-syntax})
has by-value $\arr$, $*$, $+$ and $\mu$ connectives,
$\AllSym$, and a $\xU$ connective (for thunks).

The $\AllSym$ connective has explicit introduction and elimination
  forms $\ttylam M$ and $\ttyapp M$.  This ``type-free'' style
  is a compromise between having no explicit forms for $\forall$
  and having explicit forms that contain types
  ($\LAM{\alpha} M$ and $\tyapp{A}{M}$).  Having no explicit forms would
  complicate some proofs; including the types would mean that target terms
  contain types, giving a misleading impression that operational behaviour is influenced
  by types.
  
  The target language also has an explicit introduction form $\troll{M}$
  and elimination form $\tunroll{M}$ for $\mu$ types.

 As with $\AllSym$, we distinguish thunks to simplify some proofs:
  Source expressions typed with the $\susp{\N}$ connective
  are elaborated to $\Thunk M$, rather than to a $\lambda$ with an unused
  bound variable.  Dually, eliminating $\susp{\N}$ results in a target term
  $\Force M$, rather than to $M\unit$.

\begin{figure}[t]
  \centering
      \judgbox{\typeoft{G}{M}{A}}
                   {Target term $M$ \\ has target type $A$}
      \vspace{-8.5ex}
      \begin{mathpar}
          \hspace*{30ex}
          \Infer{\Munitintro}
                {}
                {\typeoft{G}{\tunit}{\unitty}
                }
          \\
          \Infer{\!\Mvar}
              {(\xtypeoft x A) \in G}
              {\typeoft{G}{x}{A}}
          ~~~~
          \Infer{\!\Mfixvar}
              {(\xtypeoft u A) \in G}
              {\typeoft{G}{u}{A}}
          ~~
          \Infer{\!\Mfix}
              {\typeoft{G, \xtypeoft{u}{A}}{e}{A}}
              {\typeoft{G}{(\tfix{u} e)}{A}}
          \\
          \RuleHead{\AllSym}\hspace*{-2.0ex}
             \Infer{\!\!\Mallintro}
                   {\typeoft{G, \alpha \type}{\Vable}{A}
                   }
                   {\typeoft{G}{\ttylam \Vable}{\All{\alpha} A}}
             ~~
             \Infer{\!\!\Mallelim}
                    {
                      \arrayenvbl{
                        \typeoft{G}{M}{\All{\alpha} A}
                        \\
                        G \entails A' \type
                      }
                    }
                    {\typeoft{G}{\ttyapp{M}}{[A'/\alpha]A}
                    }
          \\
          \hspace*{-0.5ex}\RuleHead{\arr}\hspace*{-1.5ex}
          \Infer{\!\Marrintro}
                {\typeoft{G, \xtypeoft{x}{A}}
                          {M}
                          {B}
                }
                {
                  \typeoft{G}
                       {(\tlam{x} M)
                       }
                       {A{\arr}B}
                }
          ~
          \Infer{\!\Marrelim}
                {\arrayenvbl{
                    \typeoft{G}{M_1}{A \arr B}
                    \\
                    \typeoft{G}{M_2}{A}
                  }
                }
                {\typeoft{G}
                          {(M_1 \, M_2)}
                          {B}
                }
        \\
        \hspace*{-0.5ex}
          \RuleHead{\thunkty}\hspace*{-3.0ex}
          \Infer{\!\Mthunkintro}
                {\typeoft{G}
                          {M}
                          {B}
                }
                {
                  \typeoft{G}
                       {\Thunk M
                       }
                       {\thunkty B}
                }
          ~~~
          \Infer{\!\Mthunkelim}
                {\typeoft{G}{M_1}{\thunkty B}
                }
                {\typeoft{G}
                          {\Force M_1}
                          {B}
                }
          \\
          \RuleHead{*}\hspace*{-2.0ex}
          \Infer{\!\!\Mprodintro}
                 {\arrayenvbl{
                     \typeoft{G}{M_1}{A_1}
                     \\
                     \typeoft{G}{M_2}{A_2}
                   }
                  }
                  {\typeoft{G}{\!\tpair{M_1}{\!M_2}\!}{\!A_1{*}A_2}}
           ~
           \Infer{\!\!\Mprodelim{k}}
                 {
                   \typeoft{G}
                           {M}
                           {A_1{*}A_2}
                 }
                 {
                   \typeoft{G}
                        {\tproj{k} M}
                        {A_k}
                 }
          \\
          \RuleHead{+}\hspace*{-2ex}
                       {\text{\small$\Infer{\!\!\Msumintro{k}}
                 {\typeoft{G}{M}{A_k}
                  }
                  {\typeoft{G}{\tinj{k}\!M}{A_1{+}A_2}}
                  $}}
          \hspace{30ex}
           \vspace{-6ex}
           \\
           ~\hspace{21ex}
           {\text{\small$
           \Infer{\!\!\Msumelim}
                 {
                   \hspace{13ex}\arrayenvbl{
                   \typeoft{G}
                           {M}
                           {A_1{+}A_2}
                   \\
                       \typeoft{G, x_1{:}A_1}
                               {M_1\!}
                               {A}
                       \\
                       \typeoft{G, x_2{:}A_2}
                               {M_2\!}
                               {A}
                   }
                 }
                 {
                   \typeoft{G}
                        {\tcase{M}{x_1}{M_1}{x_2}{M_2}}
                        {A}
                 }
           $}}
         \vspace{-0.5ex}
          \\
          \mbox{\raisebox{2.0ex}{$\RuleHead{\recsymbol}$}}
          \hspace{1.03\columnwidth}
          \vspace{-3ex}
          \\
          \hspace*{-0.2ex}\Infer{\!\!\!\Mrecintro}
                   {\typeoft{G}{M}{[\Rec{\alpha}A/\alpha]A}
                   }
                   {\typeoft{G}{\troll M}{\Rec{\alpha} A}}
             ~
             {\text{\small$
             \Infer{\!\!\!\Mrecelim}
                    {   \typeoft{G}{M}{\Rec{\alpha} A}
                    }
                    {\typeoft{G}{\tunroll M}{[\Rec{\alpha}A/\alpha]A}}
             $}}
      \end{mathpar}

  \caption{Target language type system}
  \FLabel{fig:target-types}
\end{figure}

\subsection{\mksubsection{Typing Rules}}

\Figureref{fig:target-types} shows the typing rules for our target language.
These are standard except for the \Mallintro rule and the rules for thunks:

\paragraph{Valuability restriction.}
Though we omit mutable references from the target language, we want
the type system to accommodate them.  Using the standard syntactic
value restriction~\citep{Wright95} would spoil this language as a target
for our elaboration: when source typing uses \Eallintro, it requires that
the source expression be a value (not syntactically, but according to the
source typing derivation).  Yet if that source value is typed using
\Ealleoelim, it will elaborate to a projection, which is not a syntactic value.
So we use a valuability restriction in \Mallintro.  A target term
is a \emph{valuable} $\Vable$ if it is a value
(\eg $\tlam{x} M$) or is a projection, injection, roll or unroll of something
that is valuable (\Figureref{fig:target-syntax}).  Later, we'll prove that if a source
expression is a value (according to the source typing derivation),
its elaboration is valuable (Lemma \ref{lem:elab-valuability}).

\paragraph{Thunks.}
We give $\Thunk M$ the type $\thunkty B$ for ``th$\xU$nk $B$'' (if $M$ has
type $B$); $\Force M$ eliminates this connective.

\subsection{\mksubsection{Operational Semantics}}
\Label{sec:target-opsem}

\begin{figure}[t]
  \centering

\judgbox{M \step M'}
        {Target term $M$ steps (by-value) to target term $M'$}
\vspace{-3.2ex}
  \begin{mathpar}
     \Infer{\StepContext}
           {M \stepR M'}
           {\E[M] \step \E[M']}
  \end{mathpar}

\judgbox{M \stepR M'}
        {Target redex $M$ reduces (by-value) to $M'$}
\vspace{-0.1ex}

~\!\!\!\!\begin{array}[t]{r@{~~}c@{~~}l@{~~~}ll}
  (\tlam{x} M) \vapp W &\stepR& [W/x]M
      & \RedBeta
  \\[0.3ex]
  \Force (\Thunk M)
         &\stepR&
         M
      & \RedForce
  \\[0.3ex]
  (\tfix{u} M) &\stepR& \big[(\tfix{u} M) \big/ u\big] M
      & \RedFix
  \\[0.3ex]
  \ttyapp{(\ttylam M)}
         &\stepR&
         M
      & \RedTyapp
  \\[0.3ex]
  \tproj{k}{(\tpair{W_1}{W_2})}
         &\stepR&
         W_k 
      & \RedProj
\end{array}

~\!\!\!\!\begin{array}[t]{r@{~~}c@{~~}l@{~~~}ll}
  \tcase{\tinj{k} W}{x_1}{M_1}{x_2}{M_2}
               &\stepR&
               [W/x_k]M_k
      &\RedCase
  \\[0.3ex]
  \tunroll{(\troll{W})}
         &\stepR&
         W
      & \RedUnroll  
\end{array}

  \caption{Target language operational semantics}
  \FLabel{fig:cbv-opsem}
\end{figure}

The target operational semantics has two relations:
$M \stepR M'$, read ``$M$ reduces to $M'$'',
and $M \step M'$, read ``$M$ steps to $M'$''.
The latter has only one rule, \StepContext, which says that
$\E[M] \step \E[M']$ if $M \stepR M'$, where $\E$ is
an evaluation context (\Figureref{fig:target-syntax}).
The rules for $\stepR$ (\Figureref{fig:cbv-opsem}) reduce
a $\lambda$ applied to a value; a force of a thunk;
a fixed point;
a type application;
a projection of a pair of values;
a case over an
injected value; and an unroll of a rolled value.
Apart from $\Force{(\Thunk M)}$, which we can view
as strange syntax for $(\Lam{x} M)\tunit$,
this is all standard: these definitions
use values $W$, not valuables $\Vable$.

\subsection{\mksubsection{Type Safety}}

\begin{lemma}[Valuability]
\Label{lem:valuability}
  If $\Vable \step M'$ or $\Vable \stepR M'$
  then $M'$ is valuable, that is, there exists $\Vable' = M'$.
\end{lemma}

\begin{lemma}[Substitution]
\Label{lem:target-subst}
  If $\typeoft{G, x : A', G'}{M}{A}$
  and $\typeoft{G}{W}{A'}$
  then
  $\typeoft{G, G'}{[W/x]M}{A}$.
\end{lemma}

\begin{theorem}[Type safety]
\Label{thm:target-type-safety}
~~%
  If $\typeoft{\cdot}{M}{A}$
  then either $M$ is a value,
  or $M \step M'$ and $\typeoft{G}{M'}{A}$.
\end{theorem}
\vspace{-2ex}
\begin{proof}
  By induction on the derivation of $\typeoft{G}{M}{A}$,
  using \Lemmaref{lem:target-subst} and
  standard inversion lemmas, which we omit.
\end{proof}

\section{Elaboration}
\Label{sec:elab}

\begin{figure}[t]
  \centering
  \judgbox{\tytrans{S} = A}
        {Economical type $S$ elaborates to target type $A$}
  \begin{tabular}[t]{r@{~~}c@{~~}lll}
      $\tytrans{\unitty}$  &$=$&   $\unitty$
      \\
      $\tytrans{S_1 \arr S_2}$  &$=$&   $\tytrans{S_1} \arr \tytrans{S_2}$
      \\
      $\tytrans{S_1 + S_2}$  &$=$&   $\tytrans{S_1} + \tytrans{S_2}$
      \\
      $\tytrans{\alpha}$  &$=$&   $\alpha$
      \\
      $\tytrans{\All{\alpha} S}$  &$=$&   $\All{\alpha} \tytrans{S}$
   \end{tabular}
   ~
   \begin{tabular}[t]{r@{~~}c@{~~}lll}
      $\tytrans{\susp{\V} S}$  &$=$&   $\tytrans{S}$
      \\
      $\tytrans{\susp{\N} S}$  &$=$&   $\fighi{\thunkty} \tytrans{S}$
      \\[0.3ex]
      $\tytrans{\alleo{\eovar} S}$  &$=$&   $\fighi{\tytrans{[\V/\eovar]S} * \tytrans{[\N/\eovar]S}}$
      \\[0.3ex]
      $\tytrans{\Rec{\alpha} S}$  &$=$&   $\Rec{\alpha} \tytrans{S}$
   \end{tabular}

  \medskip

  \judgbox{\ctxtrans{\Gamma} = G}
        {Economical typing context $\Gamma$
          \\ ~~elaborates to target typing context $G$}
  \begin{tabular}[t]{r@{~~}c@{~~}lll}
      $\ctxtrans{\cdot}$  &$=$&   $\cdot$
      \\
      $\ctxtrans{\Gamma, \alpha \type}$  &$=$&   $\ctxtrans{\Gamma}, \alpha \type$
      \\
      $\ctxtrans{\Gamma, \eovar \eo}$  && undefined
  \end{tabular}
  ~~~
  \begin{tabular}[t]{r@{~~}c@{~~}lll}
      $\ctxtrans{\Gamma, \xtypeof{x}{S}}$  &$=$&   $\ctxtrans{\Gamma}, \xtypeoft{x}{\tytrans{S}}$
      \\
      $\ctxtrans{\Gamma, \xtypeof{u}{S}}$  &$=$&   $\ctxtrans{\Gamma}, \xtypeoft{u}{\tytrans{S}}$
  \end{tabular}

  \caption{Translation from economical types to target types}
  \FLabel{fig:tytrans}
\end{figure}

\begin{figure*}[htbp]
  \centering

  \judgbox{\Gamma \entails \elab{e}{\VVAR}{S}{M}}
               {Erased source expression $e$ elaborates at type $S$ to target term $M$}
  \vspace{-2ex}
  \begin{mathpar}
     \Infer{\Evar}
         {(\var x S) \in \Gamma}
         {\Gamma \entails \elab{x}{\VAL}{S}{x}}
     ~~~~~
     \Infer{\Efixvar}
         {(\var u S) \in \Gamma}
         {\Gamma \entails \elab{u}{\NONVAL}{S}{u}}
     ~~~
     \Infer{\Efix}
            {\Gamma, \var{u}{S} \entails \elab{e}{\VVAR}{S}{M}}
            {\Gamma \entails \elab{(\Fix{u} e)}{\NONVAL}{S}{(\tfix{u} M)}}
    ~~~~~
        \Infer{\Eunitintro}
            { }
            {\Gamma \entails \elab{\unit}{\VAL}{\unitty}{\tunit}}
   \\
        \RuleHead{\AllSym}
        \Infer{\Eallintro}
              {\Gamma, \alpha \type \entails \elab e \VAL S M
              }
              {\Gamma \entails \elab{e}{\VAL}{\All{\alpha} S}{\ttylam M}}
        \and
        \Infer{\Eallelim}
             {\Gamma \entails \elab{e}{\VVAR}{\All{\alpha} S}{M}
              \\
              \Gamma \entails S' \type}
             {\Gamma \entails \elab{e}{\VVAR}{[S' / \alpha]S}{\ttyapp{M}}}
    \\
    \RuleHead{\alleosym}
    \Infer{\Ealleointro}
          {\arrayenvbl{
              \Gamma \entails \elab{e}{\VAL}{[\V/\eovar]S}{M_1}
              \\
              \Gamma \entails \elab{e}{\VAL}{[\N/\eovar]S}{M_2}
            }
          }
          {\Gamma \entails \elab{e}{\VAL}{(\alleo{\eovar} S)}{\tpair{M_1}{M_2}}
          }
   \and
   \Infer{\Ealleoelim}
       {\Gamma \entails \elab{e}{\VVAR}{(\alleo{\eovar} S)}{M}}
       {
         \arrayenvbl{
              \Gamma
              \entails
              \elab{e}{\VVAR}
                   {[\V/\eovar]S}
                   {
                     (\tproj{1} M)
                   }
              \\
              \Gamma
              \entails
              \elab{e}{\VVAR}
                   {[\N/\eovar]S}
                   {
                     (\tproj{2} M)
                   }
         }
       }
   \\
\RuleHead{\susp{\epsilon}}
\Infer{\Esuspintro}
      {\Gamma \entails \elab{e}{\VVAR}{S}{M}
      }
      {\arrayenvbl{
        \Gamma \entails\elab{e}{\VVAR}{\susp{\V} S}{M}
        \\
        \Gamma \entails\elab{e}{\VAL}{\susp{\N} S}{(\Thunk M)}
      }}
~~~~~~
\Infer{\Esuspelim{\V}}
     {\Gamma \entails \elab{e}{\VVAR}{\susp{\V} S}{M}}
     {\Gamma \entails \elab{e}{\VVAR}{S}{M}}
~~~
\Infer{\Esuspelim{\N}}
     {\Gamma \entails \elab{e}{\VVAR}{\susp{\N} S}{M}}
     {\Gamma \entails \elab{e}{\NONVAL}{S}{(\Force M)}}
     \\
   \RuleHead{\arr}
     \Infer{\Earrintro}
         {\Gamma, \var{x}{S_1} \entails
           \elab {e} {\VVAR}{S_2} {M}
         }
         {
              \Gamma
              \entails
              \elab 
                   {(\explam {x} e)}
                   {\VAL}
                   {(S_1 \arr S_2)}
                   {\tlam{x} M}
         }
     \and
     \Infer{\Earrelim}
          {
                \Gamma \entails \elab{e_1}{\VVAR_1}{(S_1 \arr S_2)}{M_1}
                \\
                \Gamma \entails \elab{e_2}{\VVAR_2}{S_1}{M_2}
          }
          {\Gamma \entails \elab
                {(\expapp{e_1}{e_2})}
                {\NONVAL}
                {S_2}
                {
                  (M_1 \; M_2) %
                }
          }
   \\
    \RuleHead{*}
       \Infer{\Eprodintro}
             {
                    \Gamma \entails \elab{e_1}{\VVAR_1}{S_1}{M_1}
                    \\
                    \Gamma \entails \elab{e_2}{\VVAR_2}{S_2}{M_2}
             }
             {
                \Gamma \entails \elab{\Pair{e_1}{e_2}}{\VVAR_1 \join \VVAR_2}{(S_1 * S_2)}{\tpair{M_1}{M_2}}
             }
       ~~~~~
       \Infer{\Eprodelim{k}}
              {   
                \Gamma \entails \elab{e}{\VVAR}{(S_1 * S_2)}{M}
              }
              {\Gamma
                \entails
                \elab{(\Proj{k} e)}{\NONVAL}{S_k}{(\tproj{k} M)}
              }
   \\
    \RuleHead{+}
       \Infer{\Esumintro{k}}
             {\Gamma \entails \elab{e}{\VVAR}{S_k}{M}
             }
             {\Gamma \entails \elab{(\Inj{k} e)}{\VVAR}{(S_1 + S_2)}{(\tinj{k} M)}}
       \and
       \Infer{\Esumelim}
              {   
                      \Gamma \entails \elab{e}{\VVAR_0}{(S_1 + S_2)}{M_0} \\
                  \arrayenvbl{
                      \Gamma, x_1 : S_1 \entails \elab{e_1}{\VVAR_1}{S}{M_1}
                      \\
                      \Gamma, x_2 : S_2 \entails \elab{e_2}{\VVAR_2}{S}{M_2}
                  }
              }
              {\Gamma
                \entails
                \arrayenvl{
                    {~~~~~\Casesum{e}{x_1}{e_1}{x_2}{e_2}}
                    \xetasym{\NONVAL}
                    {S}
                    \\
                    \elabsymbol
                    {\tcase{M_0}{x_1}{M_1}{x_2}{M_2}}
                }
              }
    \\
    \RuleHead{\recsymbol}
       \Infer{\Erecintro}
             {\Gamma \entails \elab{e}{\VVAR}{\big[(\Rec{\alpha}S)/\alpha\big]S}{M}
             }
             {\Gamma \entails \elab{e}{\VVAR}{\Rec{\alpha} S}{(\troll M)}}
       \and
       \Infer{\Erecelim}
              {   \Gamma \entails \elab{e}{\VVAR}{\Rec{\alpha} S}{M}
              }
              {
                \Gamma \entails \elab{e}{\NONVAL}{\big[(\Rec{\alpha}S)/\alpha\big]S}{(\tunroll M)}
              }

  \end{mathpar}

  \vspace{-1.0ex}
  \caption{Elaboration}
  \FLabel{fig:elab}
\end{figure*}

Now we extend the economical typing judgment with an output $M$, a
\emph{target term}:
$
  \Gamma \entails \elab{e}{\VVAR}{S}{M}
$.
The target term $M$ should be well-typed using the typing rules
in \Figureref{fig:target-types}, but what type should it have?  We answer this question by
defining another translation on types.  This function, 
defined by a function $\tytrans{S} = A$, translates an economical source type
$S$ to a target type $A$.

We will show that if $\elab{e}{\VVAR}{S}{M}$ then $M : A$, where $A = \tytrans{S}$;
this is \Theoremref{thm:elab-type-soundness}.
Our translation follows a similar approach to
\citet{Dunfield14}.  However, that system had general intersection types
$A_1 \sectty A_2$, where $A_1$ and $A_2$ don't necessarily have the same structure.
In contrast, we have $\alleo{\eovar} A$ which corresponds to $([\V/\eovar]A) \sectty ([\N/\eovar]A)$.
We also differ in having recursive types; since these are explicitly rolled (or \emph{folded})
and unrolled in our target language, our rules \Erecintro and \Erecelim add
these constructs.

\paragraph{Not bidirectional.}
We want to relate the operational behaviour of a source expression to the
operational behaviour of its elaboration.  Since our source operational semantics
is over type-erased source expressions, it will be convenient
for elaboration to work on erased source expressions.  Without type annotations,
we can collapse the bidirectional judgments into a single judgment (with ``$:$''
in place of $\chk$/$\syn$);
this obviates the need for elaboration versions of \Rsub and \Ranno, which merely
switch between $\chk$ and $\syn$.

\paragraph{Elaboration rules.}
We are elaborating the economical type system, which has
by-value connectives, into the target type system, which also
has by-value connectives.
Most of the elaboration rules just map source constructs into the
corresponding target constructs; for example, \Evar elaborates
$x$ to $x$, and \Earrintro elaborates
$\Lam{x} e$ to $\tlam{x} M$ where $e$ elaborates to $M$.

\paragraph{Elaborating $\AllSym$.}
Rule \Eallintro elaborates $e$ (which is type-erased and thus has no
explicit source construct) to the target type abstraction $\ttylam{\!M}$;
rule \Eallelim elaborates to a target type application $\ttyapp{M}$.

\paragraph{Elaborating $\alleosym$.}
Rule \Ealleointro elaborates an $e$ at type $\alleo{\eovar} S$
to a pair with the elaborations of $e$ at type $[\V/\eovar]S$ and
at $[\N/\eovar]S$.  Note that unlike the corresponding rule \Ralleointro
in the non-elaborating economical type system, which introduces
a variable $\eovar$ into $\Gamma$ and types $e$ parametrically,
\Ealleointro substitutes concrete evaluation orders $\V$ and $\N$
for $\eovar$.  Consequently, the $\Gamma$ in the elaboration judgment
never contains $\eovar \eo$ declarations.

Rule \Ealleoelim elaborates to the appropriate projection.

\paragraph{Elaborating $\suspsymbol$.}
Rule \Esuspintro has two conclusions.  The first conclusion
elaborates at type $\susp{\V} S$ as if elaborating at type $S$.
The second conclusion elaborates at $\susp{\N} S$ to a thunk.
Correspondingly, rule \Esuspelim{\V} ignores
the $\V$ suspension, and rule \Esuspelim{\N} forces the thunk
introduced via \Esuspintro.

\subsection{\mksubsection{Elaboration Type Soundness}}

The main result of this section (\Theoremref{thm:elab-type-soundness})
is that, given a non-elaborating
economical typing derivation $\chktypee{\Gamma}{e}{\VVAR}{S}$,
we can derive $\Gamma \entails \elab {\er e} {\VVAR'} S M$
such that the target term $M$ is well-typed.
The erasure function $\er e$, defined in \Figureref{fig:er},
removes type annotations, type abstractions, and type applications.

It will be useful to relate various notions of valueness.
First, if $e$ elaborates to a syntactic target value $W$,
then the elaboration rules deem $e$ to be a (source) value.

\begin{restatable}{lemma}{lemvaluemono}
\Label{lem:value-mono}
    If $\Gamma \entails \elab{e}{\VVAR}{S}{W}$
    then $\VVAR = \VAL$.
\end{restatable}

Second, if $e$ is a value according to the source
typing rules, its elaboration $M$ is valuable (but not necessarily a syntactic target
value).

\begin{restatable}[Elaboration valuability]{lemma}{lemelabvaluability}
\Label{lem:elab-valuability}
~\\
  If $\Gamma \entails \elab{e}{\VAL}{S}{M}$
  then $M$ is valuable, that is, there exists $\Vable$ such that $M = \Vable$.
\end{restatable}

Several substitution lemmas are required.  The first is
for the non-elaborating economical type system; we'll use it in the
\Ralleointro case of the main proof to remove $\eovar \eo$ declarations.

\begin{restatable}[Substitution---Evaluation orders]{lemma}{lemsubsteo}
\Label{lem:subst-eo}
~
\begin{enumerate}[(1)]
\vspace{-1.0ex}
  \item If $\Gamma, \eovar \eo, \Gamma' \entails S \type$
      and $\Gamma \entails \epsilon \eo$
      \\
      then $\Gamma, [\epsilon/\eovar]\Gamma' \entails [\epsilon/\eovar]S \type$.

  \item
    If $\Dee$ derives
    $\chktypee{\Gamma, \eovar \eo, \Gamma'}{e}{\VVAR}{S}$
    and
    $\Gamma \entails \epsilon \eo$
    \\
    then
    $\Dee'$ derives
    $\chktypee{\Gamma, [\epsilon/\eovar]\Gamma'}{e}{\VVAR}{[\epsilon/\eovar]S}$
    where $\Dee'$ is not larger than $\Dee$.

  \item
    If $\Dee$ derives
    $\syntypee{\Gamma, \eovar \eo, \Gamma'}{e}{\VVAR}{S}$
    and
    $\Gamma \entails \epsilon \eo$,
    \\
    then
    $\Dee'$ derives
    $\syntypee{\Gamma, [\epsilon/\eovar]\Gamma'}{e}{\VVAR}{[\epsilon/\eovar]S}$
    where $\Dee'$ is not larger than $\Dee$.
\end{enumerate}
\end{restatable}

{\noindent Next,} we show that an expression
$e_1$ can be substituted for a variable $x$,
provided $e_1$ elaborates to a target value $W$.

\begin{restatable}[Expression substitution]{lemma}{lemelabexprsubst}
\Label{lem:elab-expr-subst}
~
   \begin{enumerate}[(1)]
   \vspace{-1.0ex}
       \item 
         If
         $\Gamma \entails \elab{e_1}{\VVAR_1}{S_1}{W}$
         and
         $\Gamma, x : S_1, \Gamma' \entails \elab{e_2}{\VVAR_2}{S}{M}$
         \\
         then
         $\Gamma, \Gamma' \entails \elab{[e_1/x]e_2}{\VVAR_2}{S}{[W/x]M}$.

       \item 
         If
         $\Gamma \entails \elab{\Fix{u} e_1}{\NONVAL}{S_1}{\tfix{u} M_1}$
         \\
         and
         $\Gamma, u : S_1, \Gamma' \entails \elab{e_2}{\VVAR_2}{S}{M}$
         \\
         then
         $\Gamma, \Gamma' \entails \elab{\big[(\Fix{u}e_1)\big/u\big]e_2}{\VVAR_2}{S}{\big[(\tfix{u} M_1)\big/u\big]M}$.
   \end{enumerate}
\end{restatable}

\begin{lemma}[Type translation well-formedness]   \Label{lem:tytrans-wf}
~\\
   If $\Gamma \entails S \type$
   then $\ctxtrans{\Gamma} \entails \tytrans{S} \type$.
\end{lemma}

We can now state the main result of this section:

\begin{restatable}[Elaboration type soundness]{theorem}{thmelabtypesoundness}
\Label{thm:elab-type-soundness}
~\\
  If $\chktypee{\Gamma}{e}{\VVAR}{S}$
  or $\syntypee{\Gamma}{e}{\VVAR}{S}$
  \\
  where  $\Gamma \entails S \type$
  and $\Gamma$ contains no $\eovar \eo$ declarations
  \\
  then
  there exists $M$ such that
                $\Gamma \entails \elab {\er e} {\VVAR'} S M$
                \\
                where
                $\VVAR' \valleq \VVAR$
                and
                $\typeoft{\ctxtrans{\Gamma}}{M}{\tytrans{S}}$.
\end{restatable}
The proof is in \citet[Appendix B.5]{Dunfield15arxiv}.
In this theorem, the resulting elaboration judgment has a valueness
$\VVAR'$ that can be more precise than the valueness $\VVAR$
in the non-elaborating judgment.
Suppose that, inside a derivation
of $\chktypee{\eovar \eo}{e}{\VAL}{S}$,
we have
\[
  \Infer{\Rsuspelim{\epsilon}}
       {\chktypee{\eovar \eo}{e'}{\VAL}{\susp{\eovar}{S'}}}
       {\chktypee{\eovar \eo}{e'}{\NONVAL}{S'}}
\]
The valueness in the conclusion must be $\NONVAL$, because we might
substitute $\N$ for $\eovar$, which is elaborated to a \xForce, which is
not a value.  Now suppose we substitute $\V$ for $\eovar$.
We need to construct an elaboration derivation,
and the only rule that works is \Esuspelim{\V}:
\[
  \Infer{\Esuspelim{\V}}
       {\cdot \entails \elab{e'}{\VAL}{\susp{\V}{S'}}{M}}
       {\cdot \entails \elab{e'}{\VAL}{S'}{M}}
\]
This says $e'$ is a value ($\VAL$), where the original (parametric)
economical typing judgment had $\NONVAL$:
Substituting a concrete object (here, $\V$) for a variable $\eovar$
increases information, refining $\NONVAL$ (``I cannot prove this is a value'')
into $\VAL$.  %
In the introduction rules,
substituting $\N$ for $\eovar$ can replace $\NONVAL$ with $\VAL$, because
we know we're elaborating to a thunk, which is a value.

\section{Consistency}
\Label{sec:consistency}

Our main result in this section, \Theoremref{thm:consistency-star},
says that if $e$ elaborates to a target term $M$, and
$M$ steps (zero or more times) to a target value $W$, then $e$ steps (zero or more times)
to some $e'$ that elaborates to $W$.  The source language stepping relation
(\Figureref{fig:src-step}) allows both by-value and (more permissive) by-name
reductions, raising the concern that a call-by-value program might elaborate to 
a call-by-name target program, that is, one taking steps that correspond
to by-name reductions in the source program.
So we strengthen the statement, showing that if $M$ is completely free of
by-name constructs, then all the steps taken in the source program are by-value.

That still leaves the possibility that we messed up our elaboration rules,
such that a call-by-value source program elaborates to an $M$ that contains
by-name constructs.  So we prove (\Theoremref{thm:elab-preserves-n-freeness}) that if the source program is completely free
of by-name constructs, its elaboration $M$ is also free of by-name constructs.
Similarly, we prove (\Theoremref{thm:econ-preserves-n-freeness})
that creating an economical typing derivation from an impartial typing derivation
preserves $\N$-freeness.

Proofs can be found in \citet[Appendix B.6]{Dunfield15arxiv}.

\subsection{\mksubsection{Source-Side Consistency?}}
\Label{sec:source-side-consistency}
A source expression typed by name won't get stuck if a
by-value reduction is chosen, but it may diverge instead
of terminating.  Suppose we have typed $(\Lam{x} x)$
against $\tau \garr{\N} \tau$.  Taking only a by-name reduction,
we have
{\small \begin{align*}
  (\Lam{x} \unit)(\Fix{u}u)
  &~\srcstep~
    [(\Fix{u}u) / x]\unit
~=~ \unit
  \text{~~~~using \SrcRedBetaN}
\end{align*}}
However, if we ``contradict'' the typing derivation by taking
by-value reductions, we diverge:
{\small\begin{align*}
  (\Lam{x} \unit)(\Fix{u}u)
  &~\srcstep~
  (\Lam{x} \unit)\big([(\Fix{u}u)/u]u\big)
  \text{~~using \SrcRedFixV}
\\[-0.3ex] &~=~ 
  (\Lam{x} \unit)(\Fix{u} u)
~~\srcstep~ \dots
\end{align*}}%
We're used to type safety being ``up to'' nontermination in the sense
that we either get a value or diverge, without getting stuck, but this
is worse: divergence depends on which
reductions are chosen.

To get a source type safety result
that is both direct (without appealing to elaboration and target reductions)
and useful, we'd need to give a semantics of ``reduction with respect to a
typing derivation'', or else reduction \emph{of} a typing derivation.
Such a semantics would support reasoning about local transformations
of source programs.  It should also lead to a converse of the consistency
result in this section: if a source expression reduces with respect to
a typing derivation, and that typing derivation corresponds to an elaboration
derivation, then the target program obtained by elaboration can be correspondingly
reduced.

\subsection{\mksubsection{Defining $\N$-Freeness}}

\begin{definition}[$\N$-freeness---impartial]
\Label{def:n-free-impartial}
~
\begin{enumerate}[(1)]
  \item 
    An impartial type $\tau$
    is \emph{$\N$-free} iff (i) for 
    each $\epsilon$ appearing in $S$, the evaluation order $\epsilon$
    is $\V$; and (ii) $\tau$ has no $\alleosym$ quantifiers.
  
  \item
    A judgment $\chktype{\gamma}{e}{\VVAR}{\tau}$
    or $\syntype{\gamma}{e}{\VVAR}{\tau}$
    is \emph{$\N$-free} iff:
    (a) $\gamma$ has no $\eovar \eo$ declarations;
    (b) in each declaration $\xsyn x \VVAR \tau$ in $\gamma$,
         the valueness $\VVAR$ is $\VAL$ and the type $\tau$ is $\N$-free;
    (c) all types appearing in $e$ are $\N$-free;
    and
    (d) $\tau$ is $\N$-free.
\end{enumerate}
\end{definition}

\begin{definition}[$\N$-freeness---economical]
\Label{def:n-free-econ}
~
\begin{enumerate}[(1)]
  \item 
    An economical type $S$
    is \emph{$\N$-free} iff (i) for 
    each $\susp{\epsilon} S_0$ appearing in $S$, the evaluation order $\epsilon$
    is $\V$; and (ii) $S$ has no $\alleosym$ quantifiers.

  \item
    A judgment $\chktypee{\Gamma}{e}{\VVAR}{S}$
    or $\syntypee{\Gamma}{e}{\VVAR}{S}$
    is \emph{$\N$-free} iff:
    (a) $\Gamma$ has no $\eovar \eo$ declarations;
    (b) all types $S'$ in $\Gamma$ are $\N$-free;
    (c) all types appearing in $e$ are $\N$-free;
    and
    (d) $S$ is $\N$-free.
\end{enumerate}
\end{definition}

\begin{definition}[$\N$-freeness---target]
\Label{def:n-free-target}
  A target term $M$ is \emph{$\N$-free} iff
  it contains no $\xThunk$ and $\xForce$ constructs.
\end{definition}

\subsection{\mksubsection{Lemmas for Consistency}}

An inversion lemma
allows types of the form $\susp{\V} \dots \susp{\V} S$, a generalization
needed for the \Esuspelim{\V} case; when we use the lemma in the
consistency proof, the type is not headed by $\susp{\V}$:

\begin{restatable}[Inversion]{lemma}{leminversion}
\Label{lem:inversion}
     ~~Given $\cdot \entails \elab{e}{\VVAR}{\underbrace{\susp{\V} \dots \susp{\V}}_\text{0 or more} S}{M}$:
\vspace{-1ex}
     \begin{enumerate}[(1)]
       \item[(0)] %
           If $M = (\tlam{x}M_0)$
           and $S = (S_1 \arr S_2)$
           \\
           then
           $e = (\Lam{x} e_0)$
           and
           $\cdot, x : S_1 \entails \elab{e_0}{\VVAR'}{S_2}{M_0}$.

       \item %
           If $M = \tpair{W_1}{W_2}$
           and $S = (\alleo{\eovar} S_0)$
           \\
           then
           $\cdot \entails \elab{e}{\VVAR}{[\V/\eovar]S_0}{W_1}$
           and
           $\cdot \entails \elab{e}{\VVAR}{[\N/\eovar]S_0}{W_2}$.

       \item %
           If $M = \Thunk M_0$
           and $S = \susp{\N} S_0$
           then
           $\cdot \entails \elab{e}{\VVAR'}{S_0}{M_0}$.

\ifcsname INAPPENDIX\endcsname
       \item %
           If $M = \ttylam{M_0}$
           and $S = (\All{\alpha} S_0)$
           \\
           then
           $\cdot, \alpha \type \entails \elab{e}{\VAL}{S_0}{M_0}$.

       \item %
           If $M = (\tinj{k} W)$
           and $S = (S_1 + S_2)$
           \\
           then
           $e = (\Inj{k} e')$
           and
           $\cdot \entails \elab{e'}{\VVAR}{S_k}{W}$.

       \item %
           If $M = (\troll W)$
           and $S = (\Rec{\alpha} S_0)$
           \\
           then
           $\cdot \entails \elab{e}{\VVAR}{\big[(\Rec{\alpha}S_0)/\alpha\big]S_0}{W}$.

       \item %
           If $M = \tpair{W_1}{W_2}$
           and $S = (S_1 * S_2)$
           \\
           then
             $\cdot \entails \elab{e_1}{\VVAR_1}{S_1}{W_1}$
             and 
             $\cdot \entails \elab{e_2}{\VVAR_2}{S_2}{W_2}$
             \\
             where $e = \Pair{e_1}{e_2}$
             and $\VVAR = \VVAR_1 \join \VVAR_2$.
\fi
     \end{enumerate}
\end{restatable}
{\noindent Parts} (3)--(6), for $\AllSym$, $+$, $\mu$ and $*$, are stated
in the appendix.

Previously, we showed that if a source expression elaborates to
a target value, source typing says the expression is a value ($\VVAR = \VAL$);
here, we show that if a source expression elaborates to a target
value that is $\N$-free (ruling out
$\Thunk M$ produced by the second conclusion of \Esuspintro),
then $e$  is a \emph{syntactic} value.

\begin{restatable}[Syntactic values]{lemma}{lemsyntacticvalues}
\Label{lem:syntactic-value}
~\\
    If $\Gamma \entails \elab e \VAL S W$
    and $W$ is $\N$-free
    then $e$ is a syntactic value.
\end{restatable}

The next lemma just says that the $\step$ relation doesn't produce
$\xThunk$s and $\xForce$s out of thin air.

\begin{lemma}[Stepping preserves $\N$-freeness]
\Label{lem:step-n-freeness}
  If $M$ is $\N$-free and $M \step M'$
  then $M'$ is $\N$-free.
\end{lemma}

The proof is by cases on the derivation of $M \step M'$,
using the fact that if $M_0$ and $M_1$ are $\N$-free,
then $[M_0/x]M_1$ is $\N$-free.

\subsection{\mksubsection{Consistency Results}}

\begin{restatable}[Consistency]{theorem}{thmconsistency}
\Label{thm:consistency}
~\\
  If
  $\cdot \entails \elab e \VVAR S M$
  and $M \step M'$
  then
  there exists $e'$ such that
  $e \srcsteps e'$
  and $\cdot \entails \elab{e'}{\VVAR'}{S}{M'}$
  and $\VVAR' \valleq \VVAR$.
  \\
  Moreover: (1) If $\VVAR = \VAL$ then $e' = e$.
  (2) If $M$ is $\N$-free then $e \srcsteps e'$ can be derived
    without using \SrcStepContextN.
\end{restatable}

Result (1), under ``moreover'', amounts to saying that values don't step.
Result (2) stops us from lazily sneaking in uses of \SrcStepContextN instead
of showing that, given $\N$-free $M$, we can always find a by-value evaluation
context for use in \SrcStepContextV.

\begin{restatable}[Multi-step consistency]{theorem}{thmconsistencystar}
\Label{thm:consistency-star} ~\\
  If $\cdot \entails \elab e \VVAR S M$
  and $M \steps W$
  then there exists $e'$ such that
  $e \srcsteps e'$ and
  $\cdot \entails \elab {e'} {\VAL} {S} {W}$.
  Moreover, if $M$ is $\N$-free then we can derive $e \srcsteps e'$
  without using \SrcStepContextN.
\end{restatable}

\subsection{\mksubsection{Preservation of $\N$-Freeness}}

\begin{restatable}{lemma}{econbisubformula}%
\Label{lem:econ-bi-subformula}
  If $\syntypee{\Gamma}{e}{\VVAR}{S}$
  and $S$ is \emph{not} $\N$-free
  then
  it is not the case that both $\Gamma$ and $e$ are $\N$-free.
\end{restatable}

\begin{restatable}[Economizing preserves $\N$-freeness]{theorem}{econpreservesnfreeness}
\Label{thm:econ-preserves-n-freeness}
~\\
  If $\chktype{\gamma}{e}{\VVAR}{\tau}$ (resp.\ $\syn$)
  where the judgment is $\N$-free (\Definitionref{def:n-free-impartial} (2))
  then
  $\chktypee{\ctxecontrans{\gamma}}{\expecontrans{e}}{\VVAR}{\econtrans{\tau}}$
  (resp.\ $\syn$)
  where this judgment is $\N$-free (\Definitionref{def:n-free-econ} (2)).
\end{restatable}

\begin{restatable}[Elaboration preserves $\N$-freeness]{theorem}{elabpreservesnfreeness}%
\Label{thm:elab-preserves-n-freeness}
~\\
  If $\chktypee{\Gamma}{e}{\VVAR}{S}$ (or $\syn$)
  where the judgment is $\N$-free (\Definitionref{def:n-free-econ} (2))
  then
  $\Gamma \entails \elab{\er{e}}{\VVAR}{S}{M}$
  such that $M$ is $\N$-free.
\end{restatable}

\section{Related Work}

\paragraph{History of evaluation order.}
In the $\lambda$-calculus, normal-order (leftmost-outermost) reduction
seems to have preceded anything resembling call-by-value, but
\citet{Bernays36} suggested requiring that the term being substituted in a
reduction be in normal form.
In programming languages,
Algol-60 originated call-by-name and also provided call-by-value~\citep[4.7.3]{Naur60};
while the decision to make the former the default is debatable,
direct support for two evaluation orders made Algol-60 an improvement on
many of its successors.
\citet{Plotkin75} related cbv and cbn to the $\lambda$-calculus, and
developed translations between them.

Call-by-need or \emph{lazy} evaluation was developed in the 1970s with the goal
of doing as little computational work as possible, under which we
can include the unbounded work of not terminating~\citep{WadsworthThesis,Henderson76,Friedman76}.

\paragraph{Laziness in call-by-value languages.}
Type-based support for selective lazy evaluation has been developed
for cbv languages, including Standard ML \citep{Wadler98} and
Java~\citep{Warth07}.  These approaches allow programmers to
conveniently switch to another evaluation order, but don't allow
polymorphism over evaluation orders.  Like our economical type system,
these approaches are biased towards one evaluation order.

\paragraph{General coercions.}
General approaches to typed coercions were explored by
\citet{BreazuTannen91} and \citet{Barthe96}.
\citet{Swamy09} developed a general typed coercion system
for a simply-typed calculus, giving thunks as an example.
In addition to annotations on all $\lambda$ arguments,
their system requires thunks (but not forces) to be written explicitly.

\paragraph{Intersection types.}
While this paper avoids the notation of intersection types, the quantifier $\alleosym$
is essentially an intersection type of a very specific form.
Theories of intersection types were originally developed by
\citet{CDV81:IntersectionTypes}, among others; \citet{Hindley92} gives
a useful introduction and survey.
Intersections entered programming languages---as opposed to
$\lambda$-calculus---when \citet{Reynolds96:Forsythe}
put them at the heart of the Forsythe language.  Subsequently---Reynolds's
paper describes ideas he developed in the 1980s---\citet{Freeman91} started
a line of research on \emph{refinement} intersections, where both parts
of an intersection must refine the same base type (essentially, the same ML type).

The $\alleosym$ intersection in this paper mixes features of
general intersection and refinement intersection: the $\V$ and $\N$ instantiations
have close-to-identical structure, but cbv and cbn functions aren't refinements
of some ``order-agnostic'' base type.  Our approach is descended mainly from
the system of \citet{Dunfield14}, which elaborates (general) intersection and union types
into ordinary product and sum types.  We differ in
not having a source-level `merge' construct $\Merge{e_1}{e_2}$, where the
type system can select either $e_1$ or $e_2$, ignoring the other component.  Since $e_1$
and $e_2$ are not prevented from having the same type, the type system may
elaborate either expression, resulting in unpredictable behaviour.  
In our type systems, we can think of $\appsymbol$
in the source language as a merge
$(\Merge{\vappsymbol}{\nappsymbol})$, but the components have
incompatible types.  Moreover, the components must behave the same
apart from evaluation order (evoking a standard property of systems
of refinement intersection).

\paragraph{Alternative target languages.}
The impartial type system for our source language suggests
that we should consider targeting an impartial, but more explicit,
target language.
In an untyped setting,
\citet{Asperti90} %
developed a calculus with call-by-value and call-by-name 
$\lambda$-abstractions;
function application is disambiguated at run time.
In a typed setting, 
call-by-push-value~\citep{Levy99}  %
systematically distinguishes values and computations;
it has a thunk type $\xU$ (whence our notation)
but also a dual, ``lift'' $\xF$, which constructs a computation
out of a value type.  Early in the development of this paper,
we tried to elaborate directly from the impartial type system to cbpv,
without success.  Levy's elegant \emph{pair} of translations from cbv
and from cbn don't seem to fit together easily; our feeling is that
a combined translation would be either complicated, or prone
to generating many redundant forces and thunks.

\citet{ZeilbergerThesis} defined a polarized type system with positive and
negative forms of each standard connective.  In that system, $\downarrow$
and $\uparrow$ connectives alternate between polarities, akin to $\xU$ and
$\xF$ in call-by-push-value.  Zeilberger's system has a
symmetric function type, rather than the asymmetric function
type found in cbpv.  We guess that a translation into this system
would have similar issues as with call-by-push-value.

\section{Future Work}
\Label{sec:future}

This paper develops type systems with multiple evaluation orders
and polymorphism over evaluation orders, opening up the design space.
More work is needed to realize these ideas in practice.

\paragraph{Implicit polymorphism.}
We made type polymorphism explicit, to prevent the type system
  from guessing evaluation orders.  A practical system should find polymorphic instances
  without guessing, perhaps based on existential type variables~\citep{Dunfield13}.
  We could also try to use some form of (lexically scoped?) default evaluation order.
  Such a default could also be useful for deciding whether some language features, such as
  \textkw{let}-expressions, should be by-value or by-name.

\paragraph{Exponential expansion.}
Our rules elaborate a function typed with $n$ $\alleosym$ quantifiers into
$2^n$ instantiations.  Only experience can demonstrate whether
this is a problem in practice, but we have reasons to be optimistic.

First, we need the right point of comparison.
The alternative to elaborating
\textfn{map} into, say, 8 instantiations is to write 8 copies of \textfn{map}
by hand.  Viewed this way, elaboration maintains the size of the target program,
while allowing an exponentially shorter source program!
(This is the flipside of a sleight-of-hand from complexity theory, where you
can make an algorithm look faster by inflating the input:
Given an algorithm that takes $2^n$ time, where $n$
is the number of bits in the input integer, we can get a purportedly polynomial
algorithm by encoding the input in unary.)

Second, a compiler could analyze the source program and generate only the
instances actually used, similar to monomorphization of $\forall$-polymorphism
in MLton (\href{http://mlton.org}{mlton.org}).

\paragraph{Other evaluation orders.}
Our particular choice of evaluation orders is
not especially practical: the major competitor to call-by-value
is call-by-need, not call-by-name.  We chose call-by-name for simplicity (for example,
in the source reduction rules), but many of our techniques should be directly
applicable to call-by-need: elaboration would produce thunks in much the same way,
just for a different dynamic semantics.
Moreover, our approach could be extended to more than two evaluation orders,
using an $n$-way intersection that elaborates to an $n$-tuple.

One could also take ``order'' very literally, and support
left-to-right \emph{and} right-to-left call-by-value.  For low-level reasons,
OCaml uses the former when compiling to native code, and the latter when compiling to
bytecode.  Being able to specify order of evaluation via type annotations
could be useful when porting code from Standard ML (which uses left-to-right call-by-value).

\paragraph{Program design.}
We also haven't addressed questions about \emph{when} to use what
evaluation order.  Such questions seem to have been lightly studied,
perhaps because of social
factors: a programmer may choose a strict language because they
tend to solve problems that don't need laziness---which
is self-reinforcing, because laziness is less convenient in a strict language.
However, \citet{ChangThesis} developed
tools, based on both static analysis and dynamic profiling, that suggest
where laziness is likely to be helpful.

\paragraph{Existential quantification.}
By analogy to union types \citep{Dunfield14}, an existential quantifier
would elaborate to a sum type.  For example, the sum tag on a function
of type $\ExisSym{\eovar}.\,\tau \garr{\eovar} \tau$ would indicate, at run
time, whether the function was by-value or by-name.  This might resemble
a typed version of the calculus of \citet{Asperti90}.

\vspace{-0pt}

\section*{Acknowledgments}
The ICFP reviewers made suggestions and asked questions that have
(I believe) improved the paper.
The Max Planck Institute for Software Systems supported the early
stages of this work.  Dmitry Chistikov suggested the symbol $\alleosym$.

\addtolength{\bibsep}{-1.95pt}
 \vspace{-2pt}

\bibliographystyle{initials}
\bibliography{eo}

\begin{thebibliography}{36}
\providecommand{\natexlab}[1]{#1}
\providecommand{\url}[1]{\texttt{#1}}
\expandafter\ifx\csname urlstyle\endcsname\relax
  \providecommand{\doi}[1]{doi: #1}\else
  \providecommand{\doi}{doi: \begingroup \urlstyle{rm}\Url}\fi

\bibitem[Asperti(1990)]{Asperti90}
A.~Asperti.
\newblock Integrating strict and lazy evaluation: the
  $\lambda_{\text{sl}}$-calculus.
\newblock In \emph{Programming Language Implementation and Logic Programming},
  volume 456 of \emph{LNCS}, pages 238--254. Springer, 1990.

\bibitem[Barendregt et~al.(1983)Barendregt, Coppo, and
  Dezani-Ciancaglini]{Barendregt83}
H.~Barendregt, M.~Coppo, and M.~Dezani-Ciancaglini.
\newblock A filter lambda model and the completeness of type assignment.
\newblock \emph{J. Symbolic Logic}, 48\penalty0 (4):\penalty0 931--940, 1983.

\bibitem[Barthe(1996)]{Barthe96}
G.~Barthe.
\newblock Implicit coercions in type systems.
\newblock In \emph{Proc. TYPES '95}, volume 1158 of \emph{LNCS}, pages 1--15,
  1996.

\bibitem[Bernays(1936)]{Bernays36}
P.~Bernays.
\newblock Review of ``{Some Properties of Conversion}'' by {Alonzo Church} and
  {J.B. Rosser}.
\newblock \emph{J. Symbolic Logic}, 1:\penalty0 74--75, 1936.

\bibitem[Breazu-Tannen et~al.(1991)Breazu-Tannen, Coquand, Gunter, and
  Scedrov]{BreazuTannen91}
V.~Breazu-Tannen, T.~Coquand, C.~A. Gunter, and A.~Scedrov.
\newblock Inheritance as implicit coercion.
\newblock \emph{Information and Computation}, 93\penalty0 (1):\penalty0
  172--221, 1991.

\bibitem[Chang(2014)]{ChangThesis}
S.~Chang.
\newblock \emph{On the Relationship Between Laziness and Strictness}.
\newblock PhD thesis, Northeastern University, 2014.

\bibitem[Chen et~al.(2014)Chen, Dunfield, Hammer, and Acar]{Chen14}
Y.~Chen, J.~Dunfield, M.~A. Hammer, and U.~A. Acar.
\newblock Implicit self-adjusting computation for purely functional programs.
\newblock \emph{J. Functional Programming}, 24\penalty0 (1):\penalty0 56--112,
  2014.

\bibitem[Coppo et~al.(1981)Coppo, Dezani-Ciancaglini, and
  Venneri]{CDV81:IntersectionTypes}
M.~Coppo, M.~Dezani-Ciancaglini, and B.~Venneri.
\newblock Functional characters of solvable terms.
\newblock \emph{Zeitschrift f. math. Logik und Grundlagen d. Math.},
  27:\penalty0 45--58, 1981.

\bibitem[Davies(2005)]{DaviesThesis}
R.~Davies.
\newblock \emph{Practical Refinement-Type Checking}.
\newblock PhD thesis, Carnegie Mellon University, 2005.
\newblock CMU-CS-05-110.

\bibitem[Davies and Pfenning(2000)]{Davies00icfpIntersectionEffects}
R.~Davies and F.~Pfenning.
\newblock Intersection types and computational effects.
\newblock In \emph{ICFP}, pages 198--208, 2000.

\bibitem[Dunfield(2014)]{Dunfield14}
J.~Dunfield.
\newblock Elaborating intersection and union types.
\newblock \emph{J. Functional Programming}, 24\penalty0 (2--3):\penalty0
  133--165, 2014.

\bibitem[Dunfield(2015)]{Dunfield15arxiv}
J.~Dunfield.
\newblock Elaborating evaluation-order polymorphism, 2015.
\newblock Extended version with appendices.
  \href{http://arxiv.org/abs/1504.07680}{arXiv:{\tt 1504.07680 [cs.PL]}}.

\bibitem[Dunfield and Krishnaswami(2013)]{Dunfield13}
J.~Dunfield and N.~R. Krishnaswami.
\newblock Complete and easy bidirectional typechecking for higher-rank
  polymorphism.
\newblock In \emph{ICFP}, 2013.
\newblock \href{http://arxiv.org/abs/1306.6032}{arXiv:{\tt 1306.6032 [cs.PL]}}.

\bibitem[Dunfield and Pfenning(2003)]{Dunfield03:IntersectionsUnionsCBV}
J.~Dunfield and F.~Pfenning.
\newblock Type assignment for intersections and unions in call-by-value
  languages.
\newblock In \emph{FoSSaCS}, pages 250--266, 2003.

\bibitem[Dunfield and Pfenning(2004)]{Dunfield04:Tridirectional}
J.~Dunfield and F.~Pfenning.
\newblock Tridirectional typechecking.
\newblock In \emph{Principles of Programming Languages}, pages 281--292, 2004.

\bibitem[Freeman and Pfenning(1991)]{Freeman91}
T.~Freeman and F.~Pfenning.
\newblock Refinement types for {ML}.
\newblock In \emph{PLDI}, pages 268--277, 1991.

\bibitem[Friedman and Wise(1976)]{Friedman76}
D.~P. Friedman and D.~S. Wise.
\newblock {CONS} should not evaluate its arguments.
\newblock In \emph{ICALP}, pages 257--284. Edinburgh Univ. Press, 1976.

\bibitem[Frisch et~al.(2002)Frisch, Castagna, and Benzaken]{Frisch02}
A.~Frisch, G.~Castagna, and V.~Benzaken.
\newblock Semantic subtyping.
\newblock In \emph{Logic in Computer Science}, 2002.

\bibitem[Henderson and Morris(1976)]{Henderson76}
P.~Henderson and J.~H. Morris, Jr.
\newblock A lazy evaluator.
\newblock In \emph{Principles of Programming Languages}, pages 95--103. ACM,
  1976.

\bibitem[Hindley(1992)]{Hindley92}
J.~R. Hindley.
\newblock Types with intersection: An introduction.
\newblock \emph{Formal Aspects of Computing}, 4:\penalty0 470--486, 1992.

\bibitem[Leivant(1986)]{Leivant86}
D.~Leivant.
\newblock Typing and computational properties of lambda expressions.
\newblock \emph{Theoretical Computer Science}, 44\penalty0 (0):\penalty0
  51--68, 1986.

\bibitem[Levy(1999)]{Levy99}
P.~B. Levy.
\newblock Call-by-push-value: A subsuming paradigm.
\newblock In \emph{Typed Lambda Calculi and Applications}, pages 228--243.
  Springer, 1999.

\bibitem[Milner et~al.(1997)Milner, Tofte, Harper, and
  MacQueen]{RevisedDefinitionOfStandardML}
R.~Milner, M.~Tofte, R.~Harper, and D.~MacQueen.
\newblock \emph{The Definition of Standard {ML} ({Revised})}.
\newblock MIT Press, 1997.

\bibitem[Naur et~al.(1960)]{Naur60}
P.~Naur et~al.
\newblock Report on the algorithmic language {ALGOL} 60.
\newblock \emph{Comm. ACM}, 3\penalty0 (5):\penalty0 299--314, 1960.

\bibitem[Pierce(2002)]{Pierce02:TAPL}
B.~C. Pierce.
\newblock \emph{Types and Programming Languages}.
\newblock MIT Press, 2002.

\bibitem[Pierce and Turner(2000)]{Pierce00}
B.~C. Pierce and D.~N. Turner.
\newblock Local type inference.
\newblock \emph{ACM Trans. Prog. Lang. Systems}, 22:\penalty0 1--44, 2000.

\bibitem[Plotkin(1975)]{Plotkin75}
G.~Plotkin.
\newblock Call-by-name, call-by-value, and the lambda calculus.
\newblock \emph{Theoretical Computer Science}, 1:\penalty0 125--159, 1975.

\bibitem[Reynolds(1996)]{Reynolds96:Forsythe}
J.~C. Reynolds.
\newblock Design of the programming language {Forsythe}.
\newblock Technical Report CMU-CS-96-146, Carnegie Mellon University, 1996.

\bibitem[Swamy et~al.(2009)Swamy, Hicks, and Bierman]{Swamy09}
N.~Swamy, M.~Hicks, and G.~M. Bierman.
\newblock A theory of typed coercions and its applications.
\newblock In \emph{ICFP}, pages 329--340, 2009.

\bibitem[Wadler et~al.(1998)Wadler, Taha, and MacQueen]{Wadler98}
P.~Wadler, W.~Taha, and D.~MacQueen.
\newblock How to add laziness to a strict language without even being odd.
\newblock In \emph{Workshop on Standard ML}, 1998.
\newblock
  \url{http://homepages.inf.ed.ac.uk/wadler/papers/lazyinstrict/lazyinstrict.ps}.

\bibitem[Wadsworth(1971)]{WadsworthThesis}
C.~Wadsworth.
\newblock \emph{Semantics and Pragmatics of the lambda-Calculus}.
\newblock PhD thesis, University of Oxford, 1971.

\bibitem[Warth(2007)]{Warth07}
A.~Warth.
\newblock {LazyJ}: Seamless lazy evaluation in {J}ava.
\newblock In \emph{FOOL}, 2007.
\newblock
  \href{http://foolwood07.cs.uchicago.edu/program/warth.pdf}{\texttt{foolwood07.cs.uchicago.edu/program/warth.pdf}}.

\bibitem[Wexelblat(1981)]{Wexelblat81}
R.~L. Wexelblat, editor.
\newblock \emph{History of Programming Languages {I}}.
\newblock ACM, 1981.

\bibitem[Wright(1995)]{Wright95}
A.~K. Wright.
\newblock Simple imperative polymorphism.
\newblock \emph{Lisp and Symbolic Computation}, 8\penalty0 (4):\penalty0
  343--355, 1995.

\bibitem[Xi(1998)]{XiThesis}
H.~Xi.
\newblock \emph{Dependent Types in Practical Programming}.
\newblock PhD thesis, Carnegie Mellon University, 1998.

\bibitem[Zeilberger(2009)]{ZeilbergerThesis}
N.~Zeilberger.
\newblock \emph{The Logical Basis of Evaluation Order and Pattern-Matching}.
\newblock PhD thesis, Carnegie Mellon University, 2009.
\newblock CMU-CS-09-122.

\end{thebibliography}

\normalsize

\clearpage

\onecolumn

\ErratumSection{Call-by-name evaluation contexts}
\label{erratum}

\textit{Corrected in arXiv version 3.}

\subsection*{What is the mistake?}

The definition of by-name evaluation contexts in \Figureref{fig:src-step}
is wrong; it manages to define a peculiarly \emph{eager} evaluation context that
can evaluate a function's argument before the function
has been evaluated, and evaluate inside a pair.  In addition to not being
call-by-name, this is awfully nondeterministic.
\[
  \begin{bnfarray}
     \text{\small By-name eval.\ contexts}
            & \EN
            &\bnfas&
              \hole
              \bnfaltBRK
              \EN \vapp e_2
              \bnfalt
              \matherrorhi{e_1 \vapp \EN}
              \bnfaltBRK
              \matherrorhi{\Pair{\EN}{e_2}}
              \bnfalt
              \matherrorhi{\Pair{{e_1}}{\EN}}
              \bnfalt
              \Proj{k}{\EN}
              \bnfaltBRK
              \tinj{k}{\EN}
              \bnfalt
              \Casesum{\EN}{x_1}{e_1}{x_2}{e_2}
 \end{bnfarray}
\]
The fix is to omit the three \errorhi{boxed} alternatives in the grammar.
\[
  \begin{bnfarray}
     \text{\small By-name eval.\ contexts}
            & \EN
            &\bnfas&
              \hole
              \bnfaltBRK
              \EN \vapp e_2
              \bnfaltBRK
              \Proj{k}{\EN}
              \bnfaltBRK
              \tinj{k}{\EN}
              \bnfalt
              \Casesum{\EN}{x_1}{e_1}{x_2}{e_2}
  \end{bnfarray}
\]
The discussion in \Sectionref{sec:source-opsem}, marked with a red box,
notes that ``$\expapp{e_1}{\hole}$ is a $\EN$ but not a $\EV$'', which
matches the (wrong) definition; however, since the definition is wrong,
the claim that ``the definitions of $\EN$ and $\srcredN$ are standard for call-by-name''
is utterly wrong.

\subsection*{What are its consequences?}

Few (apart from embarrassment).  The consistency result is only a simulation,
not a bisimulation.  None of the metatheory goes \emph{from} a source reduction
\emph{to} a target reduction; that is, no claims have the form ``given some $e \srcstep e'$,
where $e$ is related to $M$, produce some $M'$ such that $M \step M'$''.

In fact, one could add any kind of garbage to the definition of $\EN$, and the 
metatheory wouldn't change.

\ErratumSection{Uppercase, lowercase}

\textit{Corrected in arXiv version 3.}

In the published version, the ``judgment boxes'' heading the rules had $\Gamma$
instead of $\gamma$.
Similarly, Theorem 17 had $\ctxecontrans{\Gamma}$ instead of $\ctxecontrans{\gamma}$.

As these are minor mistakes, they are not highlighted in the text.

\ifnum\OPTIONAppendix=1
\clearpage

\onecolumn
\appendix
\global\edef\INAPPENDIX{yes}

\fontsize{9.5pt}{11.5pt}\selectfont

\section*{Supplemental material for ``Elaborating Evaluation-Order Polymorphism''}

This section of the extended version~\citep{Dunfield15arxiv} contains the
(straightforward) rules for
type well-formedness
(\Appendixref{apx:type-wf}),
proofs about economical typing that belong to
\Sectionref{sec:econ} (\Appendixref{apx:econ}),
proofs about elaboration typing that belong to
\Sectionref{sec:elab} (\Appendixref{apx:elab}),
and
consistency proofs that belong to
\Sectionref{sec:consistency} (\Appendixref{apx:consistency}).

\section{Type Well-formedness}
\Label{apx:type-wf}

\begin{figure}[htbp]
  \centering

  \judgbox{\gamma \entails \epsilon \eo}
               {Evaluation order $\epsilon$ is well-formed}
  \vspace{-4ex}
  \begin{mathpar}
      \Infer{}
          {}
          { 
              \arrayenvc{
                 \gamma \entails \V \eo
                 \\
                 \gamma \entails \N \eo
              }
          }
      \and
      \Infer{}
           {(\eovar \eo) \in \gamma}
           {\gamma \entails \eovar \eo}
    \end{mathpar}

  \judgbox{\gamma \entails \tau \type}
               {Impartial type $\tau$ is well-formed}
  \vspace{-3ex}
  \begin{mathpar}
      \Infer{}
           {}
           {\gamma \entails \unitty \type}
      \and
      \Infer{}
           {(\alpha \type) \in \gamma}
           {\gamma \entails \alpha \type}
      \and
      \Infer{}
           {\gamma, \alpha \type \entails \tau \type}
           {\gamma \entails (\All{\alpha}\tau) \type}
      \and
      \Infer{}
           {\gamma, \eovar \eo \entails \tau \type}
           {\gamma \entails (\alleo{\eovar}\tau) \type}
      \\
      \Infer{}
           { \gamma \entails \epsilon \eo
             \\
             \arrayenvbl{
                 \gamma \entails \tau_1 \type
                 \\
                 \gamma \entails \tau_2 \type
             }
           }
          { 
              \arrayenvc{
                 \gamma \entails (\tau_1 \garr{\epsilon} \tau_2) \type
                 \\
                 \gamma \entails (\tau_1 *^{\epsilon} \tau_2) \type
                 \\
                 \gamma \entails (\tau_1 +^{\epsilon} \tau_2) \type
              }
          }
      \and
      \Infer{}
           { \gamma \entails \epsilon \eo
             \\
             \gamma, \alpha \type \entails \tau \type
           }
           {\gamma \entails (\rec{\epsilon}{\alpha} \tau) \type}
    \end{mathpar}

  \caption{Type well-formedness in the impartial type system}
  \FLabel{fig:impartial-wf}
\end{figure}

\begin{figure}[htbp]
  \centering
  
  \judgbox{\Gamma \entails \epsilon \eo}
               {Evaluation order $\epsilon$ is well-formed}
  \vspace{-4ex}
  \begin{mathpar}
      \Infer{}
          {}
          { 
              \arrayenvc{
                 \Gamma \entails \V \eo
                 \\
                 \Gamma \entails \N \eo
              }
          }
      \and
      \Infer{}
           {(\eovar \eo) \in \Gamma}
           {\Gamma \entails \eovar \eo}
    \end{mathpar}

  \judgbox{\Gamma \entails S \type}
               {Economical type $S$ is well-formed}
  \vspace{-3ex}
  \begin{mathpar}
      \Infer{}
           {}
           {\Gamma \entails \unitty \type}
      \and
      \Infer{}
           {(\alpha \type) \in \Gamma}
           {\Gamma \entails \alpha \type}
      \and
      \Infer{}
           {\Gamma, \alpha \type \entails S \type}
           {\Gamma \entails (\All{\alpha}S) \type}
      \and
      \Infer{}
           {\Gamma, \eovar \eo \entails S \type}
           {\Gamma \entails (\alleo{\eovar}S) \type}
      \\
      \Infer{}
          {\Gamma \entails \epsilon \eo
            \\
            \Gamma \entails S \type
          }
          {
            \Gamma \entails (\susp{\epsilon} S) \type
          } 
      \and
      \Infer{}
           { \arrayenvbl{
                 \Gamma \entails S_1 \type
                 \\
                 \Gamma \entails S_2 \type
             }
           }
          { 
              \arrayenvc{
                 \Gamma \entails (S_1 \arr S_2) \type
                 \\
                 \Gamma \entails (S_1 * S_2) \type
                 \\
                 \Gamma \entails (S_1 + S_2) \type
              }
          }
      \and
      \Infer{}
           { \Gamma, \alpha \type \entails S \type
           }
           {\Gamma \entails (\Rec{\alpha} S) \type}
    \end{mathpar}

  \caption{Type well-formedness in the economical type system}
  \FLabel{fig:econ-wf}
\end{figure}

\begin{figure}[htbp]
  \centering
  
  \judgbox{G \entails A \type}
               {Target type $A$ is well-formed}
  \vspace{-3ex}
  \begin{mathpar}
      \Infer{}
           {}
           {G \entails \unitty \type}
      \and
      \Infer{}
           {(\alpha \type) \in G}
           {G \entails \alpha \type}
      \and
      \Infer{}
           {G, \alpha \type \entails A \type}
           {G \entails (\All{\alpha}A) \type}
      \\
      \Infer{}
          {G \entails A \type
          }
          {
            G \entails (\thunkty A) \type
          } 
      \and
      \Infer{}
           { \arrayenvbl{
                 G \entails A_1 \type
                 \\
                 G \entails A_2 \type
             }
           }
          { 
              \arrayenvc{
                 G \entails (A_1 \arr A_2) \type
                 \\
                 G \entails (A_1 * A_2) \type
                 \\
                 G \entails (A_1 + A_2) \type
              }
          }
      \and
      \Infer{}
           { G, \alpha \type \entails A \type
           }
           {G \entails (\Rec{\alpha} A) \type}
    \end{mathpar}

  \caption{Type well-formedness in the target type system}
  \FLabel{fig:target-wf}
\end{figure}

\section{Proofs}
\Label{apx:proofs-wf}

\subsection*{Notation}

We present some proofs in a line-by-line style, with the justification for each
claim in the rightmost column.  We highlight with $\Hand$ what we needed to show;
this is most useful when trying to prove statements with several conclusions,
like ``if\dots then Q1 and Q2 and Q3'', where we might derive Q2 early (say, directly
from the induction hypothesis) but need several more steps to show Q1 and Q3.

\addtocounter{subsection}{1}  %
\addtocounter{subsection}{1}  %

\subsection{Economical Type System}
\Label{apx:econ}

\begin{lemma}[Suspension Points]
\Label{lem:econ-susp-point}
~
\begin{enumerate}[(1)]
    \item 
       If $\chktypee{\Gamma, \xsyn x {\VAL} {\susp{\V} S'}, \Gamma'}{e}{\VVAR}{S}$
       \\
       then
       $\chktypee{\Gamma, \xsyn x {\VAL} {S'}, \Gamma'}{e}{\VVAR}{S}$.

    \item
       If $\syntypee{\Gamma, \xsyn x {\VAL} {\susp{\V} S'}, \Gamma'}{e}{\VVAR}{S}$
       \\
       then
       $\syntypee{\Gamma, \xsyn x {\VAL} {S'}, \Gamma'}{e}{\VVAR}{S}$.
\end{enumerate}
\end{lemma}
\begin{proof}
  By mutual induction on the given derivation.
  The \Rvar case uses \Rsuspintro (first conclusion).
\end{proof}

\begin{lemma}[Economizing (Types)]
\Label{lem:econ-type}
~\\
  If
     $\gamma \entails \tau \type$
  then
     $\ctxecontrans{\gamma} \entails \econtrans{\tau} \type$. 
\end{lemma}
\begin{proof}
  By induction on the derivation of $\gamma \entails \tau \type$
  (Fig. \ref{fig:impartial-wf}).
\end{proof}

\begin{lemma}[Economizing (Eval.\ Order)]
\Label{lem:econ-eo}
~\\
  If
     $\gamma \entails \epsilon \eo$
  then
     $\ctxecontrans{\gamma} \entails \epsilon \eo$.
\end{lemma}
\begin{proof}
  By a straightforward induction on $\gamma$.
\end{proof}

\thmecon*
\begin{proof}
  By induction on the given derivation.

  \begin{itemize}
       \DerivationProofCase{\Iarrintro}
              {\chktype{\gamma, (\xsyn x {\valof\epsilon} \tau_1)}{e_0}{\VVAR}{\tau_2}
              }
              {\chktype{\gamma}{(\explam{x} e_0)}{\VAL}{(\tau_1 \garr{\epsilon} \tau_2)}}

             \begin{llproof}
               \chktypePf{\gamma, \xsyn x {\valof\epsilon} \tau_1}{e_0}{\VVAR}{\tau_2}  {Subderivation}
               \chktypeePf{\ctxecontrans{\gamma, \xsyn x {\valof\epsilon} \tau_1}}{\expecontrans{e_0}}{\VVAR}{\econtrans{\tau_2}}  {By i.h.}
               \chktypeePf{\ctxecontrans{\gamma}, x:\left(\susp{\epsilon} \econtrans{\tau_1}\right)}{\expecontrans{e_0}}{\VVAR}{\econtrans{\tau_2}}  {By def.\ of $\ctxecontrans{-}$}
               \chktypeePf{\ctxecontrans{\gamma}}{(\explam{x} \expecontrans{e_0})}{\VAL}{\left(\susp{\epsilon} \econtrans{\tau_1}\right) \arr \econtrans{\tau_2}}
                     {By \Rarrintro}
             \Hand
               \chktypeePf{\ctxecontrans{\gamma}}{\expecontrans{\explam{x} e_0}}{\VAL}{\econtrans{\tau_1 \garr{\epsilon} \tau_2}}
                     {By def.\ of $\econtrans{-}$}
             \end{llproof}
       
       \DerivationProofCase{\Iarrelim}
              {\syntype{\gamma}{e_1}{\VVAR_1}{(\tau_1 \garr{\epsilon} \tau)}
               \\
               \chktype{\gamma}{e_2}{\VVAR_2}{\tau_1}
              }
              {\syntype{\gamma}{(\expapp{e_1}{e_2})}{\NONVAL}{\tau}
              }

              \begin{llproof}
                \syntypePf{\gamma}{e_1}{\VVAR_1}{(\tau_1 \garr{\epsilon} \tau)}
                       {Subderivation}
                \proofsep
                \syntypeePf{\ctxecontrans{\gamma}}{\expecontrans{e_1}}{\VVAR_1}{\econtrans{\tau_1 \garr{\epsilon} \tau}}
                       {By i.h.}
                \syntypeePf{\ctxecontrans{\gamma}}{\expecontrans{e_1}}{\VVAR_1}{\left(\susp{\epsilon} \econtrans{\tau_1}\right) \arr \econtrans{\tau}}
                       {By def.\ of $\econtrans{-}$}
                \proofsep
                \chktypePf{\gamma}{e_2}{\VVAR_2}{\tau_1}
                       {Subderivation}
                \chktypeePf{\ctxecontrans{\gamma}}{\expecontrans{e_2}}{\VVAR_2}{\econtrans{\tau_1}}
                       {By i.h.}
                \chktypeePf{\ctxecontrans{\gamma}}{\expecontrans{e_2}}{\VVAR_2'}{\susp{\epsilon} \econtrans{\tau_1}}
                       {By \Rsuspintro}
              \Hand
                 \syntypeePf{\ctxecontrans{\gamma}}{\expecontrans{\expapp{e_1}{e_2}}}{\NONVAL}{\econtrans{\tau}}
                        {By \Rarrelim and def.\ of $\expecontrans{-}$}
              \end{llproof}

       \DerivationProofCase{\Iunitintro}
             {}
             {\chktype{\gamma}{\unit}{\VAL}{\unitty}}

                 \begin{llproof}
                   \chktypeePf{\ctxecontrans{\gamma}}{\unit}{\VAL}{\unitty}
                          {By \Runitintro}
                 \Hand
                   \chktypeePf{\ctxecontrans{\gamma}}{\expecontrans{\unit}}{\VAL}{\econtrans{\unitty}}
                          {By def.\ of $\econtrans{-}$}
                 \end{llproof}

        \DerivationProofCase{\Iallintro}
              {\chktype{\gamma, \alpha \type}{e_0}{\VAL}{\tau_0}
              }
              {\chktype{\gamma}{\tylam{\alpha} e_0}{\VAL}{\All{\alpha} \tau_0}}

              \begin{llproof}
                \chktypePf{\gamma, \alpha \type}{e_0}{\VAL}{\tau_0}
                       {Subderivation}
                \chktypeePf{\ctxecontrans{\gamma, \alpha \type}}{\expecontrans{e_0}}{\VAL}{\econtrans{\tau_0}}
                       {By i.h.}
                \chktypeePf{\ctxecontrans{\gamma}, \alpha \type}{\expecontrans{e_0}}{\VAL}{\econtrans{\tau_0}}
                       {By def.\ of $\ctxecontrans{-}$}
                \chktypeePf{\ctxecontrans{\gamma}}{\tylam{\alpha} \expecontrans{e_0}}{\VAL}{\All {\alpha} \econtrans{\tau_0}}
                       {By \Rallintro}
              \Hand
                \chktypeePf{\ctxecontrans{\gamma}}{\expecontrans{\tylam{\alpha} e_0}}{\VAL}{\econtrans{\All {\alpha} \tau_0}}
                       {By def.\ of $\econtrans{-}$}
              \end{llproof}

        \DerivationProofCase{\Iallelim}
             {\syntype{\gamma}{e_0}{\VVAR}{\All{\alpha} \tau_0}
              \\
              \gamma \entails \tau' \type}
             {\syntype{\gamma}{\tyapp{e_0}{\tau'}}{\VVAR}{[\tau' / \alpha]\tau_0}}

             \begin{llproof}
               \syntypePf{\gamma}{e_0}{\VVAR}{\All{\alpha} \tau_0}
                      {Subderivation}
               \syntypeePf{\ctxecontrans{\gamma}}{\expecontrans{e_0}}{\VVAR}{\econtrans{\All{\alpha} \tau_0}}
                      {By i.h.}
               \syntypeePf{\ctxecontrans{\gamma}}{\expecontrans{e_0}}{\VVAR}{\All{\alpha} \econtrans{\tau_0}}
                      {By def.\ of $\econtrans{-}$}
               \ePf{\gamma}{ \tau' \type}  {Subderivation}
               \ePf{\ctxecontrans{\gamma}}{\econtrans{\tau'} \type}  {By \Lemmaref{lem:econ-type}}
               \syntypeePf{\ctxecontrans{\gamma}}{\tyapp{\expecontrans{e_0}}{\econtrans{\tau'}}}{\VVAR}{[\econtrans{\tau'}/\alpha]\econtrans{\tau_0}}
                      {By \Rallelim}
             \Hand
               \syntypeePf{\ctxecontrans{\gamma}}{\expecontrans{\tyapp{e_0}{\tau'}}}{\VVAR}{\econtrans{[\tau'/\alpha]\tau_0}}
                      {By properties of $\econtrans{-}$ and substitution}
             \end{llproof}

        \DerivationProofCase{\Ialleointro}
              {\chktype{\gamma, \eovar \eo}{e}{\VAL}{\tau_0}
              }
              {\chktype{\gamma}{e}{\VAL}{\alleo{\eovar} \tau_0}}

             \begin{llproof}
               \chktypePf{\gamma, \eovar \eo}{e}{\VAL}{\tau_0}  {Subderivation}
               \chktypeePf{\ctxecontrans{\gamma, \eovar \eo}}{\expecontrans{e}}{\VAL}{\econtrans{\tau_0}}
                      {By i.h.}
               \chktypeePf{\ctxecontrans{\gamma}, \eovar \eo}{\expecontrans{e}}{\VAL}{\econtrans{\tau_0}}
                      {By def.\ of $\econtrans{-}$}
               \chktypeePf{\ctxecontrans{\gamma}}{\expecontrans{e}}{\VAL}{\alleo{\eovar} \econtrans{\tau_0}}
                      {By \Ralleointro}
             \Hand
               \chktypeePf{\ctxecontrans{\gamma}}{\expecontrans{e}}{\VAL}{\econtrans{\alleo{\eovar} \tau_0}}
                      {By def.\ of $\econtrans{-}$}
             \end{llproof}

        \DerivationProofCase{\Ialleoelim}
             {\syntype{\gamma}{e}{\VVAR}{\alleo{\eovar} \tau_0}
               \\
               \gamma \entails \epsilon \eo}
             {\syntype{\gamma}{e}{\VVAR}{[\epsilon / \eovar]\tau_0}}

             \begin{llproof}
               \syntypePf{\gamma}{e}{\VVAR}{\alleo{\eovar} \tau_0}
                     {Subderivation}
               \syntypeePf{\ctxecontrans{\gamma}}{\expecontrans{e}}{\VVAR}{\econtrans{\alleo{\eovar} \tau_0}}
                     {By i.h.}
               \proofsep
               \ePf{\gamma}{\epsilon \eo}  {Subderivation}
               \ePf{\ctxecontrans{\gamma}}{\epsilon \eo}  {By \Lemmaref{lem:econ-eo}}
               \proofsep
               \syntypeePf{\ctxecontrans{\gamma}}{\expecontrans{e}}{\VVAR}{[\epsilon / \eovar]\econtrans{\tau_0}}
                     {By \Ralleoelim}
             \Hand
               \syntypeePf{\ctxecontrans{\gamma}}{\expecontrans{e}}{\VVAR}{\econtrans{[\epsilon / \eovar]\tau_0}}
                     {By properties of $\econtrans{-}$ and substitution}
             \end{llproof}

        \DerivationProofCase{\Ivar}
              {(\xsyn x \VVAR \tau) \in \gamma}
              {\syntype \gamma x \VVAR \tau}

              \begin{llproof}
                \inPf{(\xsyn x \VVAR \tau)}{\gamma}  {Premise}
              \end{llproof}

              We distinguish cases of $\VVAR$:

              \begin{itemize}
              \item If $\VVAR = \VAL$, then:

                      \begin{llproof}
                        \inPf{(x : \susp{\V} \econtrans{\tau})} {\ctxecontrans{\gamma}}  {By def.\ of $\ctxecontrans{-}$}
                        \syntypeePf
                              {\ctxecontrans{\gamma}}
                              {x}
                              {\VAL}
                              {\susp{\V} \econtrans{\tau}}
                              {By \Rvar}
                     \Hand
                        \syntypeePf
                              {\ctxecontrans{\gamma}}
                              {x}
                              {\VAL}
                              {\econtrans{\tau}}
                              {By \Rsuspelim{\V}}
                      \end{llproof}

              \item If $\VVAR = \NONVAL$, then:

                      \begin{llproof}
                        \inPf{(x : \susp{\N} \econtrans{\tau})} {\ctxecontrans{\gamma}}  {By def.\ of $\ctxecontrans{-}$}
                        \syntypeePf
                              {\ctxecontrans{\gamma}}
                              {x}
                              {\VAL}
                              {\susp{\N} \econtrans{\tau}}
                              {By \Rvar}
                     \Hand
                        \syntypeePf
                              {\ctxecontrans{\gamma}}
                              {x}
                              {\NONVAL}
                              {\econtrans{\tau}}
                              {By \Rsuspelim{\epsilon}}
                      \end{llproof}
              \end{itemize}

       \DerivationProofCase{\Ifixvar}
            {(\xsyn u \NONVAL \tau) \in \gamma}
            {\syntype \gamma u \NONVAL \tau}

            \begin{llproof}
              \inPf{(u : \econtrans{\tau})}{\ctxecontrans{\gamma}}   {By def.\ of $\ctxecontrans{-}$}
            \Hand
              \syntypeePf {\ctxecontrans \gamma} u \NONVAL {\econtrans \tau}
                    {By \Rfixvar}
            \end{llproof}

        \DerivationProofCase{\Ifix}
              {\chktype{\gamma, \xsyn{u}{\NONVAL}{\tau}}{e_0}{\VVAR'}{\tau}}
              {\chktype{\gamma}{(\Fix{u} e_0)}{\NONVAL}{\tau}}

              \begin{llproof}
                \chktypeePf
                      {\ctxecontrans{\gamma, \xsyn{u}{\NONVAL}{\tau}}}
                      {e_0}
                      {\VVAR}
                      {\econtrans{\tau}}
                      {By i.h.}
                \chktypeePf
                      {\ctxecontrans{\gamma}, \xtypeofe{u}{\econtrans{\tau}}}
                      {e_0}
                      {\VVAR}
                      {\econtrans{\tau}}
                      {By def. of $\econtrans{-}$}
             \Hand
                \chktypeePf
                      {\ctxecontrans{\gamma}}
                      {(\Fix{u} e_0)}
                      {\NONVAL}
                      {\econtrans{\tau}}
                      {By \Rfix}
              \end{llproof}

        \DerivationProofCase{\Isub}
              {\syntype{\gamma}{e}{\VVAR}{\tau}}
              {\chktype{\gamma}{e}{\VVAR}{\tau}}

              By i.h.\ and \Rsub.

        \DerivationProofCase{\Ianno}
              {\chktype{\gamma}{e_0}{\VVAR}{\tau}}
              {\syntype{\gamma}{\Anno{e_0}{\tau}}{\VVAR}{\tau}}

              By i.h.\ and \Ranno.

        \DerivationProofCase{\Iprodintro}
                 {\chktype{\gamma}{e_1}{\VVAR_1}{\tau_1}
                   \\
                   \chktype{\gamma}{e_2}{\VVAR_2}{\tau_2}
                  }
                  {\chktype{\gamma}{\Pair{e_1}{e_2}}{\VVAR_1 \join \VVAR_2}{(\tau_1 *^\epsilon \tau_2)}}

                \begin{llproof}
                  \chktypePf{\gamma}{e_1}{\VVAR_1}{\tau_1}
                        {Subderivation}
                  \chktypeePf{\ctxecontrans{\gamma}}{\expecontrans{e_1}}{\VVAR_1}{\econtrans{\tau_1}}
                        {By i.h.}
                  \chktypeePf{\ctxecontrans{\gamma}}{\expecontrans{e_1}}{\VVAR_1}{\susp{\epsilon} \econtrans{\tau_1}}
                        {By \Rsuspintro}
                  \proofsep
                  \chktypeePf{\ctxecontrans{\gamma}}{\expecontrans{e_2}}{\VVAR_2}{\susp{\epsilon} \econtrans{\tau_2}}
                        {Similar}
                  \proofsep
                  \chktypeePf{\ctxecontrans{\gamma}}{\Pair{\expecontrans{e_1}}{\expecontrans{e_2}}}{\VVAR_1\join \VVAR_2}{\left(\susp{\epsilon} \econtrans{\tau_1}\right) * \left(\susp{\epsilon} \econtrans{\tau_2}\right)}
                        {By \Rprodintro}
               \Hand
                  \chktypeePf{\ctxecontrans{\gamma}}{\expecontrans{\Pair{e_1}{e_2}}}{\VVAR_1 \join \VVAR_2}{\econtrans{\tau_1 *^\epsilon \tau_2}}
                        {By def.\ of $\econtrans{-}$}
                \end{llproof}

         \DerivationProofCase{\Iprodelim{k}}
                 {
                   \syntype{\gamma}
                           {e_0}
                           {\VVAR}
                           {(\tau_1 *^\epsilon \tau_2)}
                 }
                 {
                   \syntype{\gamma}
                        {(\Proj{k} e_0)}
                        {\NONVAL}
                        {\tau_k}
                 }

               \begin{llproof}
                 \syntypePf{\gamma}
                         {e_0}
                         {\VVAR}
                         {(\tau_1 *^\epsilon \tau_2)}
                         {Subderivation}
                 \syntypeePf{\ctxecontrans{\gamma}}
                         {\expecontrans{e_0}}
                         {\VVAR}
                         {\econtrans{\tau_1 *^\epsilon \tau_2}}
                         {By i.h.}
                 \syntypeePf{\ctxecontrans{\gamma}}
                         {\expecontrans{e_0}}
                         {\VVAR}
                         {\left(\susp{\epsilon} \econtrans{\tau_1}\right) * \left(\susp{\epsilon} \econtrans{\tau_2}\right)}
                         {By def.\ of $\econtrans{-}$}
                 \syntypeePf{\ctxecontrans{\gamma}}
                         {(\Proj{k} \expecontrans{e_0})}
                         {\NONVAL}
                         {\left(\susp{\epsilon} \econtrans{\tau_k}\right)}
                         {By \Rprodelim{k}}
               \Hand
                 \syntypeePf{\ctxecontrans{\gamma}}
                         {\expecontrans{\Proj{k} e_0}}
                         {\NONVAL}
                         {\econtrans{\tau_k}}
                         {By \Rsuspelim{\epsilon} and def.\ of $\expecontrans{-}$}
               \end{llproof}

         \DerivationProofCase{\Isumintro{k}}
               {\chktype{\gamma}{e_0}{\VVAR}{\tau_k}
               }
               {\chktype{\gamma}{(\Inj{k} e_0)}{\VVAR}{(\tau_1 +^\epsilon \tau_2)}}

               \begin{llproof}
                 \chktypePf{\gamma}{e_0}{\VVAR}{\tau_k}   {Subderivation}
                 \chktypeePf{\ctxecontrans{\gamma}}{\expecontrans{e_0}}{\VVAR}{\econtrans{\tau_k}}   {By i.h.}
                 \chktypeePf{\ctxecontrans{\gamma}}{(\Inj{k} \expecontrans{e_0})}{\VVAR}{\econtrans{\tau_1} + \econtrans{\tau_2}}
                        {By \Rsumintro{k}}
                 \chktypeePf{\ctxecontrans{\gamma}}{(\Inj{k} \expecontrans{e_0})}{\VVAR}{\susp{\epsilon}(\econtrans{\tau_1} + \econtrans{\tau_2})}
                        {By \Rsuspintro (first conclusion)}
               \Hand
                 \chktypeePf{\ctxecontrans{\gamma}}{\expecontrans{\Inj{k} e_0}}{\VVAR}{\econtrans{\tau_1} + \econtrans{\tau_2}}
                        {By def.\ of $\econtrans{-}$}
               \end{llproof}

         \DerivationProofCase{\Isumelim}
              {   \syntype{\gamma}{e_0}{\VVAR_0}{(\tau_1 +^\epsilon \tau_2)}
                  \\
                  \arrayenvbl{
                      \chktype{\gamma, \xsyn{x_1}{\VAL}{\tau_1}}{e_1}{\VVAR_1}{\tau}
                      \\
                      \chktype{\gamma, \xsyn{x_2}{\VAL}{\tau_2}}{e_2}{\VVAR_2}{\tau}
                  }
              }
              {\chktype{\gamma}{\Casesum{e_0}{x_1}{e_1}{x_2}{e_2}}{\NONVAL}{\tau}
              }

                \begin{llproof}
                  \syntypePf{\gamma}{e_0}{\VVAR_0}{(\tau_1 +^\epsilon \tau_2)}  {Subderivation}
                  \syntypeePf{\ctxecontrans{\gamma}}{\expecontrans{e_0}}{\VVAR_0}{\econtrans{\tau_1 +^\epsilon \tau_2}}  {By i.h.}
                  \syntypeePf{\ctxecontrans{\gamma}}{\expecontrans{e_0}}{\VVAR_0}{\susp{\epsilon}(\econtrans{\tau_1} + \econtrans{\tau_2})}  {By def.\ of $\econtrans{-}$}
                  \syntypeePf{\ctxecontrans{\gamma}}{\expecontrans{e_0}}{\NONVAL}{(\econtrans{\tau_1} + \econtrans{\tau_2})}  {By \Rsuspelim{\epsilon}}
                  \proofsep
                  \syntypePf{\gamma, \xsyn{x_1}{\VAL}{\tau_1}}{e_1}{\VVAR_1}{\tau}   {Subderivation}
                  \syntypeePf{\ctxecontrans{\gamma}, x_1 : \susp{\V} \econtrans{\tau_1}}{\expecontrans{e_1}}{\VVAR_1}{\econtrans{\tau}}
                          {By i.h.\ and def.\ of $\ctxecontrans{-}$}
                  \syntypeePf{\ctxecontrans{\gamma}, x_1 : \econtrans{\tau_1}}{\expecontrans{e_1}}{\VVAR_1}{\econtrans{\tau}}
                          {By \Lemmaref{lem:econ-susp-point}}
                  \proofsep
                  \syntypeePf{\ctxecontrans{\gamma}, x_2 : \econtrans{\tau_2}}{\expecontrans{e_2}}{\VVAR_2}{\econtrans{\tau}}
                          {Similarly}
                  \proofsep
               \Hand
                  \syntypeePf
                      {\ctxecontrans{\gamma}}
                      {\expecontrans{\Casesum{e_0}{x_1}{e_1}{x_2}{e_2}}}
                      {\NONVAL}
                      {\econtrans{\tau}}
                      {By \Rsumelim}
                \end{llproof}

         \DerivationProofCase{\Irecintro}
             {\chktype{\gamma}{e}{\VVAR}{\big[(\rec{\epsilon}{\alpha}\tau_0)\big/\alpha\big]\tau_0}
             }
             {\chktype{\gamma}{e}{\VVAR}{\rec{\epsilon}{\alpha} \tau_0}}

               \begin{llproof}
                 \chktypePf{\gamma}{e}{\VVAR}{\big[(\rec{\epsilon}{\alpha}\tau_0)/\alpha\big]\tau_0}
                         {Subderivation}
                 \chktypeePf{\ctxecontrans{\gamma}}{\expecontrans{e}}{\VVAR}{\econtrans{\big[(\rec{\epsilon}{\alpha}\tau_0)/\alpha\big]\tau_0}}
                         {By i.h.}
                 \chktypeePf{\ctxecontrans{\gamma}}{\expecontrans{e}}{\VVAR}{\big[\econtrans{\rec{\epsilon}{\alpha}\tau_0} / \alpha\big]\econtrans{\tau_0}}
                         {By a property of substitution/$\econtrans{-}$}
                 \chktypeePf{\ctxecontrans{\gamma}}{\expecontrans{e}}{\VVAR}
                         {\big[\left(\Rec{\alpha} \susp{\epsilon} \econtrans{\tau_0}\right) / \alpha\big]
                                 \econtrans{\tau_0}}
                         {By def.\ of $\econtrans{-}$}
                 \proofsep
                 \chktypeePf{\ctxecontrans{\gamma}}{\expecontrans{e}}{\VVAR}
                         {\Rec{\alpha} \susp{\epsilon} \econtrans{\tau_0}}
                         {By \Rrecintro}                 
               \Hand
                 \chktypeePf{\ctxecontrans{\gamma}}{\expecontrans{e}}{\VVAR}
                         {\econtrans{\rec{\epsilon}{\alpha} \tau_0}}
                         {By def.\ of $\econtrans{-}$}
               \end{llproof}

         \DerivationProofCase{\Irecelim}
                {   \syntype{\gamma}{e}{\VVAR_0}{\rec{\epsilon}{\alpha} \tau_0}
                }
                {
                  \syntype{\gamma}{e}{\NONVAL}{\big[(\rec{\epsilon}{\alpha}\tau_0)/\alpha\big]\tau_0}
                }

               \begin{llproof}
                 \syntypePf{\gamma}{e}{\VVAR_0}{\rec{\epsilon}{\alpha} \tau_0}
                         {Subderivation}
                 \syntypeePf{\ctxecontrans{\gamma}}{\expecontrans{e}}{\VVAR_0}{\econtrans{\rec{\epsilon}{\alpha} \tau_0}}
                         {By i.h.}
                 \syntypeePf{\ctxecontrans{\gamma}}{\expecontrans{e}}{\VVAR_0}
                         {\Rec{\alpha} \susp{\epsilon} \econtrans{\tau_0}}
                         {By def.\ of $\econtrans{-}$}
                 \proofsep
                 \syntypeePf{\ctxecontrans{\gamma}}{\expecontrans{e}}{\NONVAL}
                         {\big[\left(\Rec{\alpha} \susp{\epsilon} \econtrans{\tau_0}\right) \;/\; \alpha\big]\,\susp{\epsilon} \econtrans{\tau_0}}
                         {By \Rrecelim}
                 \syntypeePf{\ctxecontrans{\gamma}}{\expecontrans{e}}{\NONVAL}
                         {\big[\econtrans{\rec{\epsilon}{\alpha} \tau_0} \;/\; \alpha\big]\,\susp{\epsilon} \econtrans{\tau_0}}
                         {By def.\ of $\econtrans{-}$}
                 \syntypeePf{\ctxecontrans{\gamma}}{\expecontrans{e}}{\NONVAL}
                         {\susp{\epsilon} \,\big[\econtrans{\rec{\epsilon}{\alpha} \tau_0} \,/\, \alpha\big]\econtrans{\tau_0}}
                         {By a property of substitution}
                 \syntypeePf{\ctxecontrans{\gamma}}{\expecontrans{e}}{\NONVAL}
                         {\left[\econtrans{\rec{\epsilon}{\alpha} \tau_0} \;/\; \alpha\right]\econtrans{\tau_0}}
                         {By \Rsuspelim{\epsilon}}
               \Hand
                 \syntypeePf{\ctxecontrans{\gamma}}{\expecontrans{e}}{\NONVAL}
                         {\econtrans{\big[(\rec{\epsilon}{\alpha}\tau_0)/\alpha\big]\tau_0}}
                         {By a property of substitution/$\econtrans{-}$}
               \end{llproof}
  \qedhere
  \end{itemize}
\end{proof}

\addtocounter{subsection}{1}  %

\subsection{Elaboration}
\Label{apx:elab}

\lemvaluemono*
\begin{proof}
  By induction on the given derivation.

  For any rule whose conclusion has $\VAL$, we already have our result.
  This takes care of \Eunitintro, \Eallintro, \Ealleointro, the second conclusion
  of \Esuspintro, \Evar, and \Earrintro.
  Rules whose conclusions have target terms that can never be a value are impossible,
  which takes care of \Eallelim, \Ealleoelim, \Esuspelim{\N}, \Efixvar, \Efix, \Earrelim,
  \Eprodelim{k}, \Esumelim, and \Erecelim.
  We are left with:

  \begin{itemize}
  \ProofCaseRule{\Esuspintro (first conclusion)}  The result follows by i.h.\ and \Esuspintro.

  \ProofCaseRule{\Eprodintro}
         We have $W = \tpair{W_1}{W_2}$.  By i.h.\ twice, $\VVAR_1 = \VAL$
         and $\VVAR_2 = \VAL$.  Applying \Eprodintro gives the result (using
         $\VAL \join \VAL = \VAL$).

  \ProofCasesRules{\Esuspelim{\V}, \Esumintro{k}, \Erecintro}
     The result follows by i.h.\ and applying the same rule.  
  \qedhere
  \end{itemize}
\end{proof}

\lemelabvaluability*
\begin{proof}
  By induction on the given derivation.

  \begin{itemize}
      \ProofCasesRules{\Evar, \Eunitintro, \Earrintro}
          Immediate.

      \ProofCasesRules{%
              \Esuspintro ($\N$ conclusion),
              \Esuspelim{\N},
              \Efix, \Efixvar,
              \Earrelim, \Eprodelim{k}, \Esumelim,
              \Erecelim
          }

          Impossible: these rules cannot elaborate values.

      \ProofCaseRule{\Ealleointro}
          By i.h., $M_1$ and $M_2$ are valuable; therefore $\tpair{M_1}{M_2}$ is valuable.
      \ProofCaseRule{\Ealleoelim}
          By i.h., $M_0$ is valuable; therefore $\tproj{1}{M_0}$ and $\tproj{2}{M_0}$
          are valuable.

      \ProofCaseRule{\Eprodintro}
          Similar to the \Ealleointro case.

      \ProofCasesRules{\Eallintro, \Eallelim}
          By i.h., $M_0$ is valuable; therefore $\ttylam M_0$ and $\ttyapp M_0$ are valuable.

      \ProofCasesRules{\Esuspintro ($\V$ conclusion), \Esuspelim{\V}}
          By i.h.

      \ProofCaseRule{\Esumintro{k}}
          By i.h., $M_0$ is valuable; therefore $\tinj{k} M_0$ is valuable.          

      \ProofCaseRule{\Erecintro}
          By i.h., $M_0$ is valuable; therefore $\troll M_0$ is valuable.
  \qedhere
  \end{itemize}
\end{proof}

\lemsubsteo*
\begin{proof}  
  Part (1): By induction on the first derivation.  Part (1) does not depend on the other parts.

  Parts (2) and (3): By induction on the given derivation, using part (1):

  \begin{itemize}
  \ProofCaseRule{\Rallintro}  By i.h.\ and \Rallintro.

  \DerivationProofCase{\Rallelim}
       {\syntypee{\Gamma, \eovar \eo, \Gamma'}{e}{\VVAR}{\All{\alpha} S_0}
        \\
        \Gamma, \eovar \eo, \Gamma' \entails S' \type}
       {\syntypee{\Gamma, \eovar \eo, \Gamma'}{e}{\VVAR}{[S' / \alpha]S_0}}
  
       \begin{llproof}
         \syntypeePf{\Gamma, \eovar \eo, \Gamma'}{e}{\VVAR}{\All{\alpha} S_0}  {Subderivation}
         \syntypeePf{\Gamma, [\epsilon/\eovar]\Gamma'}{e}{\VVAR}{[\epsilon/\eovar](\All{\alpha} S_0)}  {By i.h.}
         \syntypeePf{\Gamma, [\epsilon/\eovar]\Gamma'}{e}{\VVAR}{\All{\alpha} [\epsilon/\eovar]S_0}  {By def.\ of subst.}
         \ePf{\Gamma, \eovar \eo, \Gamma'}{S'}   {Subderivation}
         \ePf{\Gamma, [\epsilon/\eovar]\Gamma'}{[\epsilon/\eovar]S'}   {By part (1)}
         \syntypeePf{\Gamma, [\epsilon/\eovar]\Gamma'}{e}{\VVAR}{\big[[\epsilon/\eovar]S' / \alpha\big][\epsilon/\eovar]S_0}   {By \Rallelim}
       \Hand  \syntypeePf{\Gamma, [\epsilon/\eovar]\Gamma'}{e}{\VVAR}{[\epsilon/\eovar][S' / \alpha]S_0}   {By def.\ of subst.}
       \end{llproof}

  \DerivationProofCase{\Rvar}
        {(x : S) \in (\Gamma, \eovar \eo, \Gamma')}
        {\syntypee{\Gamma, \eovar \eo, \Gamma'}{x}{\VAL}{S}}

        Follows from the definition of substitution on contexts.

  \ProofCaseRule{\Rfixvar}  Similar to the \Rvar case.

  \end{itemize}

  The remaining cases are straightforward, using the i.h.\ and properties
  of substitution.
\end{proof}

\begin{lemma}[Type substitution]
\Label{lem:elab-type-subst}
~
\begin{enumerate}[(1)]
    \item 
             If
             $\Gamma \entails S' \type$
             and
             $\Gamma, \alpha \type \entails S \type$
             then
             $\Gamma \entails [S'/\alpha]S \type$.
    \item 
             If
             $\Gamma \entails S' \type$
             and
             $\Gamma, \alpha \type \entails \elab{e}{\VVAR}{S}{M}$
             then
             $\Gamma \entails \elab{e}{\VVAR}{[S'/\alpha]S}{M}$.
\end{enumerate}
\end{lemma}
\begin{proof}
  In each part, by induction on the second derivation.
  In part (2), the \Eallelim case uses part (1).
\end{proof}

\lemelabexprsubst*
\begin{proof}
  Part (1): By induction on the given derivation.  In the \Evar case,
  use \Lemmaref{lem:value-mono} to get 
  $\Gamma \entails \elab{e_1}{\VAL}{S_1}{W}$. By weakening,
  $\Gamma, \Gamma' \entails \elab{e_1}{\VAL}{S_1}{W}$, which is
  $\Gamma, \Gamma' \entails \elab{[e_1/x]x}{\VAL}{S}{[W/x]M}$.

  Part (2): By induction on the given derivation.  Note that in the \Efixvar case,
  $\VVAR_2 = \NONVAL$.
\end{proof}

\thmelabtypesoundness*
\begin{proof}
  By induction on the size of the given derivation.  If $\VVAR' = \VVAR$,
  we often don't bother to state $\VVAR \valleq \VVAR$ explicitly.
  
  \begin{itemize}
        \DerivationProofCase{\Rvar}
            {(\xtypeofe{x}{S}) \in \Gamma}
            {\syntypee{\Gamma}{x}{\VAL}{S}}

              \begin{llproof}
               \inPf{(x : S)}{\Gamma}  {Premise}
              \Hand    \elabPf {\Gamma} {\er x} {\VVAR}{S} {x}   {By \Evar}
               \inPf{(x : \tytrans{S})}{\ctxtrans{\Gamma}}
                        {By def.\ of $\ctxtrans{-}$}
              \Hand      \typeoftPf {\ctxtrans{\Gamma}}
                              {x}
                              {\tytrans{S}}
                              {By \Mvar}
              \end{llproof}

       \ProofCaseRule{\Rfixvar}  Similar to the \Rvar case.

       \DerivationProofCase{\Rfix}
            {\chktypee{\Gamma, \xtypeofe{u}{S}}{e_0}{\VVAR_0}{S}}
            {\chktypee{\Gamma}{(\Fix{u} e_0)}{\NONVAL}{S}}

             \begin{llproof}
               \chktypeePf{\Gamma, \xtypeofe{u}{S}}{e_0}{\VVAR_0}{S}  {Subderivation}
               \elabPf {\Gamma, \xtypeofe{u}{S}}
                              {\er{e_0}}
                              {\VVAR_0'}
                              {S}
                              {M_0}
                              {By i.h.}
               \typeoftPf{\ctxtrans{\Gamma, \xtypeofe{u}{S}}}{M_0}{\tytrans{S}} {\ditto}
               \typeoftPf{\ctxtrans{\Gamma}, u : \tytrans{S}}{M_0}{\tytrans{S}} {By def.\ of $\ctxtrans{-}$}
               \proofsep
               \Hand  \elabPf {\Gamma, \xtypeofe{u}{S}}
                              {\er{e_0}}
                              {\NONVAL}
                              {S}
                              {\tfix{u} M_0}
                              {By \Efix}
               \Hand  \typeoftPf{\ctxtrans{\Gamma}} {(\tfix{u} M_0)} {\tytrans{S}}
                       {By \Mfix}
             \end{llproof}

      \DerivationProofCase{\Rsub}
            {\syntypee{\Gamma}{e}{\VVAR}{S}}
            {\chktypee{\Gamma}{e}{\VVAR}{S}}

            \begin{llproof}
                \syntypeePf{\Gamma}{e}{\VVAR}{S}  {Subderivation}
            \Hand \elabPf{\Gamma}{\er e}{\VVAR'}{S}{M}  {By i.h.}
            \Hand \valleqPf{\VVAR'}{\VVAR}  {\ditto}
            \Hand \typeoftPf{\ctxtrans{\Gamma}}{M}{\tytrans{S}}  {\ditto}
            \end{llproof}

      \DerivationProofCase{\Ranno}
            {\chktypee{\Gamma}{e_0}{\VVAR}{S}}
            {\syntypee{\Gamma}{\Anno{e_0}{S}}{\VVAR}{S}}
            
            \begin{llproof}
                \chktypeePf{\Gamma}{e_0}{\VVAR}{S}  {Subderivation}
                \elabPf{\Gamma}{\er{e_0}}{\VVAR'}{S}{M}  {By i.h.}
            \Hand \valleqPf{\VVAR'}{\VVAR}  {\ditto}
            \Hand \typeoftPf{\ctxtrans{\Gamma}}{M}{\tytrans{S}}  {\ditto}
            \Hand \elabPf{\Gamma}{\er{\Anno{e_0}{S}}}{\VVAR'}{S}{M}  {By def.\ of $\er{-}$}
            \end{llproof}

       \DerivationProofCase{\Runitintro}
                {}
                {\chktypee{\Gamma}{\unit}{\VAL}{\unitty}}

             \begin{llproof}
             \Hand  \elabPf {\Gamma} {\er{\unit}} {\VVAR} {\unitty} {\tunit}   {By \Eunitintro}
                 \typeoftPf {\ctxtrans{\Gamma}} {\tunit} {\unitty}   {By \Munitintro}
             \Hand  \typeoftPf {\ctxtrans{\Gamma}} {\tunit} {\tytrans{\unitty}}   {By def.\ of $\tytrans{-}$}
             \end{llproof}

       \DerivationProofCase{\Ralleointro}
              {
                \chktypee{\Gamma, \eovar \eo}{e}{\VAL}{S_0}
              }
              {\chktypee{\Gamma}{e}{\VAL}{\alleo{\eovar} S_0}}

             \begin{llproof}
                 \chktypeePf{\Gamma, \eovar}{e}{\VAL}{S_0}  {Subd.}
                 \chktypeePf{\Gamma}{e}{\VAL}{[\V/\eovar]S_0}   {By \Lemmaref{lem:subst-eo} (2)}
                 \elabPf {\Gamma} {\er e} {\VAL} {[\V/\eovar]S_0} {M_\V}   {By i.h.}
                 \typeoftPf {\ctxtrans{\Gamma}} {M_\V} {\tytrans{[\V/\eovar]S_0}}   {\ditto}
                 \proofsep
                 \chktypeePf{\Gamma}{e}{\VAL}{[\N/\eovar]S_0}   {By \Lemmaref{lem:subst-eo} (2)}
                 \elabPf {\Gamma} {\er e} {\VAL} {[\N/\eovar]S_0} {M_\N}   {By i.h.}
                 \typeoftPf {\ctxtrans{\Gamma}} {M_\N} {\tytrans{[\N/\eovar]S_0}}   {\ditto}
                 \proofsep
            \Hand     \elabPf {\Gamma} {\er e} {\VAL} {\alleo{\eovar} S_0} {\tpair{M_\V}{M_\N}}   {By \Ealleointro}
                 \typeoftPf {\ctxtrans{\Gamma}} {\tpair{M_\V}{M_\N}} {\tytrans{S_1} * \tytrans{S_2}}   {By \Mprodintro}
             \Hand    \typeoftPf {\ctxtrans{\Gamma}} {\tpair{M_\V}{M_\N}} {\tytrans{\alleo{\eovar} S_0}}   {By def.\ of $\tytrans{-}$}
             \end{llproof}

       \DerivationProofCase{\Ralleoelim}
             {\syntypee{\Gamma}{e}{\VVAR}{\alleo{\eovar} S_0}
               \\
               \Gamma \entails \epsilon \eo}
             {\syntypee{\Gamma}{e}{\VVAR}{[\epsilon / \eovar]S_0}}

             \begin{llproof}
                 \elabPf {\Gamma} {\er e} {\VVAR'} {\alleo{\eovar} S_0} {M_0}   {By i.h.}
             \Hand    \valleqPf{\VVAR'}{\VVAR}   {\ditto}
                 \typeoftPf {\ctxtrans{\Gamma}} {M_0} {\tytrans{[\V/\eovar]S_0} * \tytrans{[\N/\eovar]S_0}}   {\ditto}
             \end{llproof}

             If $\epsilon = \V$ then:

             \begin{llproof}
             \Hand  \elabPf {\Gamma} {\er e} {\VVAR'} {[\V/\eovar]S_0} {\Proj{1}{M_0}}   {By \Ealleoelim}
            \Hand   \typeoftPf {\ctxtrans{\Gamma}}
                                {\Proj{1}{M_0}}
                                {\tytrans{[\V/\eovar]S_0}}
                                {By \Mprodelim{1}}
             \end{llproof}

             Otherwise, $\epsilon \neq \V$.  It is given that $\Gamma$ contains no
             $\eovar$-declarations, and we also have $\Gamma \entails \epsilon \eo$.
             It follows that $\epsilon$ cannot be a variable $\eovar$.
             Therefore $\epsilon = \N$.

             \smallskip
             
             \begin{llproof}
             \Hand  \elabPf {\Gamma} {\er e} {\VVAR'} {[\N/\eovar]S_0} {\Proj{2}{M_0}}   {By \Ealleoelim}
            \Hand   \typeoftPf {\ctxtrans{\Gamma}}
                                {\Proj{2}{M_0}}
                                {\tytrans{[\N/\eovar]S_0}}
                                {By \Mprodelim{2}}
             \end{llproof}

        \DerivationProofCase{\Rsuspintro \text{\normalfont (first conclusion)}}
              {\chktypee{\Gamma}{e}{\VVAR}{S_0}
              }
              {\chktypee{\Gamma}{e}{\VVAR}{\susp{\epsilon} S_0}}

              \begin{llproof}
                  \chktypeePf{\Gamma}{e}{\VVAR}{S_0}   {Subderivation}
                  \elabPf {\Gamma} {\er e} {\VVAR'} {S_0} {M_0}   {By i.h.}
                  \valleqPf{\VVAR'}{\VVAR}   {\ditto}
                  \typeoftPf {\ctxtrans{\Gamma}} {M_0} {\tytrans{S_0}}   {\ditto}
              \end{llproof}

              By similar reasoning as in the \Ralleoelim case, either $\epsilon = \V$
              or $\epsilon = \N$.

              If $\epsilon = \V$:

              \begin{llproof}
                        \eqPf{\tytrans{S_0}}{\tytrans{\susp{\V} S_0}}
                                {By def.\ of $\tytrans{-}$}
                        \LetPf{M}{M_0} {}
              \Hand      \elabPf {\Gamma} {\er e} {\VVAR'} {\susp{\V} S_0} {M}   {By \Esuspintro (first conclusion)}
              \Hand \valleqPf{\VVAR'}{\VVAR}   {Above}
              \Hand      \typeoftPf {\ctxtrans{\Gamma}}
                              {M}
                              {\tytrans{\susp{\V} S_0}}
                              {By above equality}
              \end{llproof}

              If $\epsilon = \N$:

              \begin{llproof}
                        \eqPf{\thunkty \tytrans{S_0}}{\tytrans{\susp{\N} S_0}}
                                {By def.\ of $\tytrans{-}$}
                        \LetPf{M}{\Thunk M_0} {}
              \Hand      \elabPf {\Gamma} {\er e} {\VAL} {\susp{\N} S_0} {\Thunk M_0}   {By \Esuspintro (second conclusion)}
              \Hand   \valleqPf{\VAL}{\VVAR}   {By def.\ of $\valleq$}
                   \typeoftPf {\ctxtrans{\Gamma}}
                              {\Thunk M_0}
                              {\thunkty \tytrans{S_0}}
                              {By \Marrintro}
              \Hand
                   \typeoftPf {\ctxtrans{\Gamma}}
                              {M}
                              {\tytrans{\susp{\N} S_0}}
                              {By above equalities}
              \end{llproof}

        \DerivationProofCase{\Rsuspintro \text{\normalfont (second conclusion)}}
              {\chktypee{\Gamma}{e}{\VVAR'}{S_0}
              }
              {\chktypee{\Gamma}{e}{\VAL}{\susp{\N} S_0}}

              \begin{llproof}
                  \chktypeePf{\Gamma}{e}{\VVAR'}{S_0}   {Subderivation}
                  \elabPf {\Gamma} {\er e} {\VVAR'} {S_0} {M_0}   {By i.h.}
                  \valleqPf{\VVAR'}{\VVAR}   {\ditto}
                  \typeoftPf {\ctxtrans{\Gamma}} {M_0} {\tytrans{S_0}}   {\ditto}
                  \eqPf{\thunkty \tytrans{S_0}}{\tytrans{\susp{\N} S_0}}
                                {By def.\ of $\tytrans{-}$}
                  \LetPf{M}{\Thunk M_0} {}
              \Hand      \elabPf {\Gamma} {\er e} {\VAL} {\susp{\N} S_0} {\Thunk M_0}   {By \Esuspintro (second conclusion)}
              \Hand   \valleqPf{\VAL}{\VVAR}   {By def.\ of $\valleq$}
                   \typeoftPf {\ctxtrans{\Gamma}}
                              {\Thunk M_0}
                              {\thunkty \tytrans{S_0}}
                              {By \Marrintro}
              \Hand
                   \typeoftPf {\ctxtrans{\Gamma}}
                              {M}
                              {\tytrans{\susp{\N} S_0}}
                              {By above equalities}
              \end{llproof}

        \DerivationProofCase{\Rsuspelim{\V}}
             {\syntypee{\Gamma}{e}{\VVAR}{\susp{\V} S}}
             {\syntypee{\Gamma}{e}{\VVAR}{S}}

             \begin{llproof}
               \syntypeePf{\Gamma}{e}{\VVAR}{\susp{\V} S}   {Subderivation}
               \elabPf{\Gamma}{\er e}{\VVAR'}{\susp{\V} S}{M_0}   {By i.h.}
               \Hand    \valleqPf{\VVAR'}{\VVAR}   {\ditto}
               \typeoftPf{\ctxtrans{\Gamma}}{M_0}{\tytrans{\susp{\V} S}}   {\ditto}
                        \eqPf{\tytrans{\susp{\V} S}}{\tytrans{S}}
                                {By def.\ of $\tytrans{-}$}
                        \LetPf{M}{M_0} {}
              \Hand      \elabPf {\Gamma} {\er e}{\VVAR'} {S} {M}   {By \Esuspelim{\V}}
              \Hand      \typeoftPf {\ctxtrans{\Gamma}}
                              {M}
                              {\tytrans{S}}
                              {By above equalities}
             \end{llproof}

        \DerivationProofCase{\Rsuspelim{\epsilon}}
             {\syntypee{\Gamma}{e}{\VVAR'}{\susp{\epsilon} S}}
             {\syntypee{\Gamma}{e}{\NONVAL}{S}}

              By similar reasoning as in the \Ralleoelim case, either $\epsilon = \V$
              or $\epsilon = \N$.

              If $\epsilon = \V$, follow the \Rsuspelim{\V} case above.

              If $\epsilon = \N$:

              \begin{llproof}
                        \syntypeePf{\Gamma}{e}{\VVAR'}{\susp{\N} S}   {Subderivation}
                        \elabPf{\Gamma}{\er e}{\VVAR''}{\susp{\N} S}{M_0}   {By i.h.}
                        \typeoftPf{\ctxtrans{\Gamma}}{M_0}{\tytrans{\susp{\N} S}}   {\ditto}
                        \eqPf{\tytrans{\susp{\N} S}}{\thunkty \tytrans{S}}
                                {By def.\ of $\tytrans{-}$}
                        \LetPf{M}{(\Force M_0)} {}
              \Hand     \elabPf {\Gamma} {\er e} {\NONVAL} {S} {\Force M_0}   {By \Esuspelim{\N}}
              \Hand  \valleqPf{\NONVAL}{\NONVAL}   {By def.\ of $\valleq$}
                        \typeoftPf{\ctxtrans{\Gamma}}{M_0}{\thunkty \tytrans{S}}   {Above ($\tytrans{\susp{\N} S} = \thunkty \tytrans{S}$)}
              \Hand
                   \typeoftPf {\ctxtrans{\Gamma}}
                              {\Force M_0}
                              {\tytrans{S}}
                              {By \Mthunkelim}
              \end{llproof}

          \DerivationProofCase{\Rprodintro}
                 {\chktypee{\Gamma}{e_1}{\VVAR_1}{S_1}
                   \\
                   \chktypee{\Gamma}{e_2}{\VVAR_2}{S_2}
                  }
                  {\chktypee{\Gamma}{\Pair{e_1}{e_2}}{\VVAR_1 \join \VVAR_2}{(S_1 * S_2)}}

              \begin{llproof}
                \elabPf{\Gamma}{\er{e_1}}{\VVAR}{S_1}{M_1}   {By i.h.}
                \valleqPf{\VVAR_1'}{\VVAR_1}   {\ditto}
                \typeoftPf{\ctxtrans{\Gamma}}{M_1}{\tytrans{S_1}}   {\ditto}
                \proofsep
                \elabPf{\Gamma}{\er{e_2}}{\VVAR}{S_2}{M_2}   {By i.h.}
                \valleqPf{\VVAR_2'}{\VVAR_2}   {\ditto}
                \typeoftPf{\ctxtrans{\Gamma}}{M_2}{\tytrans{S_2}}   {\ditto}
                \proofsep
             \Hand
                \elabPf{\Gamma}{\Pair{\er{e_1}}{\er{e_2}}}{\VVAR_1' \join \VVAR_2' }{(S_1 * S_2)}{\tpair{M_1}{M_2}}
                        {By \Eprodintro}
               \Hand    \valleqPf{\VVAR_1' \join \VVAR_2'}{\VVAR_1 \join \VVAR_2}   {$\VVAR_1' \valleq \VVAR_1$ and $\VVAR_2' \valleq \VVAR_2$}
                \typeoftPf{\ctxtrans{\Gamma}}{\tpair{M_1}{M_2}}{\tytrans{S_1} * \tytrans{S_2}}
                        {By \Mprodintro}
             \Hand
                \typeoftPf{\ctxtrans{\Gamma}}{\tpair{M_1}{M_2}}{\tytrans{S_1 * S_2}}
                        {By def.\ of $\tytrans{-}$}
              \end{llproof}
           
           \DerivationProofCase{\Rprodelim{k}}
                 {
                   \syntypee{\Gamma}
                           {e_0}
                           {\VVAR_0}
                           {(S_1 * S_2)}
                 }
                 {
                   \syntypee{\Gamma}
                        {(\Proj{k} e_0)}
                        {\NONVAL}
                        {S_k}
                 }

              \begin{llproof}
                \elabPf{\Gamma}{\er{e_0}}{\VVAR_0'}{(S_1 * S_2)}{M_0}   {By i.h.}
                \typeoftPf{\ctxtrans{\Gamma}}{M_0}{\tytrans{S_1 * S_2}}   {\ditto}
                \typeoftPf{\ctxtrans{\Gamma}}{M_0}{\tytrans{S_1} * \tytrans{S_2}}   {By def.\ of $\ctxtrans{-}$}
                \proofsep
             \Hand
                \elabPf{\Gamma}{(\Proj{k} \er{e_0})}{\NONVAL}{S_k}{(\tproj{k}{M_0})}
                        {By \Eprodelim{k}}
             \Hand
                \typeoftPf{\ctxtrans{\Gamma}}{(\tproj{k}{M_0})}{\tytrans{S_k}}
                        {By \Mprodelim{k}}
              \end{llproof}

       \DerivationProofCase{\Rarrintro}
                 {\chktypee{\Gamma, x:S_1}{e_0}{\VVAR_0}{S_2}
                 }
                 {\chktypee{\Gamma}{(\explam{x} e_0)}{\VAL}{(S_1 \arr S_2)}}

              \begin{llproof}
                  \elabPf {\Gamma, \var{x}{S_1}} {\er{e_0}} {\VVAR_0'} {S_2} {M_0}   {By i.h.}
                  \typeoftPf {\ctxtrans{\Gamma, \var{x}{S_1}}} {M_0} {\tytrans{S_2}}    {\ditto}
                  \proofsep
                  \eqPf{\ctxtrans{\Gamma, \var{x}{S_1}}}
                         {(\ctxtrans{\Gamma}, x : \tytrans{S_1})}
                         {By def.\ of $\ctxtrans{-}$}
                  \proofsep
                  \elabPf {\Gamma, \var{x}{S_1}} {\er{e_0}} {\VVAR_0'} {S_2} {M_0}   {Above}
               \Hand   \elabPf {\Gamma} {(\explam{x} e_0)} {\VAL} {(S_1 \arr S_2)} {(\tlam{x} M_0)}   {By \Earrintro}
                  \proofsep
                  \typeoftPf {\ctxtrans{\Gamma}, \var{x}{\tytrans{S_1}}} {M_0} {\tytrans{S_2}}
                             {Above}
                  \typeoftPf {\ctxtrans{\Gamma}, x : \tytrans{S_1}} {(\tlam{x} M_0)} {\tytrans{S_1} \arr \tytrans{S_2}}
                             {By \Marrintro}
             \Hand     \typeoftPf {\ctxtrans{\Gamma}, x : \tytrans{S_1}} {(\tlam{x} M_0)} {\tytrans{S_1 \arr S_2}}
                             {By def.\ of $\tytrans{-}$}
              \end{llproof}

       \DerivationProofCase{\Rarrelim}
              {\syntypee{\Gamma}{e_1}{\VVAR_1}{(S_1 \arr S)}
               \\
               \chktypee{\Gamma}{e_2}{\VVAR_2}{S_1}
              }
              {\syntypee{\Gamma}{(\expapp{e_1}{e_2})}{\NONVAL}{S}
              }

               \begin{llproof}
                 \elabPf {\Gamma} {\er{e_1}} {\VVAR_1'} {(S' \arr S)} {M_1}   {By i.h.}
                 \typeoftPf {\ctxtrans{\Gamma}} {M_1} {\tytrans{S' \arr S}}   {\ditto}
                 \typeoftPf {\ctxtrans{\Gamma}} {M_1} {\tytrans{S'} \arr \tytrans{S}}   {By def.\ of $\tytrans{-}$}
                 \proofsep
                 \elabPf {\Gamma} {\er{e_2}} {\VVAR_2'} {S'} {M_2}   {By i.h.}
                 \typeoftPf {\ctxtrans{\Gamma}} {M_2} {\tytrans{S'}}   {\ditto}
               \Hand  \elabPf {\Gamma}
                              {\er{\expapp{e_1}{e_2}}}
                              {\NONVAL}
                              {(S' \arr S)}
                              {(M_1 \, M_2)}
                              {By \Earrelim}
               \Hand  \typeoftPf{\ctxtrans{\Gamma}} {(M_1 \, M_2)} {\tytrans{S}}
                       {By \Marrelim}
               \end{llproof}

       \DerivationProofCase{\Rallintro}
              { 
                \chktypee{\Gamma, \alpha \type}{e_0}{\VAL}{S_0}
              }
              {\chktypee{\Gamma}{\tylam{\alpha} e_0}{\VAL}{\All{\alpha} S_0}}

             \begin{llproof}
               \chktypeePf{\Gamma, \alpha \type}{e_0}{\VAL}{S_0}  {Subderivation}
               \elabPf {\Gamma, \alpha \type}
                              {\er{e_0}}
                              {\VAL}
                              {S_0}
                              {M_0}
                              {By i.h.}
               \typeoftPf{\ctxtrans{\Gamma, \alpha \type}}{M_0}{\tytrans{S_0}} {\ditto}
               \typeoftPf{\ctxtrans{\Gamma}, \alpha \type}{M_0}{\tytrans{S_0}} {By def.\ of $\ctxtrans{-}$}
               \proofsep
               \elabPf {\Gamma}
                              {\er{e_0}}
                              {\VAL}
                              {\All{\alpha} S_0}
                              {\ttylam M_0}
                              {By \Eallintro}
             \Hand
               \elabPf {\Gamma}
                              {\er{\tylam{\alpha} e_0}}
                              {\VAL}
                              {\All{\alpha} S_0}
                              {\ttylam M_0}
                              {By def.\ of $\er{-}$}
               \typeoftPf{\ctxtrans{\Gamma}} {\ttylam M_0} {\All{\alpha} \tytrans{S_0}}
                       {By \Mallintro}
               \Hand  \typeoftPf{\ctxtrans{\Gamma}} {\ttylam M_0} {\tytrans{\All{\alpha} S_0}}
                       {By def.\ of subst.}
             \end{llproof}

       \DerivationProofCase{\Rallelim}
             {\syntypee{\Gamma}{e_0}{\VVAR}{\All{\alpha} S_0}
              \\
              \Gamma \entails S' \type}
             {\syntypee{\Gamma}{\tyapp{e_0}{S'}}{\VVAR}{[S' / \alpha]S_0}}

            \begin{llproof}
              \syntypeePf{\Gamma}{e_0}{\VVAR}{\All{\alpha} S_0}{Subderivation}
              \elabPf{\Gamma}{\er{e_0}}{\VVAR'}{\All{\alpha} S_0}{M_0}  {By i.h.}
              \Hand  \valleqPf{\VVAR'}{\VVAR} {\ditto}
              \typeoftPf{\ctxtrans{\Gamma}}{M_0}{\tytrans{\All{\alpha} S_0}}  {\ditto}
              \proofsep
              \ePf{\Gamma}{S' \type}{Subderivation}
              \elabPf{\Gamma}{\er{e_0}}{\VVAR'}{[S'/\alpha]S_0}{\ttyapp{M_0}} {By \Eallelim}
            \Hand  \elabPf{\Gamma}{\er{\tyapp{e_0}{S'}}}{\VVAR'}{[S'/\alpha]S_0}{\ttyapp{M_0}} {By def.\ of $\er{-}$}
              \ePf{\ctxtrans{\Gamma}}{\tytrans{S'}}{By \Lemmaref{lem:tytrans-wf}}
              \proofsep
              \typeoftPf{\ctxtrans{\Gamma}}{M_0}{\All{\alpha} \tytrans{S_0}}  {By def.\ of $\tytrans{-}$}
              \typeoftPf{\ctxtrans{\Gamma}} {\ttyapp{M_0}} {\big[\tytrans{S'}/\alpha\big]\tytrans{S_0}}
                       {By \Mallelim}
              \eqPf{\big[\tytrans{S'}/\alpha\big]\tytrans{S_0}}
                   {\tytrans{[S'/\alpha]S_0}}                       {From def.\ of subst.}
            \Hand \typeoftPf{\ctxtrans{\Gamma}} {\ttyapp{M_0}} {\tytrans{[S'/\alpha]S_0}}
                   {By above equality}
            \end{llproof}

       \DerivationProofCase{\Rsumintro{k}}
             {\chktypee{\Gamma}{e_0}{\VVAR}{S_k}
             }
             {\chktypee{\Gamma}{(\Inj{k} e_0)}{\VVAR}{(S_1 + S_2)}}

             \begin{llproof}
               \chktypeePf{\Gamma}{e_0}{\VVAR}{S_k}  {Subderivation}
               \elabPf{\Gamma}{\er{e_0}}{\VVAR'}{S_k}{M_0}  {By i.h.}
              \Hand  \valleqPf{\VVAR'}{\VVAR} {\ditto}
               \typeoftPf{\ctxtrans{\Gamma}}{M_0}{\tytrans{S_k}}  {\ditto}
             \Hand  \elabPf{\Gamma}{\Inj{k} \er{e_0}}{\VVAR'}{(S_1 + S_2)}{\tinj{k} M_0}  {By \Esumintro{k}}
               \typeoftPf{\ctxtrans{\Gamma}}{\tinj{k} M_0}{\tytrans{S_1} + \tytrans{S_2}}  {By \Msumintro{k}}
             \Hand
               \typeoftPf{\ctxtrans{\Gamma}}{\tinj{k} M_0}{\tytrans{S_1 + S_2}}  {By def.\ of $\tytrans{-}$}
             \end{llproof}
       
       \DerivationProofCase{\Rsumelim}
              {   \syntypee{\Gamma}{e_0}{\VVAR_0}{(S_1 + S_2)}
                  \\
                  \arrayenvbl{
                      \chktypee{\Gamma, x_1 : S_1}{e_1}{\VVAR_1}{S}
                      \\
                      \chktypee{\Gamma, x_2 : S_2}{e_2}{\VVAR_2}{S}
                  }
              }
              {\chktypee{\Gamma}{\Casesum{e_0}{x_1}{e_1}{x_2}{e_2}}{\NONVAL}{S}
              }

              \begin{llproof}
                \syntypeePf{\Gamma}{e_0}{\VVAR_0}{S_1 + S_2}  {Subderivation}
                \elabPf{\Gamma}{\er{e_0}}{\VVAR_0'}{(S_1 + S_2)}{M_0}  {By i.h.}
                \typeoftPf{\ctxtrans{\Gamma}}{M_0}{\tytrans{S_1 + S_2}}  {\ditto}
                \typeoftPf{\ctxtrans{\Gamma}}{M_0}{\tytrans{S_1} + \tytrans{S_2}}  {By def.\ of $\tytrans{-}$}
                \proofsep
                \chktypeePf{\Gamma, x_1 : S_1}{e_1}{\VVAR_1}{S}   {Subderivation}
                \elabPf{\Gamma, x_1 : S_1}{\er{e_1}}{\VVAR_1'}{S}{M_1}   {By i.h.}
                \typeoftPf{\ctxtrans{\Gamma, x_1 : S_1}}{M_1}{\tytrans{S}} {\ditto}
                \typeoftPf{\ctxtrans{\Gamma}, x_1 : \tytrans{S_1}}{M_1}{\tytrans{S}}
                        {By def.\ of $\ctxtrans{-}$}
                \proofsep
                \elabPf{\Gamma, x_2 : S_2}{\er{e_2}}{\VVAR_2'}{S}{M_2} {Similar to above}
                \typeoftPf{\ctxtrans{\Gamma}, x_2 : \tytrans{S_2}}{M_2}{\tytrans{S}} {\ditto}
                \decolumnizePf
              \Hand
                \elabPf{\Gamma}
                       {\er{\Casesum{e_0}{x_1}{e_1}{x_2}{e_2}}}
                       {\NONVAL}
                       {S}
                       {\tcase{M_0}{x_1}{M_1}{x_2}{M_2}}
                       {By \Esumelim}
              \Hand
                \typeoftPf{\ctxtrans{\Gamma}}{\tcase{M_0}{x_1}{M_1}{x_2}{M_2}}{\tytrans{S}}
                    {By \Msumelim}
              \end{llproof}

       \DerivationProofCase{\Rrecintro}
             {\chktypee{\Gamma}{e}{\VVAR}{\big[(\Rec{\alpha}S_0)\big/\alpha\big]S_0}
             }
             {\chktypee{\Gamma}{e}{\VVAR}{\Rec{\alpha} S_0}}

             \begin{llproof}
               \chktypeePf{\Gamma}{e}{\VVAR}{\big[(\Rec{\alpha}S_0)/\alpha\big]S_0}
                       {Subderivation}
               \elabPf {\Gamma}
                              {\er e}
                              {\VVAR'}
                              {\big[(\Rec{\alpha}S_0)/\alpha\big]S_0}
                              {M_0}
                              {By i.h.}
              \Hand  \valleqPf{\VVAR'}{\VVAR} {\ditto}
               \typeoftPf{\ctxtrans{\Gamma}}
                         {M_0}
                         {\tytrans{\big[(\Rec{\alpha}S_0)/\alpha\big]S_0}}
                         {\ditto}
             \Hand
               \elabPf {\Gamma}
                              {\er e}
                              {\VVAR'}
                              {(\Rec{\alpha} S_0)}
                              {(\troll M_0)}
                              {By \Erecintro}
               \eqPf{\tytrans{\big[(\Rec{\alpha}S_0)/\alpha\big]S_0}}
                        {\big[\tytrans{\Rec{\alpha}S_0} / \alpha\big]\,\tytrans{S_0}}
                         {From def.\ of $\tytrans{-}$}
               \typeoftPf{\ctxtrans{\Gamma}}
                         {M_0}
                         {\big[\tytrans{\Rec{\alpha}S_0} / \alpha\big]\,\tytrans{S_0}}
                         {By above equality}
               \typeoftPf{\ctxtrans{\Gamma}}
                         {(\troll M_0)}
                         {\Rec{\alpha} \tytrans{S_0}}
                         {By \Mrecintro}
             \Hand
               \typeoftPf{\ctxtrans{\Gamma}}
                         {(\troll M_0)}
                         {\tytrans{\Rec{\alpha} {S_0}}}
                         {By def.\ of subst.}
             \end{llproof}
       
       \DerivationProofCase{\Rrecelim}
              {   \syntypee{\Gamma}{e}{\VVAR}{\Rec{\alpha} S_0}
              }
              {
                \syntypee{\Gamma}{e}{\NONVAL}{\big[(\Rec{\alpha}S_0)\big/\alpha\big]S_0}
              }

              Broadly similar to the \Rrecintro case.
  \qedhere
  \end{itemize}
\end{proof}

\subsection{Consistency}
\Label{apx:consistency}

\leminversion*
\begin{proof}  By induction on the given derivation.

  For some rules, the proof cases are the same for all parts:

  \begin{itemize}
  \ProofCasesRules{\Esuspintro ($\V$ conclusion), \Esuspelim{\V}}

     The result follows by i.h.  In the \Esuspintro case, we apply the i.h.\ with one less $\susp{\V}$;
     in the \Esuspelim{\V} case, we have one more $\susp{\V}$.  
  \end{itemize}

  For part (0):
     
  \begin{itemize}

     \ProofCaseRule{\Earrintro}  The subderivation gives the result.     
  \end{itemize}

  For part (1):

  \begin{itemize}  

     \ProofCaseRule{\Ealleointro}  The subderivations give the result.

    \end{itemize}

  For part (2):

  \begin{itemize}

    \ProofCaseRule{\Esuspintro ($\N$ conclusion)}
      The subderivation gives the result.
  \end{itemize}

  For part (3):

  \begin{itemize}
    \ProofCaseRule{\Eallintro}  The subderivation gives the result.
  \end{itemize}

  For part (4):

  \begin{itemize}

    \ProofCaseRule{\Esumintro{k}}
        The subderivation gives the result.
  \end{itemize}

  For part (5):

  \begin{itemize}

    \ProofCaseRule{\Erecintro}
        The subderivation gives the result.
  \end{itemize}

  For part (6):

  \begin{itemize}

    \ProofCaseRule{\Eprodintro}
      The subderivations give the result.
  \end{itemize}
  All other cases are impossible: either $M$ has the wrong form, or $S$ has the wrong form.
\end{proof}

\lemsyntacticvalues*
\begin{proof}
  By induction on the given derivation.

  \begin{itemize}
  \ProofCasesRules{\Eunitintro, \Evar, \Earrintro}
      Immediate: the rule requires that $e$ is a syntactic value.

  \ProofCasesRules{\Esuspelim{\N}, \Efixvar, \Efix, \Earrelim,
           \Eprodelim{k}, \Esumelim}

      Impossible: these rules require that $\VAL$ be $\NONVAL$.

  \ProofCaseRule{\Esuspintro ($\N$-conclusion)}
      Impossible: $\Thunk M_0$ is not $\N$-free.

  \ProofCaseRule{\Erecelim}
      Impossible: $\tunroll{M_0}$ is not a value $W$.

  \ProofCasesRules{\Eallintro, \Eallelim,
    \Esuspintro ($\V$-conclusion),
    \Esuspelim{\V}}

     Apply the i.h.\ to the subderivation.

  \ProofCasesRules{\Eprodintro, \Esumintro{k}, \Erecintro}

     Apply the i.h.\ to the subderivation(s).

  \ProofCaseRule{\Ealleointro}
     Apply the i.h.\ to the $\Gamma \entails \elab{e}{\VAL}{[W/\eovar]S_0}{W_1}$
     subderivation.

  \ProofCaseRule{\Ealleoelim}
     Imposible: $W$ must be a projection, but projections are not values.
  \qedhere
  \end{itemize}
\end{proof}

\thmconsistency*
\begin{proof}
  By induction on the derivation of $\cdot \entails \elab e \VVAR S M$.
  
  \begin{itemize}
     \ProofCasesRules{\Evar, \Efixvar}
         Impossible, because the typing context is empty.
     
     \DerivationProofCase{\Efix}
            {\cdot, \var u S \entails \elab {e_0} {\VVAR} {S} {M_0}}
            {\cdot \entails \elab{(\Fix{u} e_0)}{\NONVAL}{S}{(\tfix{u} M_0)}}

            \begin{llproof}
              \elabPf{\cdot, \var u S} {e_0} {\VVAR} {S} {M_0}
                     {Subderivation}
              \stepPf{(\tfix{u} M_0)}{M'}  {Given}
              \eqPf{M'}{\big[(\tfix{u} M_0) \big/ u\big] M_0}  {\byinv{\RedFix}}
            \Hand  \srcstepPf{(\Fix{u} e_0)}{\big[(\Fix{u} e_0) \big/ u\big] e_0}
                    {By \SrcRedFixV and \SrcStepContextV}
              \decolumnizePf
              \elabPf{\cdot}{(\Fix{u} e_0)}{\NONVAL}{S}{(\tfix{u} M_0)}
                     {Given}
              \elabPf{\cdot, \var u S} {e_0} {\VVAR}{S} {M_0}
                     {Subderivation}
            \Hand  \elabPf{\cdot}{\big[(\Fix{u} e_0) \big/ u\big] e_0} {\VVAR} {S} {\big[(\tfix{u} M_0) \big/ u\big] M_0}
                     {By \Lemmaref{lem:elab-expr-subst} (2)}
            (1)\Hand  \FreePf{(holds vacuously)} {$\VVAR = \NONVAL$}
            (2)\Hand  \FreePf{Derivation does not use \SrcStepContextN} {}
            \end{llproof}

     \DerivationProofCase{\Eunitintro}
          { }
          {\cdot \entails \elab{\unit}{\VAL}{\unitty}{\tunit}}

          Impossible, since $M = \tunit$ but $\tunit \step M'$ is not derivable.

     \DerivationProofCase{\Earrintro}
         {\cdot, \var{x}{S_1} \entails
           \elab {e_0} {\VVAR} {S_2} {M_0}
         }
         {
              \cdot
              \entails
              \elab 
                   {(\explam {x} e_0)}
                   {\VAL}
                   {(S_1 \arr S_2)}
                   {\tlam{x} M_0}
         }

         Impossible, since $M = \tlam{x} M_0$ but $(\tlam{x} M_0) \step M'$ is not derivable.
     
     \DerivationProofCase{\Earrelim}
          {
            \arrayenvbl{
                \cdot \entails \elab{e_1}{\VVAR_1}{(S_1 \arr S)}{M_1}
                \\
                \cdot \entails \elab{e_2}{\VVAR_2}{S_1}{M_2}
            }
          }
          {\cdot \entails \elab
                {(\expapp{e_1}{e_2})}
                {\NONVAL}
                {S}
                {(M_1 \, M_2)
                }
          }

          First, note that $\VVAR = \NONVAL$ so ``moreover'' part (1) is vacuously satisfied.

          We have $(M_1 \, M_2) \step M'$.
          By inversion on \StepContext, %
          $M = (M_1 \, M_2) = \E[M_0]$ and $M' = \E[M_0']$.
          From $(M_1 \, M_2) = \E[M_0]$
          and the definition of $\E$, either $\E = \hole$, or
          $\E = (\E_1 \, M_2)$,
          or $\E = (M_1 \, \E_2)$ with $M_1$ a value.

          \begin{itemize}
          \item           If $\E = \hole$, then $M = M_0$ and $M' = M_0'$.
            By inversion on \RedBeta with
            $(M_1 \, M_2) \stepR M'$, we have
            $M_1 = (\tlam{x} Mbody)$
            and $M_2 = W$
            and $M' = [W/x]Mbody$.

            If $M_1 \, M_2$ is \emph{not} $\N$-free, then:

            \begin{llproof}
              \elabPf{\cdot}{e_1}{\VVAR_1}{(S_1 \arr S)}{(\tlam{x}Mbody)}   {Subderivation}
              \eqPf{e_1}{(\Lam{x} ebody)}   {By \Lemmaref{lem:inversion} (0)}
              \elabPf{\cdot, x : S_1}{ebody}{\VVAR''}{S}{Mbody}   {\ditto}
              \proofsep
              \elabPf{\cdot}{e_2}{\VVAR_2}{S_1}{W}   {Subderivation ($M_2 = W$)}
            \Hand
              \elabPf{\cdot}{[e_2/x]ebody}{\VVAR'}{S}{[W/x]Mbody}   {By \Lemmaref{lem:elab-expr-subst} (1)}
            \Hand  \valleqPf{\VVAR'}{\VVAR}   {\ditto}
              \srcredNPf{\expapp{(\Lam{x} ebody)}{e_2}}{[e_2/x]ebody}   {By \SrcRedBetaN}
            \Hand
              \srcstepsPf{\expapp{(\Lam{x} ebody)}{e_2}}{[e_2/x]ebody}   {By \SrcStepContextN}
            \end{llproof}

            \smallskip
            
            If $M_1 \, M_2$ is $\N$-free, then:

            \begin{llproof}
              \elabPf{\cdot}{e_1}{\VVAR}{(S_1 \arr S)}{(\tlam{x}Mbody)}   {Subderivation}
             \eqPf{e_1}{(\Lam{x} ebody)}   {By \Lemmaref{lem:inversion} (0)}
              \elabPf{\cdot, x : S_1}{ebody}{\VVAR''}{S}{Mbody}   {\ditto}
              \proofsep
              \elabPf{\cdot}{e_2}{\VVAR_2}{S_1}{W}   {Subderivation ($M_2 = W$)}
              \elabPf{\cdot}{e_2}{\VAL}{S_1}{W}   {By \Lemmaref{lem:value-mono}}
              \FreePfC{$W$ is $\N$-free} {$M_1 \, W$ is $\N$-free}
               \elabPf{\cdot}{v}{\VAL}{S_1}{W}   {By \Lemmaref{lem:syntactic-value}}
            \Hand
              \elabPf{\cdot}{[v/x]ebody}{\VVAR'}{S}{[W/x]Mbody}   {By \Lemmaref{lem:elab-expr-subst} (1)}
            \Hand  \valleqPf{\VVAR'}{\VVAR}   {\ditto}
              \srcredNPf{\expapp{(\Lam{x} ebody)}{v}}{[v/x]ebody}   {By \SrcRedBetaV}
            \Hand
              \srcstepsPf{\expapp{(\Lam{x} ebody)}{v}}{[v/x]ebody}   {By \SrcStepContextV}
            \end{llproof}

            \medskip
            
          \item If $\E = (\E_1 \, M_2)$, then:

            \begin{llproof}
              \stepPf{M_1 \, M_2}{M'} {Given}
              \stepPf{\underbrace{\E_1[M_R]}_{M_1} \, M_2}{\underbrace{\E_1[M_R']}_{M_1'} \, M_2}  {\byinv{\StepContext}}
              \stepRPf{M_R}{M_R'}  {\byinv{\StepContext}}
              \stepPf{\E_1[M_R]}{\E_1[M_R']}  {By \StepContext}
              \stepPf{M_1}{M_1'}  {By known equalities}
              \decolumnizePf
              \elabPf{\cdot}{e_1}{\VVAR_1}{(S_1 \arr S)}{M_1}   {Subderivation}
              \srcstepsPf{e_1}{e_1'}   {By i.h.}
              \elabPf{\cdot}{e_1'}{\VVAR_1'}{(S_1 \arr S)}{M_1'}  {\ditto}
          \Hand  \srcstepsPf{\expapp{e_1}{e_2}}{\expapp{e_1'}{e_2}}  {By \SrcStepContextV}
          \Hand    \elabPf{\cdot}{\expapp{e_1'}{e_2}}{\NONVAL}{S}{M_1' \, M_2}  {By \Earrelim}
            \end{llproof}

            If $M$ is $\N$-free, then $M_1$ is $\N$-free and the i.h.\ is sufficient for
            ``moreover'' part (2).
            
          \medskip

          \item If $\E = (M_1 \, \E_2)$ where $M_1$ is a value, then
            we have $M_2 \step M_2'$.

            If $M$ is not $\N$-free, then:

            \begin{llproof}
              \elabPf{\cdot}{e_2}{\VVAR_2}{S_1}{M_2}  {Subderivation}
              \srcstepsPf{e_2}{e_2'}   {By i.h.}
              \elabPf{\cdot}{e_2'}{\VVAR_2'}{S_1}{M_2'}  {\ditto}
          \Hand  \srcstepsPf{\expapp{e_1}{e_2}}{\expapp{e_1}{e_2'}}  {By \SrcStepContextN}
          \Hand  \elabPf{\cdot}{\expapp{e_1}{e_2'}}{\NONVAL}{S}{M_1 \, M_2'}  {By \Earrelim}
            \end{llproof}

            If $M$ \emph{is} $\N$-free, then:

            \begin{llproof}
              \elabPf{\cdot}{e_1}{\VVAR_1}{(S_1 \arr S)}{M_1}   {Subderivation}
              \elabPf{\cdot}{e_1}{\VAL}{(S_1 \arr S)}{M_1}   {By \Lemmaref{lem:value-mono}}
              \eqPf{e_1}{v_1}   {By \Lemmaref{lem:syntactic-value}}
              \elabPf{\cdot}{e_2}{\VVAR_2}{S_1}{M_2}  {Subderivation}
              \srcstepsPf{e_2}{e_2'}   {By i.h.}
              \elabPf{\cdot}{e_2'}{\VVAR_2'}{S_1}{M_2'}  {\ditto}
          \Hand  \srcstepsPf{\expapp{v_1}{e_2}}{\expapp{v_1}{e_2'}}  {By \SrcStepContextV}
          \Hand    \elabPf{\cdot}{\expapp{v_1}{e_2'}}{\NONVAL}{(S_1 \arr S)}{M_1' \, M_2}  {By \Earrelim}
            \end{llproof}
          \end{itemize}

    \DerivationProofCase{\Ealleointro}
          {\arrayenvbl{
              \cdot \entails \elab{e}{\VAL}{[\V/\eovar]S_0}{M_1}
              \\
              \cdot \entails \elab{e}{\VAL}{[\N/\eovar]S_0}{M_2}
            }
          }
          {\cdot \entails \elab{e}{\VAL}{(\alleo{\eovar} S_0)}{\tpair{M_1}{M_2}}
          }

          By inversion on $\tpair{M_1}{M_2} \step M'$, either
          $M' = \tpair{M_1'}{M_2}$ and $M_1 \step M_1'$,
          or $M' = \tpair{M_1}{M_2}$ and $M_2 \step M_2'$.
          
          In the first case:
          
          \begin{llproof}
            \elabPf{\cdot}{e}{\VAL}{[\V/\eovar]S_0}{M_1}   {Subderivation}
            \stepPf{M_1}{M_1'}  {Above}
            \elabPf{\cdot}{e}{\VAL}{[\V/\eovar]S_0}{M_1'}   {By i.h.\ ($\VVAR = \VAL$ so $e' = e$)}
            \proofsep
            \elabPf{\cdot}{e}{\VAL}{[\N/\eovar]S_0}{M_2}   {Subderivation}
            \proofsep
          \Hand  \elabPf{\cdot}{e}{\VAL}{(\alleo{\eovar} S_0)}{\tpair{M_1'}{M_2}}
                {By \Ealleointro}
          \Hand  $\Dee \derives~$ \srcstepsPf{e}{e}   {}
          (1)\Hand \eqPf{e'}{e}    {Above}
          (2)\Hand \FreePf{$\Dee$ does not use \SrcStepContextN}    {Zero steps in $e \srcsteps e$}
          \end{llproof}

          The second case is similar.
   
    \DerivationProofCase{\Ealleoelim}
         {\cdot \entails \elab{e}{\VVAR}{(\alleo{\eovar} S_0)}{M_0}}
         {
           \arrayenvbl{
                \cdot
                \entails
                \elab{e}
                     {\VVAR}
                     {[\V/\eovar]S_0}
                     {
                       (\tproj{1} M_0)
                     }
                \\
                \cdot
                \entails
                \elab{e}
                     {\VVAR}
                     {[\N/\eovar]S_0}
                     {(\tproj{2} M_0)}
           }
         }

         First conclusion:

         \begin{llproof}
           \stepPf{(\tproj{1} M_0)}{M'}   {Given}
         \end{llproof}

         Either $M' = \tproj{1} M_0'$ where $M_0 \step M_0'$,
         or $M' = W_1$ and $M_0 = \tpair{W_1}{W_2}$.
         
         \begin{itemize}
         \item In the first case:

           \begin{llproof}
             \elabPf{\cdot}{e}{\VVAR}{(\alleo{\eovar} S_0)}{M_0}   {Subderivation}
             \stepPf{M_0}{M_0'}  {Above}
             \elabPf{\cdot}{e'}{\VVAR}{(\alleo{\eovar} S_0)}{M_0'}   {By i.h.}
           \Hand  $\Dee \derives~$ \srcstepsPf{e}{e'}  {\ditto}
           (1)\Hand    \FreePf{If $\VVAR = \VAL$ then $e = e'$} {\ditto}
           \FreePf {If $M_0$ is $\N$-free then $\Dee$ does not use \SrcStepContextN} {\ditto}
          (2)\Hand    \FreePf {If $(\tproj{1} M_0)$ is $\N$-free then $\Dee$ does not use \SrcStepContextN} {Definition of $\N$-free}
           \Hand  \elabPf{\cdot}{e'}{\VVAR}{[\V/\eovar]S_0}{M_0'}   {By \Ealleoelim}
           \end{llproof}

         \item In the second case:

           \begin{llproof}
             \elabPf{\cdot}{e}{\VVAR}{(\alleo{\eovar} S_0)}{\tpair{W_1}{W_2}}   {Subderivation}
           \Hand  \elabPf{\cdot}{e}{\VVAR}{[\V/\eovar]S_0}{W_1}    {By \Lemmaref{lem:inversion} (1)}
             \stepPf{\tproj{1} \tpair{W_1}{W_2}}{W_1}   {Given}
           (1)\Hand  \LetPf{e'}{e}  {}
           \Hand   \srcstepsPf{\Dee \derives~~e}{e'}  {$e' = e$}
           (2)\Hand  \FreePf{$\Dee$ does not use \SrcStepContextN}  {Zero steps in $e \srcsteps e'$}
           \end{llproof}
         \end{itemize}

         Second conclusion:

         Either $M' = \tproj{2} M_0'$ where $M_0 \step M_0'$,
         or $M' = W_1$ and $M_0 = \tpair{W_1}{W_2}$.

         \begin{itemize}
         \item In the first case: similar to the first subcase of the $[\V/\eovar]$ part above.
         \item In the second case: similar to the second subcase of the $[\V/\eovar]$ part above.
         \end{itemize}

    \DerivationProofCase{\Eallintro}
          {\cdot, \alpha \entails \elab e \VAL S M
          }
          {\cdot \entails \elab{e}{\VAL}{\All{\alpha} S}{\ttylam M}}
          
          This case is impossible, because $(\ttylam M) \step M'$ is not derivable.

    \DerivationProofCase{\Eallelim}
         {\cdot \entails \elab{e}{\VVAR}{\All{\alpha} S_0}{M_0}
          \\
          \cdot \entails S' \type}
         {\cdot \entails \elab{e}{\VVAR}{[S' / \alpha]S_0}{\ttyapp{M_0}}}

         \begin{llproof}
           \stepPf{(\ttyapp{M_0})}{M'}   {Given}
           \eqPf{M_0}{(\ttylam M')}  {By inversion}
           \elabPf{\cdot}{e}{\VVAR}{\All{\alpha} S_0}{M_0}   {Subderivation}
           \elabPf{\cdot}{e}{\VVAR}{\All{\alpha} S_0}{(\ttylam M')}   {By above equality}
           \elabPf{\cdot, \alpha \type}{e}{\VVAR}{S_0}{M'}   {By \Lemmaref{lem:inversion} (3)}
          \Hand
           \elabPf{\cdot}{e}{\VVAR}{[S'/\alpha]S_0}{M'}   {By \Lemmaref{lem:elab-type-subst}}
         \Hand  \srcstepsPf{e}{e}  {Zero steps}
         \end{llproof}

         ``Moreover'' parts (1) and (2) are immediately satisfied, because $e' = e$.

    \DerivationProofCase{\Esuspintro}
          {\cdot \entails \elab{e}{\VVAR}{S_0}{M_0}
          }
          {\arrayenvbl{
            \cdot \entails \elab{e}{\VVAR}{\susp{\V} S_0}{M_0}
            \\
            \cdot \entails \elab{e}{\VAL}{\susp{\N} S_0}{\Thunk M_0}
          }}

        The second conclusion is not possible, because $(\Thunk M_0) \step M'$
        is not derivable.

        For the first conclusion:  We have $M_0 = M$.

        \begin{llproof}
          \elabPf{\cdot}{e}{\VVAR}{S_0}{M}   {Subderivation}
        \Hand  $\Dee \derives~$ \srcstepsPf{e}{e'}  {By i.h.}
          \elabPf{\cdot}{e'}{\VVAR'}{S_0}{M'}  {\ditto}
        \Hand \valleqPf{\VVAR'}{\VVAR}  {\ditto}
        (1)\Hand    \FreePf{If $\VVAR = \VAL$ then $e = e'$} {\ditto}
        (2)\Hand   \FreePf {If $M$ is $\N$-free then $\Dee$ does not use \SrcStepContextN} {\ditto}
        \Hand  \elabPf{\cdot}{e'}{\VVAR'}{\susp{\V} S_0}{M'}  {By \Esuspintro}
        \end{llproof}

    \DerivationProofCase{\Esuspelim{\V}}
         {\cdot \entails \elab{e}{\VVAR}{\susp{\V} S}{M}}
         {\cdot \entails \elab{e}{\VVAR}{S}{M}}

         By i.h.\ and \Esuspelim{\V}.

    \DerivationProofCase{\Esuspelim{\N}}
         {\cdot \entails \elab{e}{\VVAR_0}{\susp{\N} S}{M_0}}
         {\cdot \entails \elab{e}{\NONVAL}{S}{(\Force M_0)}}

         We have $(\Force M_0) \step M'$.  If $M_0 \step M_0'$,
         use the i.h.\ and then apply \Esuspelim{\N}.
         Otherwise, $M_0 = \Thunk M'$.

         \begin{llproof}
           \elabPf{\cdot}{e}{\VVAR_0}{\susp{\N} S}{\Thunk M'}   {Subderivation}
         \Hand  \elabPf{\cdot}{e}{\VVAR_0'}{S}{M'}   {By \Lemmaref{lem:inversion} (2)}
         \Hand  \valleqPf{\VVAR_0'}{\NONVAL}  {By def.\ of $\valleq$}
           \decolumnizePf
         \Hand  \srcstepsPf{e}{e}   {Zero steps}
            (1)\Hand  \FreePf{(holds vacuously)} {$\VVAR = \NONVAL$}
            (2)\Hand  \FreePf{Derivation does not use \SrcStepContextN} {Zero steps}
         \end{llproof}

    \DerivationProofCase{\Eprodintro}
             {
                    \cdot \entails \elab{e_1}{\VVAR}{S_1}{M_1}
                    \\
                    \cdot \entails \elab{e_2}{\VVAR}{S_2}{M_2}
             }
             {
                \cdot \entails \elab{\Pair{e_1}{e_2}}{\VVAR}{(S_1 * S_2)}{\tpair{M_1}{M_2}}
             }

             Apply the i.h.\ to the appropriate subderivation, then apply \Eprodintro
             and \SrcStepContextV.

             ``Moreover'' part (1): \\
             If $\VVAR = \VAL$,
             the i.h.\ shows that $e_1' = e_1$ (or $e_2' = e_2$ if $M_2 \step M_2'$);
             thus, $\Pair{e_1'}{e_2} = \Pair{e_1}{e_2}$ (or $\Pair{e_1}{e_2'} = \Pair{e_1}{e_2}$).

             ``Moreover'' part (2): \\
             If $\tpair{M_1}{M_2}$ is $\N$-free, then $M_1$ and $M_2$ are $\N$-free,
             and the i.h.\ shows that $\Dee_0 \derives e_k \srcsteps e_k'$ does not use
             \SrcStepContextN.  Therefore $\Pair{e_1}{e_2} \srcsteps \dots$ does not use 
             \SrcStepContextN.

       \DerivationProofCase{\Eprodelim{k}}
              {   
                \cdot \entails \elab{e_0}{\VVAR_0}{(S_1 * S_2)}{M_0}
              }
              {\cdot
                \entails
                \elab{(\Proj{k} e_0)}{\NONVAL}{S_k}{(\tproj{k} M_0)}
              }

             We have $(\tproj{k} M_0) \step M'$.

             If $M_0 \step M_0'$
             then use the i.h.\ and apply \Eprodelim{k}.

             Otherwise, $M_0 = \tpair{W_1}{W_2}$ and $M' = W_k$.

             \begin{itemize}
               \item 
                  If $M$ is \emph{not} $\N$-free, we can use \SrcRedProjN:

                     \begin{llproof}
                     \Hand  \elabPf{\cdot}{e_k}{\VVAR_k}{S_k}{W_k}   {By \Lemmaref{lem:inversion} (6)}
                       \eqPf{e_0}{\Pair{e_1}{e_2}}  {\ditto}
                       \proofsep
                     \Hand  \srcstepPf{\Proj{k} \Pair{e_1}{e_2}}{e_k}   {By \SrcRedProjN}
                     \end{llproof}

                     ``Moreover'' part (2): $M$ is not $\N$-free.

               \item
                 If $M$ is $\N$-free, we have the obligation not to use \SrcRedProjN.

                 \begin{llproof}
                   \elabPf{\cdot}{e_0}{\VVAR_0}{(S_1 * S_2)}{\tpair{W_1}{W_2}}   {Subderivation}
                   \elabPf{\cdot}{e_0}{\VAL}{(S_1 * S_2)}{\tpair{W_1}{W_2}}   {By \Lemmaref{lem:value-mono}}
                   \elabPf{\cdot}{v}{\VAL}{(S_1 * S_2)}{\tpair{W_1}{W_2}}   {By \Lemmaref{lem:syntactic-value}}
                   \elabPf{\cdot}{\Pair{v_1}{v_2}}{\VAL}{(S_1 * S_2)}{\tpair{W_1}{W_2}}   {By \Lemmaref{lem:inversion} (6)}
                 \Hand  \elabPf{\cdot}{v_k}{\VAL}{S_k}{W_k}   {\ditto}
                 \decolumnizePf
                 \Hand  \srcstepsPf{\Proj{k} \Pair{v_1}{v_2}}{v_k}   {By \SrcRedProjV and \SrcStepContextV}
                 \end{llproof}

                 ``Moreover'' part (2): we did not use \SrcStepContextN.
             \end{itemize}

          ``Moreover'' part (1): $\VVAR = \NONVAL$.

       \DerivationProofCase{\Esumintro{k}}
             {\cdot \entails \elab{e_0}{\VVAR}{S_k}{M_0}
             }
             {\cdot \entails \elab{(\Inj{k} e_0)}{\VVAR}{(S_1 + S_2)}{(\tinj{k} M_0)}}

             \begin{llproof}
               \stepPf{(\tinj{k} M_0)}{M'}   {Given}
               \eqPf{M'}{(\tinj{k} M_0') ~\AND~ M_0 \step M_0'}   {By inversion}
               \elabPf{\cdot}{e_0}{\VVAR}{S_k}{M_0}   {Subderivation}
               \elabPf{\cdot}{e_0'}{\VVAR'}{S_k}{M_0'}   {By i.h.}
            \Hand \valleqPf{\VVAR'}{\VVAR}  {\ditto}
               \srcstepsPf{e_0}{e_0'}  {\ditto}
             \Hand  \srcstepsPf{(\Inj{k} e_0)}{(\Inj{k} e_0')}  {}
             \Hand  \elabPf{\cdot}{(\Inj{k} e_0')}{\VVAR'}{(S_1 + S_2)}{(\tinj{k} M_0')}   {By \Esumintro{k}}
             \end{llproof}

             ``Moreover'' part (1) follows from the i.h.

             ``Moreover'' part (2) follows from the i.h.:  If $\tinj{k} M_0$ is $\N$-free,
             then $M_0$ is $\N$-free; if $e_0 \srcsteps e_0'$ does not use \SrcStepContextN,
             we can derive $(\Inj{k} e_0) \srcsteps (\Inj{k} e_0')$ without \SrcStepContextN.

       \DerivationProofCase{\Esumelim}
              {   
                  \cdot \entails \elab{e_0}{\VVAR_0}{(S_1 + S_2)}{M_0} \\
                  \arrayenvbl{
                      \cdot, x_1 : S_1 \entails \elab{e_1}{\VVAR_1}{S}{M_1}
                      \\
                      \cdot, x_2 : S_2 \entails \elab{e_2}{\VVAR_2}{S}{M_2}
                  }
              }
              {\cdot
                \entails
                \arrayenvl{
                    {\Casesum{e_0}{x_1}{e_1}{x_2}{e_2}}
                    \xetasym{\NONVAL}
                    {S}
                    ~\elabsymbol~
                    {\tcase{M_0}{x_1}{M_1}{x_2}{M_2}}
                }
              }

              First note that ``Moreover'' part (1) is vacuously satisfied, since $\VVAR = \NONVAL$.

             We have $\tcase{M_0}{x_1}{M_1}{x_2}{M_2} \step M'$.
             Either (1) $M_0 \step M_0'$ and $M' = \tcase{M_0'}{x_1}{M_1}{x_2}{M_2}$
             or (2) $M_0 = (\tinj{k} W)$ and $M' = [W/x_k]M_k$.
             
             For (1), apply the i.h.\ to $\cdot \entails \elab{e_0}{\VVAR}{(S_1 + S_2)}{M_0}$
             and apply \Esumelim.  ``Moreover'' part (2) follows from the i.h.

             For (2) if $M$ is \emph{not} $\N$-free, we can use
             \SrcStepContextN:

             \begin{llproof}
               \elabPf{\cdot}{e_0}{\VVAR_0}{(S_1 + S_2)}{(\tinj{k} W)}  {Subderivation}
               \eqPf{e_0}{\Inj{k} e_0'}   {By \Lemmaref{lem:inversion} (4)}
               \elabPf{\cdot}{e_0'}{\VVAR_0'}{S_k}{W}  {\ditto}
               \elabPf{\cdot, x_k : S_k}{e_k}{\VVAR_k}{S}{M_k}  {Subderivation}
            \Hand   \elabPf{\cdot}{[e_0'/x_k]e_k}{\VVAR_k'}{S}{[W/x_k]M_k}  {By \Lemmaref{lem:elab-expr-subst} (1)}
               \decolumnizePf
               \eqPf{e_0}{\Inj{k} e_0'}   {Above}
               \srcredNPf{\Casesum{\Inj{k}e_0'}{x_1}{e_1}{x_2}{e_2}}
                          {[e_0'/x_k]e_k}   {By \SrcRedCaseN}
            \Hand   \srcstepsPf{\Casesum{e_0}{x_1}{e_1}{x_2}{e_2}}
                          {[e_0'/x_k]e_k}   {By \SrcStepContextN}
             \end{llproof}

             For (2) if $M$ is $\N$-free, we can show
             $\cdot \entails \elab{[e_0'/x_k]e_k}{\VVAR_k'}{S}{[W/x_k]M_k}$
             as in the case when $M$ is not $\N$-free, but
             we have an obligation (``Moreover'' part (2)) not to use \SrcRedCaseN.

                 \begin{llproof}
                   \elabPf{\cdot}{e_0}{\VVAR_0}{(S_1 + S_2)}{\tinj{k} W}   {Subderivation}
                   \elabPf{\cdot}{e_0}{\VAL}{(S_1 + S_2)}{\tinj{k} W}   {By \Lemmaref{lem:value-mono}}
                   \eqPf{e_0}{v}   {By \Lemmaref{lem:syntactic-value}}
                   \elabPf{\cdot}{v}{\VAL}{(S_1 + S_2)}{\tinj{k} W}   {By above equality}
                   \eqPf{v_0}{\Inj{k} v_0'}   {By \Lemmaref{lem:inversion} (4)}
                 \Hand  \srcstepPf{e}{[v_0'/x]e_k}   {By \SrcRedCaseV and \SrcStepContextV}
                 \end{llproof}

       \DerivationProofCase{\Erecintro}
             {\cdot \entails \elab{e}{\VVAR}{\big[(\Rec{\alpha}S_0)/\alpha\big]S_0}{M_0}
             }
             {\cdot \entails \elab{e}{\VVAR}{\Rec{\alpha} S_0}{(\troll M_0)}}

             By inversion, $M_0 \step M_0'$ and $M' = (\troll M_0')$.

             \begin{llproof}
               \elabPf{\cdot}{e}{\VVAR}{\big[(\Rec{\alpha}S_0)/\alpha\big]S_0}{M_0}  {Subderivation}
               \elabPf{\cdot}{e'}{\VVAR}{\big[(\Rec{\alpha}S_0)/\alpha\big]S_0}{M_0'}  {By i.h.}
             \Hand  \srcstepsPf{e}{e'} {\ditto}
             \Hand  \elabPf{\cdot}{e'}{\VVAR}{\Rec{\alpha} S_0}{(\troll M_0')}   {By \Erecintro}
             \end{llproof}

             ``Moreover'' parts (1) and (2) follow from the i.h.

       \DerivationProofCase{\Erecelim}
              {   \cdot \entails \elab{e}{\VVAR_0}{\Rec{\alpha} S_0}{M_0}
              }
              {
                \cdot \entails \elab{e}{\NONVAL}{\big[(\Rec{\alpha}S_0)/\alpha\big]S_0}{(\tunroll M_0)}
              }

              We have $(\tunroll M_0) \step M'$.
              Either (1) $M' = (\tunroll M_0')$ and $M_0 \step M_0'$
              or
              (2) $M_0 = (\troll W)$ and $M' = W$.

              If (1), similar to the \Erecintro case.  %

              If (2):
              
              \begin{llproof}
                \elabPf{\cdot}{e}{\VVAR_0}{\Rec{\alpha} S_0}{(\troll W)}   {Subderivation}
             \Hand   \elabPf{\cdot}{e'}{\VVAR'}{\big[(\Rec{\alpha}S_0)/\alpha\big]S_0}{W}   {By \Lemmaref{lem:inversion} (5)}
              \Hand
                \srcstepsPf{e}{e'}   {\ditto}
              \end{llproof}

              ``Moreover'' part (1) is vacuously satisfied; part (2) follows from the i.h.
  \qedhere
  \end{itemize}
\end{proof}

\thmconsistencystar*
\begin{proof}
  By induction on the derivation of $M \steps W$.

  If $M = W$ then let $e'$ be $e$.  By \Lemmaref{lem:value-mono},
  $\cdot \entails \elab {e'} {\VAL} {S} {W}$.  The source expression $e$ steps
  to itself in zero steps, so $e \srcsteps e$, \ie $e \srcsteps e'$.
  We did not use \SrcStepContextN.
  
  Otherwise, we have $M \step M'$ and $M' \steps W$ for some $M'$.
  By \Theoremref{thm:consistency}, $\cdot \entails \elab {e_1} {\VVAR} {S} {M'}$,
  where $e \srcsteps e_1$; also, if $M$ is $\N$-free, then \Theoremref{thm:consistency}
  showed that we did not use \SrcStepContextN.
  If $M$ is $\N$-free, then by \Lemmaref{lem:step-n-freeness},
  $M'$ is $\N$-free.
  By i.h., there exists $e'$ such that
  $e_1 \srcsteps e'$ and $\cdot \entails \elab {e'} {\VAL} {S} {W}$.
  It follows that $e \srcsteps e'$.
\end{proof}

If a source type, economical typing judgment, or target term is not $\N$-free,
we say it is \emph{$\N$-tainted}.

\econbisubformula*
\begin{proof}
  By induction on the given derivation.
  
  \begin{itemize}
  \ProofCaseRule{\Rallelim}
     If $S'$ is not $\N$-free, then $e = \tyapp{e_0}{S'}$ is not $\N$-free.
     Otherwise, we have that $S = [S'/\alpha]S_0$ is not $\N$-free;
     since $S'$ is $\N$-free, $S_0$ must not be $\N$-free, which lets us
     apply the i.h., giving the resut.
  
  \ProofCasesRules{\Ralleoelim, \Rsuspelim{\V}, \Rsuspelim{\epsilon}}
     The i.h.\ gives the result.

  \ProofCasesRules{\Rvar, \Rfixvar}  The type $S$ appears in $\Gamma$,
      so $\Gamma$ is $\N$-tainted.

  \ProofCaseRule{\Ranno}  The type $S$ appears in $e = \Anno{e_0}{S}$,
      so $e$ is $\N$-tainted.

  \ProofCaseRule{\Rarrelim}  If $S$ is $\N$-tainted then $S_1 \arr S$
      is $\N$-tainted, and the result follows by i.h.

  \ProofCasesRules{\Rprodelim, \Rrecelim} Similar to the \Rarrelim case.
  \qedhere
  \end{itemize}
\end{proof}

\econpreservesnfreeness*
\begin{proof}
  By induction on the given derivation.
  We can simply follow the proof of \Theoremref{thm:econ}, observing that if the given
  impartial judgment is $\N$-free, the resulting economical judgment is $\N$-free.
  For example, in the \Iarrintro case, we have $\tau = (\tau_1 \garr{\epsilon} \tau_2)$.
  Since we know that $\tau$ is $\N$-free, $\epsilon = \V$, so the translation
  of $\tau$ is $\big(\susp{\V} \econtrans{\tau_1}\big) \arr \econtrans{\tau_2}$,
  which is $\N$-free.  Note that \Definitionref{def:n-free-impartial} (2)(b) bars
  $\xsyn x \NONVAL \tau$ declarations---which would result in
  $x : \susp{\N} \cdots$---from $\gamma$.
\end{proof}

\elabpreservesnfreeness*
\begin{proof}
  By induction on the given derivation.
  
  \begin{itemize}
    \ProofCaseRule{\Runitintro}  Apply \Eunitintro.

    \ProofCaseRule{\Ralleointro}
       Impossible: $S = \alleo{\eovar}S_0$, which is not $\N$-free (\Definitionref{def:n-free-econ} (1)(ii)).
       
    \ProofCaseRule{\Ralleoelim}

       We have $\syntypee{\Gamma}{e}{\VVAR}{\alleo{\eovar} S_0}$,
       where $S = [S'/\eovar]S_0$.

       By \Definitionref{def:n-free-econ} (1)(ii), the type
       $\alleo{\eovar} S_0$ is $\N$-tainted.
       So, by \Lemmaref{lem:econ-bi-subformula},
       at least one of $\Gamma$ and $e$ is $\N$-tainted.
       But it was given that the judgment $\syntypee{\Gamma}{e}{\VVAR}{S}$
       is $\N$-free, which means that $\Gamma$ and $e$ are $\N$-free.
       We have a contradiction: this case is impossible.

    \ProofCaseRule{\Rsuspintro (first conclusion)}
        Use the i.h.\ and apply rule \Esuspintro (first conclusion).

    \ProofCaseRule{\Rsuspintro (second conclusion)}
       Impossible: $S = \susp{\N}S_0$, which is not $\N$-free.

    \ProofCaseRule{\Rsuspelim{\V}}
       Use the i.h.\ and apply rule \Esuspelim{\V}.

     \ProofCaseRule{\Rsuspelim{\epsilon}}

       We have $\syntypee{\Gamma}{e}{\VVAR}{\susp{\epsilon} S}$.
       
       If $\epsilon = \V$ then use the i.h., apply rule \Esuspelim{\V}. %
       
       Otherwise, $\susp{\epsilon} S$ is not $\N$-free.
       As in the \Ralleoelim case, we can use \Lemmaref{lem:econ-bi-subformula}
       to reach a contradiction.
  
     \ProofCasesRules{\Rvar, \Rfixvar, \Rfix,
       \Rallintro, \Rallelim,
       \Rarrintro, \Rarrelim,
       \Rprodelim,
        \Rsumintro{k}, \Rsumelim,
        \Rrecintro, \Rrecelim}

        Use the i.h.\ on all subderivations (if any) and apply the corresponding elaboration rule,
        \eg in the \Rfix case, apply \Efix.

     \ProofCasesRules{\Rsub, \Ranno}  Use the i.h.

     \ProofCaseRule{\Rprodintro}  Use the i.h.\ on each subderivation,
     and apply \Eprodintro.
  \qedhere
  \end{itemize}
\end{proof}

\fi
\end{document}